\documentclass[a4paper, 11pt, headings = small, abstract]{scrartcl}

\usepackage[english]{babel}
\usepackage[utf8]{inputenc}
\usepackage{geometry}

\usepackage[semibold]{libertine}
\usepackage[upint,amsthm]{libertinust1math}

\usepackage{amsmath}
\usepackage{amsfonts}
\usepackage{amssymb}
\usepackage{amsthm}

\usepackage{paralist}
\usepackage{comment}
\usepackage{graphicx}     % Pour inclure des graphiques
\usepackage{caption}      % Pour gérer les légendes
\usepackage{subcaption}   % Pour les sous-figures
\usepackage{multirow}
\usepackage{float} % For precise placement with [H]

\usepackage[authoryear]{natbib}
\usepackage[dvipsnames]{xcolor}
\usepackage[colorlinks=true,linkcolor=blue,citecolor=blue,pdfborder={0 0 0}]{hyperref}
\usepackage{orcidlink}

\usepackage{bm}
\usepackage{graphicx}

\usepackage{algorithm}
\usepackage{algpseudocode}

\usepackage{tikz}
\usetikzlibrary{positioning}
\usetikzlibrary{backgrounds}
\usepackage{pgfplots}
\usepackage{pgfplotstable} 
\usepgfplotslibrary{fillbetween}
\pgfplotsset{compat=1.18}

\newtheorem{theorem}{Theorem}[section]
\newtheorem{proposition}[theorem]{Proposition}
\newtheorem{lemma}[theorem]{Lemma}

\theoremstyle{definition}
\newtheorem{definition}[theorem]{Definition}

\newtheorem{remark}[theorem]{Remark}

\allowdisplaybreaks[4]

\numberwithin{equation}{section}

\newcommand{\R}{\mathbb{R}}

\newcommand{\N}{\mathbb{N}}
\newcommand{\eps}{\varepsilon}

\newcommand{\Eb}{\mathbb {E}}

\newcommand{\Mb}{\mathbb {M}}

\newcommand{\Ic}{\mathcal{I}}

\newcommand{\weak}{\rightsquigarrow}

\newcommand{\Prob}{\mathbb{P}}   
\newcommand{\Exp}{\operatorname{E}}

\newcommand{\diag}{\operatorname{diag}}
\newcommand{\argmin}{\operatornamewithlimits{\arg\min}}
\newcommand{\argmax}{\operatornamewithlimits{\arg\max}}
\newcommand{\diff}{\mathrm{d}}

\def\@tvsp{\mathchoice{{}\mkern-4.5mu}{{}\mkern-4.5mu}{{}\mkern-2.5mu}{}}
\def\ltrivert{\left|\@tvsp\left|\@tvsp\left|}
\def\rtrivert{\right|\@tvsp\right|\@tvsp\right|}
\makeatother

\newcommand{\loadn}{A}
\newcommand{\load}{B}
\newcommand{\setloadn}{\mathcal A_s(d,K)}
\newcommand{\setloadnpure}{\mathcal A_{\mathrm{Pure}}(d,K)}
\newcommand{\setloadpurerv}{\mathcal B_{\mathrm{Pure}}(d,K)}

\newcommand{\setstdf}[1]{\mathcal L(#1)}
\newcommand{\simplexpos}[1]{\mathbb S_{+}^{#1}}
\newcommand{\simplex}[1]{\mathbb S^{#1}}
\newcommand{\setspectraldepposK}[1]{\mathcal S_{\mathrm d}(\simplexpos{K-1}, {#1})}
\newcommand{\setspectraldepposd}[1]{\mathcal S_{\mathrm d}(\simplexpos{d-1}, {#1})}
\newcommand{\setspectrald}[1]{\mathcal S(\simplex{d-1}, {#1})}
\newcommand{\setspectralK}[1]{\mathcal S(\simplex{K-1}, {#1})}
\newcommand{\Phifactor}{\phi}
\newcommand{\Psifactor}{\psi}
\newcommand{\diagdom}[1]{\mathcal D_{\mathrm{DiagDom}}({#1})}
\newcommand{\paraspacetaildep}{\Theta_L}
\newcommand{\paraspacespectral}{\Theta_\Phi}

\newcommand{\TPDM}{\Sigma}
\newcommand{\shittyTPDM}{\Xi}
\newcommand{\indicator}[1]{\bm 1_{{#1}}}

\newcommand{\be}{{\bm e}}
\newcommand{\bE}{{\bm E}}
\newcommand{\bW}{{\bm W}}

\newcommand{\bX}{{\bm X}}
\newcommand{\bx}{{\bm x}}
\newcommand{\bY}{{\bm Y}}
\newcommand{\by}{{\bm y}}
\newcommand{\bZ}{{\bm Z}}
\newcommand{\bz}{{\bm z}}

\newcommand{\shear}{\alpha_{\operatorname{sh}}}

\begin{document}

\title{\fontsize{16}{19} Dimension Reduction in Multivariate Extremes via Latent Linear Factor Models}

\author{
Alexis Boulin\thanks{Ruhr-Universität Bochum, Fakultät für Mathematik. Email: \href{mailto:alexis.boulin@rub.de}{alexis.boulin@rub.de}}
%~\orcidlink{XXX}
\and
Axel B\"ucher\thanks{Ruhr-Universität Bochum, Fakultät für Mathematik. Email: \href{mailto:axel.buecher@rub.de}{axel.buecher@rub.de}}
~\orcidlink{0000-0002-1947-1617}
}

\date{\today}

\maketitle

\begin{abstract}
We propose a new and interpretable class of high-dimensional tail dependence models based on latent linear factor structures. Specifically, extremal dependence of an observable vector is assumed to be driven by a lower-dimensional latent $K$-factor model, where $K \ll d$, thereby inducing an explicit form of dimension reduction. Geometrically, this is reflected in the support of the associated spectral dependence measure, whose intrinsic dimension is at most $K-1$. The loading structure may additionally exhibit sparsity, meaning that each component is influenced by only a small number of latent factors, which further enhances interpretability and scalability. 
Under mild structural assumptions, we establish identifiability of the model parameters and provide a constructive recovery procedure based on a margin-free tail pairwise dependence matrix, which also yields practical rank-based estimation methods. The framework combines naturally with marginal tail models and is particularly well suited to high-dimensional settings. We illustrate its applicability in a spatial wind energy application, where the latent factor structure enables tractable estimation of the risk that a large proportion of turbines simultaneously fall below their cut-in wind speed thresholds.
\end{abstract}


\section{Introduction}

Extreme value statistics investigates the probabilistic behavior of rare events, that is, realizations of a random sample that occur at unusually high or low levels \citep{BeiGoeSegTeu04}. A central concept in this field is \emph{tail dependence}, which describes the strength and structure of dependence between components of a random vector when one or more coordinates take extreme values. In contrast to classical correlation-based measures, tail dependence captures co-movements in the far tails of a joint distribution and is therefore essential for risk assessment in areas such as finance, environmental science, and engineering.

After standardizing the margins to a common reference distribution, tail dependence admits several equivalent representations. Important characterizations include the Pickands dependence function \citep{Pic81}, the stable tail dependence function \citep{Hua92}, the tail copula \citep{SchSta06}, as well as formulations based on multivariate regular variation, exponent measures, and spectral measures \citep{Res87}. These objects are in one-to-one correspondence and offer complementary perspectives on extremal dependence; see Section~\ref{sec:mathematical-preliminaries} for a concise overview.

The present paper focuses on the construction of new and interpretable tail dependence models that are particularly suitable for high-dimensional settings. It is convenient to introduce these models at the level of the \emph{stable tail dependence function} (STDF) $L:[0,\infty)^d \to [0, \infty)$, formally defined by
\begin{align} 
\label{eq:definition-stdf}
        L(\bx) &= \lim_{n \to \infty} n \mathbb{P}\left\{ F_1(X_1) > 1-\frac{x_1}{n} \textrm{ or } \dots \textrm{ or } F_d(X_d) > 1-\frac{x_d}{n} \right\}, \qquad \bx \in [0, \infty)^d
\end{align}
where $F_1, \dots, F_d$ are the marginal cdfs of $F$. 
At a conceptual level, $L$ plays a role for tail dependence modeling analogous to that of the copula $C$ in general dependence modeling. Parametric modeling of $L$ can hence draw inspiration from parametric modeling of $C$, an approach that has, for instance, been successfully carried out for vine copula models recently, with their extremal analogues called X-vines \citep{KirilioukSegers2025}.  Our paper follows similar lines, and is based on transferring the concept of factor copulas \citep{OhPat17} to the realm of tail dependence; see also \cite[Section 6]{EinmahlKrajinaSegers2012} and \cite{boulin2026structured} for some constructions of that form.

The concept of factor copulas is simple and intuitive: the copula $C$ of some observable $\bX$ is assumed to be the same copula that arises in a classical linear latent factor model. The parameter of the copula  hence consists of the factor loading matrix in the latent linear model, as well as on further, possibly nonparameteric assumptions on the latent factors. We follow this construction:
suppose that some random vector $\bY' \in [0,\infty)^d$ is generated by the following simple linear factor model:
\begin{align} \label{eq:genuine-linear-factor-dependence}
\bY' := \loadn \bZ \quad \text{ with } \quad \bZ := K R \bm \Lambda,
\end{align}
where $\loadn \in [0,1]^{d \times K}$ with $K \in [d]$ is a factor loading matrix (with row sums 1 and full column rank),
$R$ is a standard Pareto random variable, and $\bm \Lambda$ is a random vector on the unit simplex $\simplexpos{K-1}$ in $[0,\infty)^K$ whose margins $\Lambda_j$ have expectation  $1/K$ and which is independent of $R$. 
One should think of $K$ being much smaller than $d$, such that the dependence within $\bY'$ is essentially induced by some $K$-dimensional dependence only; some form of dimension reduction.
A standard calculation (see Proposition~\ref{prop:genuine-linear-factor-dependence} below) then shows that the STDF of $\bY'$ is given by
\begin{align} \label{eq:stdf-factor-introduction}
    L_{K, \loadn, \Psifactor}(\bx)  
    =
    K  \int_{\simplexpos{K-1}} \bigvee_{j \in [d]} \Big( x_j\sum_{a\in [K]} \loadn_{ja} z_a\Big)  \, \Psifactor(\diff\bz), \quad \bx\in [0,\infty)^d,
\end{align}
where $\Psifactor$ is the distribution of $\bm \Lambda$. Our modeling assumption for the observable $\bX$ is now as follows: we assume that, for some unknown triple $(K, \loadn, \Psifactor)$, the STDF $L$ of $\bX$ satisfies $L=L_{K, \loadn, \Psifactor}$, an assumption which we call  \textit{latent linear $K$-factor tail dependence}; see Definition~\ref{def:latent-linear-factor-tail-dependence} below for details. The model allows for similar interpretation as a factor copula model in the sense that the tail dependence arises from a lower dimensional latent vector $\bZ$, with the influence of each factor on the coordinates of $\bX$ controlled by a loading matrix $\loadn$. 

The dimension reduction inherent in the above model assumption admits a geometric interpretation at the level of spectral (dependence) measures. Specifically, let $\Psi_{\bm X}$ denote the spectral measure associated with $L$; see \eqref{eq:relation-stdf-spectral-dependence-measure} for the general relationship between $L$ and $\Psi_{\bm X}$, and Lemma~\ref{lem:stdf-spectral-dependence-measure-factor-model} for its specific form under the model in \eqref{eq:stdf-factor-introduction}. In this setting, the support of $\Psi_{\bm X}$ is contained in the column space of $\loadn$; see Figure~\ref{fig:sphere-plane} for an illustration. The intrinsic dimension of the support is (at most) $K-1$.

\begin{figure}[t!]
    \centering    
    \includegraphics[width=.5\textwidth, keepaspectratio]{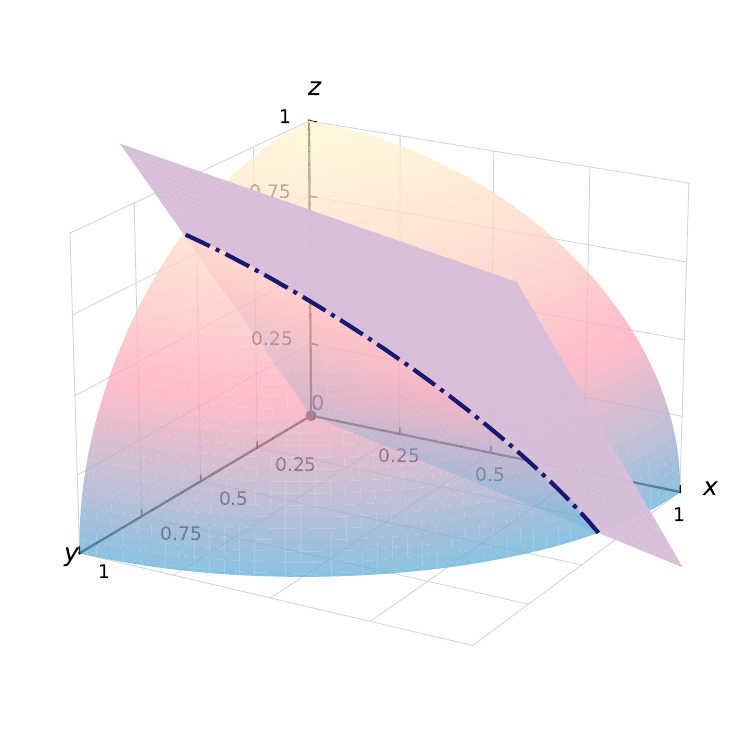}
    \vspace{-1.1cm}
    \caption{Illustration of the support of a spectral dependence measure with respect to the Euclidean norm (dashed black) when $\loadn \in \R^{3\times 2}$ has columns $(1, 1/3, 0)$ and $(0,2/3,1)$.}
    \label{fig:sphere-plane}
\end{figure}

In the general, the model in \eqref{eq:stdf-factor-introduction} with unknown $(K, \loadn, \Psifactor)$  is not identifiable; in fact, it is easy to see that any STDF $L$ can be written as $L=L_{d,I_d, \Psi_{\bX}}$ with $I_d$ the $d\times d$ identity matrix and $\Psi_\bX$ the spectral dependence measure of $\bX$ (see Section~\ref{sec:mathematical-preliminaries} for details). Hence, further restrictions on the parameter space are necessary to ensure that $(K, \loadn, \Psifactor)$ can be reconstructed from $L$. A step in that direction was taken in \cite{boulin2026structured}, who additionally assume that $\Psifactor$ is equal to $K^{-1}\sum_{a \in [K]} \delta_{\bm e_a}$ with $\delta_{\bm e_a}$ the dirac measure at the $a$th unit vector in $\R^K$ (equivalently, the latent factors $Z_1, \dots, Z_K$ are independent, which is a standard assumption in related work as explained below) and that each unit vector $\bm e_a$ appears at least once among the rows of $\loadn$. Under these assumptions, the parameter $(K, \loadn)$ can be identified from the tail correlation matrix $\mathcal X$ with entries $\chi(j,\ell) = 2 - L(\bm e_j + \bm e_\ell)$ for $j,\ell \in [d]$. However, the resulting model was found to be too restrictive for many practical applications. One of the main mathematical results of this paper is that, under the same assumption on $\loadn$ and an additional assumption on $\Psifactor$, the full triple $(K, \loadn, \Psifactor)$ becomes identifiable. 

The proof for the identifiability result is constructive, and relies on the fact that $(K,\loadn)$ can be identified from a margin-free version of the \textit{tail pairwise dependence matrix} from \cite{larsson2012extremal}, formally defined in Definition~\ref{def:tail-pairwise-dependence-matrix} below. Since the latter can be easily estimated from an observed sample, we immediately obtain an estimator for $(K,\loadn)$. In particular, we then know the rows of $\loadn$ which are unit vectors, which allows for estimation of $\Psifactor$, either nonparametrically \citep{EinmahldeHaanPiterbarg2001, EinmahlSegers2009}, or parametrically if one imposes a parametric model \citep{EinmahlKrajinaSegers2012}.

The proposed tail dependence model can be naturally combined with marginal tail models, yielding a wide range of applications, particularly in high dimensions with, say, $d \ge 500$. We develop such an application in detail for a spatial setting.
Suppose that each coordinate of an observable random vector $\bW$ represents the maximal daily wind speed measured at wind turbine $j$, where $d$ denotes the total number of turbines in a region of interest. Wind turbines operate only above a cut-in speed threshold and generate no electricity below this level. For long-term planning and control of energy production, it is therefore important to assess the risk that a large proportion (say, more than 80\%) of turbines simultaneously fail to produce electricity because wind speeds fall below their respective cut-in thresholds \citep{ohlendorf2020frequency}.
Formally, let $\tau_j$ denote the cut-in threshold of turbine $j$. Since $\tau_j$ typically lies in the lower tail of the distribution of $W_j$, this naturally motivates the use of extreme value methods. We are interested in the probability
\begin{align} \label{eq:tail-parameter-introduction}
p 
&= \Prob\Big( \exists J \subset [d] \text{ with } |J| \ge 0.8\, d 
: W_j \le \tau_j \ \text{for all } j \in J \Big).
\end{align}
Below, we show that under a latent linear $K$-factor model for lower-tail dependence as introduced above, and assuming that $\loadn$ satisfies the stated unit-row constraint, the probability $p$ can be approximated by a conditional expectation involving essentially only $K$ coordinates of the observable vector $\bW$. This reduction enables straightforward estimation, which we illustrate in Section~\ref{sec:case-study}.

\medskip
\noindent
\textbf{Related work.} The problem of dimension reduction for multivariate extremes has attracted considerable attention in recent years, motivated by the need to model tail dependence structures in high-dimensional settings. Most contributions focus on the unsupervised setting, where the objective is to characterize and simplify the geometry of the joint tail of a high-dimensional random vector. Existing approaches include clustering methods that group components exhibiting simultaneous extreme behavior \citep{chautru2015dimension, chiapino2019identifying, janssen2020kmeans}, support and sparse extremal directions identification procedures based on the spectral measure \citep{goix2017sparse, simpson2020determining, meyer2024multivariate}, principal component techniques applied to the angular component of extremes \citep{cooley2019decompositions, drees2021dimension, reinbott2025principal}, and graphical models that exploit conditional independence structures in the tail to obtain sparse and interpretable dependence representations \citep{engelke2020graphical, EngelkeVolgushev2022, engelke2025learningextremalgraphicalstructures}. 

Variants of the models in \eqref{eq:genuine-linear-factor-dependence} and \eqref{eq:stdf-factor-introduction} have been widely studied, typically under (asymptotic) independence of the latent factors, leading to a discrete spectral measure \citep{EinmahlKrajinaSegers2012, boulin2026structured}. In this case, the STDF in \eqref{eq:stdf-factor-introduction} has a max-linear form, and the corresponding class of multivariate extreme value distributions is dense in the full model class \citep{Fougeres2013}, highlighting its flexibility.
When $K=d$, such models arise naturally from structural equation models (SEMs); see \cite{gissibl2018maxlinear, kluppelberg2021estimating} for max-linear and \cite{GneccoMeinshausenPetersEngelke2021} for linear formulations. In these and subsequent works, however, the SEM structure is typically imposed on the entire distribution rather than restricted to the tail. Linear and max-linear models have also served as benchmarks for spherical clustering methods designed to recover the atom locations of a discrete spectral measure \citep{janssen2020kmeans, medina2024spectral}.

\medskip
\noindent
\textbf{Outline.} 
We review preliminaries on tail dependence and multivariate regular variation in Section~\ref{sec:mathematical-preliminaries}. Latent linear factor models for tail dependence of the form \eqref{eq:stdf-factor-introduction} are formally introduced in Section~\ref{sec:latent-linear-factor-models}, where we also establish a key identifiability result and summarize relevant tail properties of the model. 
Section~\ref{sec:estimating-latent-linear-factor-models} is devoted to statistical inference for latent linear factor models. In subsequent subsections, we address estimation of $K$, $\loadn$, and $\Psifactor$, as well as of derived tail characteristics that additionally require marginal tail estimation. 
An extensive case study is presented in Section~\ref{sec:case-study}, demonstrating that the proposed methods perform well in practice and yield reasonable extrapolation beyond the range of the observed data.
Section~\ref{sec:conclusion} concludes. All proofs are deferred to a sequence of appendices, which also includes a discussion of latent linear factor modeling for observable regularly varying random vectors and additional details on selected aspects of the case study.

\section{Mathematical preliminaries on tail dependence and multivariate regular variation}
\label{sec:mathematical-preliminaries}

A key mathematical tool for understanding the model for $L$ from \eqref{eq:stdf-factor-introduction} is the concept of multivariate regular variation of a $d$-variate random vector. We briefly summarize that concept, following the exposition in \cite{kulik2020heavy}: let $\Eb_0^d = \R^d \setminus \{ \bm 0\}$, and write $\Mb(\Eb_0^d)$ for the set of Borel measures on $\Eb_0^d$ that are finite for any Borel set that is bounded away from zero. We denote by $\mathcal{C}(\Eb_0^d)$ the set of continuous, bounded, positive functions on $\Eb_0^d$ whose supports are bounded away from zero. For $\nu, \nu_1, \nu_2, \dots \in \Mb(\Eb_0^d)$, we say $\nu_n \rightarrow \nu$ in $\Mb(\Eb_0^d)$ if $\int f d\nu_n \rightarrow \int f d \nu$ for all $f \in \mathcal{C}(\Eb_0^d)$; see Definition B.1.16 in \cite{kulik2020heavy}.

A $d$-variate random vector $\bY$ is regularly varying on $\R^d$ (Definition 2.1.2 and Theorem 2.1.3 in \citealp{kulik2020heavy}) if there exists $\alpha>0$ (the \textit{tail index}), a  function $g:(0,\infty) \to (0,\infty)$ which is regularly varying at infinity with index $\alpha$ (the \textit{auxiliary function}) and a non-zero measure $\nu_\bY \in \Mb(\Eb_0^d)$ (the \textit{exponent measure associated with $g$}) such that
\begin{equation} 
\label{eq:rv} 
g(x) \Prob\left\{ x^{-1} \bY \in \cdot \right\} \rightarrow \nu_{\bY}(\cdot) \quad \textrm{in } \Mb(\Eb_0^d) \quad \text{as }  x \to \infty.
\end{equation} 
In that case, the limit measure $\nu_{\bY}$ has the scaling property $\nu_{\bY}(cA) = c^{-\alpha} \nu_{\bY}(A)$ for any $c > 0$ and any Borel set $A$ bounded away from zero. Moreover, $g$ and $\nu_\bY$ are unique up to a factor in the sense that, if the definition above holds with another choice $(g', \nu_\bY')$, then there exists a constant $\xi>0$ such that $\nu_\bY = \xi \nu_\bY'$ and $\lim_{x \to \infty}g(x) / g'(x) = \xi$ (Theorem B.2.2 in \citealp{kulik2020heavy}).

Let $\| \cdot \|$ be an arbitrary norm on $\R^d$, and let $\simplex{d-1}=\simplex{d-1}(\|\cdot\|)$ be the associated unit sphere; often, we suppress the norm from the notation when it is clear from the context. Consider the bijective transform $T : \Eb_0^d \mapsto (0,\infty) \times \simplex{d-1}$ defined by $T(\bx)=( \| \bx\|, \bx / \| \bx \|)$,
whose inverse $T^{-1} : (0,\infty) \times \simplex{d-1} \mapsto \Eb_0^d$ satisfies $T^{-1}(r,\bm{\lambda}) = r\bm{\lambda}$. The homogeneity of $\nu_{\bY}$ implies that, for any $r > 0$ and any Borel set $A$ in $\simplex{d-1}$,
\begin{align} \label{eq:polar-decomposition}
    \nu_{\bY} \circ T^{-1}((r,\infty)\times A) 
    &= \nonumber
    \nu_{\bY}\Big( \Big\{ \bx \in \Eb_0^d : \| \bx \|> r, \frac{\bx}{\| \bx \|} \in A \Big\} \Big) 
    \\&= 
    r^{-\alpha} \nu_{\bY}\Big( \Big\{ \bx \in \Eb_0^d : \| \bx \| > 1, \frac{\bx}{\| \bx \|} \in A \Big\} \Big) = \varsigma_{\bY} \cdot (\nu_\alpha \otimes \Phi_{\bY})\big( (r,\infty) \times A\big),
\end{align}
where $\varsigma_{\bY} = \varsigma^{\| \cdot\|}_{\bY} = \nu_{\bY}(\{ \bx\in \Eb_0^d : \| \bx \| > 1 \})$ is a strictly positive constant, $\nu_\alpha$ is the measure on $(0,\infty)$ defined by $\nu_\alpha(\diff r) = \alpha r^{-\alpha -1}\diff r$ and
\begin{align} \label{eq:definition-spectral-measure}
    \Phi_{\bY}(\cdot) 
    := 
    \Phi_{\bY}^{\|\cdot\|}(\cdot) 
    := 
    \varsigma^{-1}_{\bY} \cdot \nu_{\bY}\Big( \Big\{ \bx \in \Eb_0^d : \| \bx \| > 1, \frac{\bx}{\| \bx \|} \in \cdot \Big\} \Big),
\end{align}
is the \emph{spectral measure of $\bY$ associated with the norm $\| \cdot\|$}, and where $\otimes$ is the product measure. Note that the spectral measure is a probability measure on $\simplex{d-1}$ that does not depend on the auxiliary function $g(\cdot)$. Moreover, any probability measure on $\simplex{d-1}$ can arise as a spectral measure; see Lemma~2.2.2 in \cite{kulik2020heavy}. It can further be shown that the convergence in \eqref{eq:rv} implies that
\begin{equation} 
\label{eq:spectral_measure}
\mathbb{P} \Big( \frac{\bY}{\|\bY\|} \in \, \cdot \,\, \Big| \, \|\bY\| > x \Big) \rightsquigarrow \Phi_{\bY}(\cdot) \quad \textrm{ in } \simplex{d-1} \quad \text{as }  x \to \infty,
\end{equation} 
where the arrow $\weak$ denote weak convergence of probability measures. Moreover and conversely, if \eqref{eq:spectral_measure} is met for some spectral measure $\Phi_\bY$, and if $\|\bY\|$ is regularly varying of index $\alpha>0$, and if there exists a scaling function $g$ such that $\lim_{y \to \infty}g(y) \Prob(\|\bY\|>y)$ exists and is positive, then $\bY$ is regularly varying of index $\alpha$ and with spectral measure $\Phi_\bY$.

The previous concept of regular variation is typically too strict as a model assumption for an observed random vector $\bX$. However, it can directly be linked to the assumption that the stable dependence function $L$ of $\bX$ from \eqref{eq:definition-stdf} exists. More specifically, 
if $\bX$ has continuous marginal distributions functions $F_1, \dots, F_d$, then the STDF of $\bX$ exists if and only if the random vector $\bY$ with coordinates $Y_j = 1/(1-F_j(X_j))$ (which are standard Pareto distributed) is regularly varying. 
In that case, the index of regular variation is necessarily 1, the spectral measure $\Phi_\bY$ of $\bY$ with respect to the norm $\|\cdot\|$ is concentrated on $\simplexpos{d-1}=\simplexpos{d-1}(\|\cdot\|)$ (the set of all vectors $\bx \in [0,\infty)^d$ with $\| x \|=1$) and it is linked to $L$ by
\begin{align} 
\label{eq:relation-stdf-spectral-dependence-measure}
L(\bm x) = \varsigma_{\bm Y} \int_{\simplexpos{d-1}} \max_{j \in [d]}(\lambda_j x_j) \, \Phi_\bY( \diff \bm \lambda),
\end{align}
where $\varsigma_{\bm Y} = \varsigma_{\bm Y}^{\| \cdot \|}$ is a constant depending on the norm that is explicitly given in Lemma~\ref{lem:constant-in-front-of-L}, with the two most convenient expressions being $\varsigma_{\bm Y}=d$ for the 1-norm and $\varsigma_{\bm Y}=L(\bm 1)$ for the max-norm. 
We call $\Psi_\bX:= \Phi_\bY|_{\simplexpos{d-1}}$ the \textit{spectral dependence measure of $\bX$ with respect to the norm $\| \cdot \|$}, and we write $\setspectraldepposd{\| \cdot \|}$ for the set of spectral dependence measures with respect to the norm $\| \cdot \|$; here, the index $\mathrm d$ stands for \textit{dependence}. 
Since the spectral dependence measure is uniquely determined by $L$ (see, e.g., Lemma 2.1 in \citealp{boulin2026structured}), there exists a bijection 
\begin{align} \label{eq:stdf-spectral-dependence-measure-bijection}
    S_{\| \cdot \|} : \setstdf{d} \to \setspectraldepposd{\|\cdot\|}
\end{align} 
between $\setstdf{d}$, the set of $d$-variate STDFs, and $\setspectraldepposd{\|\cdot\|}$.
Restricting attention to the $p$-norms with $p \in [1,\infty]$, it  can be shown that
\begin{align}
\label{eq:set-of-spectral-dependence-measures}
    \setspectraldepposd{\| \cdot \|_p}
    =
    \Big\{ \Psi \text{ prob.\ measure on } \simplexpos{d-1}(\| \cdot\|_p) \, \Big| \, \Exp_{\bm \Lambda \sim \Psi}[\Lambda_1] =  \dots = \Exp_{\bm \Lambda \sim \Psi}[\Lambda_d] \in [d^{-1},d^{-1/p}]\Big\},
\end{align}
where we interpret $d^{-1/\infty}=d^0=1$; see Lemma~\ref{lem:set-of-spectral-measures}. Remarkably, for $d=1$, all marginal expectations are equal to $d^{-1}$. For all choices $p>1$, the two extreme cases $\Exp_{\bm \Lambda \sim \Psi}[\Lambda_1]=d^{-1}$ and $\Exp_{\bm \Lambda \sim \Psi}[\Lambda_1]=d^{-1/p}$ correspond to tail independence and perfect tail dependence, respectively.
Finally note that, if $\bX$ itself is regularly varying and $[0,\infty)^d$-valued, the spectral measure $\Phi_\bX$ of $\bX$ (which also lives on $\simplexpos{d-1}(\| \cdot\|)$) is in general different from its spectral dependence measure $\Psi_\bX$; even their supports may be different.

\section{Latent linear factor models for tail dependence}
\label{sec:latent-linear-factor-models}

Write $\setloadn$ for the set of matrices $\loadn \in [0,1]^{d\times K}$ with row sums equal to one and with full column rank; the index $s$ stands for \textit{standardized}. Recall the set of spectral dependence measure $\setspectraldepposK{\|\cdot\|_1}$ from \eqref{eq:set-of-spectral-dependence-measures}.
We start by calculating the STDF for the simple linear factor model from \eqref{eq:genuine-linear-factor-dependence} in the introduction.

\begin{proposition}
\label{prop:genuine-linear-factor-dependence}
Let $K \in [d]$, $\loadn \in \setloadn$ and $\Psifactor\in \setspectraldepposK{\|\cdot\|_1}$. Then the random vector
\[
\bY' := \loadn \bZ \quad \text{ with } \quad \bZ := K R \bm \Lambda,
\]
where $R$ is standard Pareto and $\bm \Lambda  \sim \Psifactor$ is independent of $R$, has STDF $L=L_{K, \loadn, \Psifactor}$ as in \eqref{eq:stdf-factor-introduction}, that is, 
\begin{align} \label{eq:stdf-factor}
    L_{K, \loadn, \Psifactor}(\bx)  
    = 
    K  \int_{\simplexpos{K-1}} \bigvee_{j \in [d]} \Big( x_j\sum_{a\in [K]} \loadn_{ja} z_a\Big)  \, \Psifactor(\diff\bz), \quad \bx\in [0,\infty)^d.
\end{align}
Further, the spectral dependence measure of $\bm Z$ with respect to the 1-norm is given by $\Psi_\bZ = \Psifactor$.
\end{proposition}

The following definition contains the essential modeling assumption for the observable random vector $\bX$ employed throughout this paper.

\begin{definition}[Latent linear factor tail dependence]
\label{def:latent-linear-factor-tail-dependence}
Let $\bX$ be a $d$-variate random vector with STDF $L$ and let $K\in [d]$. The tail dependence of $\bX$ is said to be of \textit{latent (linear) $K$-factor form} with \textit{factor loading matrix} $\loadn \in \setloadn$  and \textit{factor spectral dependence measure} $\Psifactor \in \setspectraldepposK{\|\cdot\|_1}$
if $L=L_{K, \loadn, \Psifactor}$ from \eqref{eq:stdf-factor}.
\end{definition}

As explained in Section~\ref{sec:mathematical-preliminaries}, for any fixed norm $\|\cdot\|$ on $\R^d$, the STDF of $\bX$ is in one-to-one correspondence to the spectral dependence measure $\Psi_\bX$ with respect to $\|\cdot\|$. The following lemma characterizes $\Psi_\bX$ under the assumption of latent linear $K$-factor tail dependence.

\begin{lemma}[Latent linear factor tail dependence on the level of spectral dependence measures]
\label{lem:stdf-spectral-dependence-measure-factor-model}
Fix a norm  $\|\cdot\|$ on $\R^d$. The $d$-variate random vector $\bX$ has STDF $L=L_{K, \loadn, \Psifactor}$ as in \eqref{eq:stdf-factor} if and only if the spectral dependence measure $\Psi_{\bX}$ of $\bX$ with respect to $\|\cdot\|$ is the pushforward of $\Psifactor^\loadn$ by $f_{\loadn}$, i.e., 
\begin{align}
\label{eq:spectral-dependence-measure-factor}
\Psi_\bX = f_{\loadn}  \, \# \, \Psifactor^\loadn =  \Psifactor^\loadn(f^{-1}_{\loadn}(\cdot) ),
\end{align}
where $f_\loadn=f_{\loadn, \| \cdot \|_1, \| \cdot \|}:\simplexpos{K-1}(\|\cdot\|_1) \to \simplexpos{d-1}(\|\cdot\|)$ is defined by $f_\loadn(\bz) = \loadn \bz / \| \loadn \bz\|$
and where $\Psifactor^\loadn$ is the (tilted) probability measure on $\simplexpos{K-1}(\|\cdot\|_1)$ defined by
\begin{align}
\label{eq:cZA}
\Psifactor^\loadn(\diff \bz) 
&= c_{\loadn, \Psifactor}^{-1} \, \| \loadn \bz \| \,  \Psifactor(\diff z) 
\quad \text{ with } \quad
c_{\loadn, \Psifactor} = \int_{\simplexpos{K-1}(\|\cdot\|_1)}\| \loadn \bz \| \, \Psifactor(\diff z).
\end{align}
As a consequence, the support of $\Psi_\bX$ is a subset of the column span of $\loadn$, and hence has intrinsic dimension at most $K-1$; see Figure~\ref{fig:sphere-plane} for a graphical illustration.
\end{lemma}

The identifiability results announced in the introduction heavily rely on the following margin-free version of the tail pairwise dependence matrix, whose entries essentially correspond to the extremal dependence measure from \cite{larsson2012extremal} applied to the unobservable random vector $\bY$ with coordinates $Y_j=1/(1-F_j(X_j))$.

\begin{definition}[Tail pairwise dependence matrix]
\label{def:tail-pairwise-dependence-matrix}
    Let $\bX$ be a random vector with spectral dependence measure $\Psi_\bX$ with respect to the norm $\|\cdot\|$. The \textit{tail pairwise dependence matrix (TPDM) of $\bX$ with respect to the norm $\|\cdot\|$} is defined as 
    \begin{align} \label{eq:definition-tpdm}
    \TPDM_\bX = \TPDM_\bX^{\|\cdot\|} = \int_{\simplexpos{d-1}(\|\cdot\|)} \bx \bx^\top \Psi_\bX(\diff \bx).
    \end{align}
\end{definition}

Recall $\bY$ with coordinates $Y_j = 1/(1-F_j(X_j))$.
In view of the fact that
$
\Prob(\bY/\| \bY \| \in \,\cdot\, \,|\, \| \bY \| > y)
\weak
\Phi_\bY = \Psi_\bX
$
as $y \to \infty$, 
we have
\begin{align}
\label{eq:tail-pairwise-dependence-matrix-as-limit-of-conditional-expectation}
\TPDM_\bX = \lim_{y \to \infty}\Exp\Big[ \frac{\bY \bY^\top}{\| \bY \|^2} \,\Big|\, \| \bY \| > y\Big].
\end{align}
The TPDM under the assumption of latent linear factor tail dependence can be calculated explicitly.

\begin{proposition}[TPDM for latent linear factor tail dependence]
\label{prop:tpdm-latent-factor-linear-model}
    Suppose $\bX$ has STDF $L=L_{K, \loadn, \Psifactor}$ from \eqref{eq:stdf-factor}. Then, with $c_{\loadn, \Psifactor}$ from \eqref{eq:cZA}, we have, for any norm $\|\cdot\|$,
    \begin{align} \label{eq:C_ZA}
    \TPDM_{\bX} = \loadn C_{\loadn, \Psifactor} \loadn^\top \quad \text{ with } \quad C_{\loadn, \Psifactor} = c_{\loadn, \Psifactor}^{-1} \int_{\simplexpos{K-1}(\|\cdot\|_1)} \frac{\bz \bz^\top}{\| \loadn \bz\|} \, \Psifactor(\diff \bz).
    \end{align}
\end{proposition}

The formula from Proposition~\ref{prop:tpdm-latent-factor-linear-model} is essential for the subsequent identifiability result, which, as noted in the introduction, necessitates further assumptions about the triplet $(K,\loadn, \Psifactor)$. First of all, we require that $\loadn$ satisfies the \textit{pure variable assumption} introduced in \cite{bingadaptive2020} and \cite{boulin2026structured}: we have $\loadn \in \setloadnpure$, where
\begin{align}
\label{eq:pure-loading-matrices-taildep}
    \setloadnpure = \big\{ \loadn \in \setloadn \mid \forall a \in [K]\, \exists j \in [d]: \loadn_{j\cdot}= \bm e_a^\top \big\}
\end{align}
and where $\bm e_a$ denotes the $a$th unit vector in $\R^K$; recall that $\setloadn$ is the set of matrices $\loadn \in [0,1]^{d\times K}$ with row sums equal to one and with full column rank. 
The assumption $\loadn \in \setloadnpure$ essentially means that each latent factor appears at least once among the observed coordinates of~$\bX$. Consequently, the tail dependence of $\bX$ is essentially explained by those coordinates of $\bX$ that correspond to the unit-vector rows of~$\loadn$.

As shown in \cite{boulin2026structured}, the pure variable assumption is sufficient for identifiability of $(K, \loadn)$ in case the latent factors are required to be asymptotically independent, that is, if $\Psifactor = K^{-1}\sum_{a \in [K]} \delta_{\bm e_a}$. If, on the contrary, $\Psifactor$ is allowed to be an arbitrary (1-norm) spectral dependence measure, identifiability does not hold: one can always choose $K=d, \loadn=I_d$ (the identity matrix) and $\Psifactor=\Psi_\bX$, the 1-norm spectral dependence measure of $\bm X$. A composite model that is identifiable is obtained under the following \textit{diagonal dominance assumption}: we have
$\Sigma_{\Psifactor} \in \diagdom{K}$, where
\[
\Sigma_{\Psifactor} =  \int_{\simplexpos{K-1}(\|\cdot\|_1)} \bz \bz^\top \, \Psifactor(\diff \bz)
\]
is the tail pairwise dependence matrix of $\Psifactor$,  where
\begin{align}
\label{eq:diagdom}
\diagdom{K} = \big\{ C \in \R^{K \times K}: \Delta(C)>0 \text{ and } C \text{ is positive definite}\big\},
\end{align}
and where, for  $C \in \R^{K \times K}$,
\begin{align*}
    \Delta(C) := \min_{(a,b) \in [K], a \ne b} C(a,a) \wedge C(b,b) - |C(a,b)|.
\end{align*}
The assumption that $\Delta(\Sigma_{\Psifactor})>0$ implies that $\Psifactor(\{\bz \in \simplexpos{K-1}: z_a = z_b\} )<1$ for any $a\ne b$\footnote{Indeed, suppose that there exist $a\ne b$ such that $\Psifactor(\{\bz \in \simplexpos{K-1}: z_a = z_b\} ) = 1$. Suppose that $\Sigma_{\Psifactor}(a,a) = \Sigma_{\Psifactor}(a,a) \wedge \Sigma_\Psifactor(b,b)$; the other case can be treated similarly. Then
$
    \Sigma_{\Psifactor}(a,a) \wedge \Sigma_{\Psifactor}(b,b) - \Sigma_{\Psifactor}(a,b) = \Sigma_{\Psifactor}(a,a) - \Sigma_{\Psifactor}(a,b) = \int_{\simplexpos{K-1}(\| \cdot \|_1)}z_a(z_a-z_b) \Psifactor(\diff \bz) =0,
$
which implies $\Delta(\Sigma_{\Psifactor})=0$.}, which means that there does not exist a pair of latent variables that is perfectly tail dependent. The necessity of the latter constraint is intuitively obvious: if we had $\Psifactor(\{\bz \in \simplexpos{K-1}: z_a = z_b\} )=1$ for some $a\ne b$, we would obtain the same STDF if we replace $K$ by $K-1$, $\loadn$ by the matrix in which we delete the $b$th column and replace the $a$th column by the sum of the $a$th and $b$th column, and replace $\psi$ by the distribution of $(\Lambda_{c})_{c \ne b}$ for $\Lambda \sim \psi$. Finally, note that $\Sigma_\psi \in \diagdom{K}$ for $\psi=K^{-1} \sum_{a \in [K]} \delta_{\bm e_a}$, the spectral dependence measure corresponding to tail independence.

Overall, the previous considerations lead to the following parameter space
\begin{multline}
    \label{eq:parameter-space-tail-dependence}
    \paraspacetaildep
    = 
    \big\{ (K, \loadn,  \Psifactor) \mid K \in [d],
     \loadn \in \setloadnpure \text{ and } \psi \in \setspectraldepposK{\|\cdot\|_1} 
     \\ 
     \textrm{ satisfies } \Sigma_\psi \in \diagdom{K} \big\}.
\end{multline}
Clearly, identifiability can only hold up to label permutations of the factors. Mathematically, 
two pairs $(\loadn, \Psifactor)$ and $(\loadn', \Psifactor')$ from $\setloadnpure \times \setspectraldepposK{\|\cdot\|_1}$ are called equivalent, notation  $(\loadn, \Psifactor) \sim (\loadn', \Psifactor')$ if there exists a permutation matrix $P \in \{0,1\}^{K \times K}$ with $P^{-1}=P^\top$ such that
\[
\loadn = \loadn'P \quad \text{ and } \quad  \Psifactor = L_{P^\top}\, \#\, \Psifactor' = \Psifactor'(L_{P^\top}^{-1}(\cdot)),
\]
where $L_{P^\top}:\simplexpos{K-1} \to \simplexpos{K-1}$ is defined by $ \bz \mapsto L_{P^\top}(\bz) = P^\top \bz$. In terms of random variables, the second equation in the last display means that $\bm \Lambda =_d P^\top \bm \Lambda'$ for $\bm \Lambda \sim \Psifactor$ and $\bm \Lambda' \sim \Psifactor'$, and together with the first equation, we get that $\loadn \bm \Lambda =_d \loadn'P P^\top \bm \Lambda' = \loadn'\bm \Lambda'$. We can now state the formal result.

\begin{theorem}
\label{thm:identifieability-tail-dependence}
Recall $L=L_{K,\loadn, \Psifactor}$ from \eqref{eq:stdf-factor-introduction},
and let $\setstdf{d}$ denote the set of $d$-variate STDFs.
The mapping 
\[
T: \paraspacetaildep \mapsto \setstdf{d}, \quad (K, \loadn, \Psifactor) \mapsto L_{K,\loadn, \Psifactor}
\] 
is injective up to label permutations, that is, 
$L_{K,\loadn, \Psifactor} = L_{K',\loadn', \Psifactor'}$
implies $K=K'$ and $(\loadn,\Psifactor) \sim (\loadn',\Psifactor')$.
\end{theorem}

The proof of Theorem \ref{thm:identifieability-tail-dependence} is constructive, and we explain the most important steps in the following remark. 
For a symmetric matrix $\Sigma=(\Sigma(j,\ell))_{j,\ell=1}^d \in [0,\infty)^{d \times d}$,  the row-wise maxima with corresponding maximizing set of indices will be denoted by
\begin{align} \label{eq:m_j-and_S_j}
m_j(\Sigma) = \max_{\ell \in [d]} \Sigma(j, \ell), \qquad S_j(\Sigma) = \argmax_{\ell\in[d]} \Sigma(j,\ell) = \big\{ \ell \in [d]: \Sigma(j,\ell)=m_j\big\}, \qquad j \in [d].
\end{align}
For $\loadn \in \setloadnpure$, let $I(\loadn)$ denote the set of pure variable indices of $\loadn$, and let $I_a(\loadn) \subset I(\loadn)$ denote the subset of such indices which correspond to the $a$th factor; formally,
\begin{align} \label{eq:pure-variables}
I(\loadn) =\bigcup_{a \in [K]} I_a(\loadn), \qquad I_a(\loadn) = \{ j \in [d]: \loadn_{j \cdot} = \bm e_a^\top\}.
\end{align}

\begin{remark}[Reconstructing $(K,\loadn,\Psifactor)$ from $L$]
\label{rem:reconstructing-parameters-tail-dependence}
The following steps allow for reconstructing $K,\loadn$ and $\Psifactor$ from $L$:
\begin{compactenum}[(1)]
    \item Suppose $\bX$ has STDF $L=L_{K,\loadn, \Psifactor}$. It can be shown that the assumption $\Sigma_\psi \in \diagdom{K}$ guarantees that there exists norm on $\R^d$ for which the matrix $C=C_{\loadn, \Psifactor}$ from \eqref{eq:C_ZA} satisfies $C \in \diagdom{K}$. If $\int \bm z \bm z^\top / \|\bm z \|_\infty \psi(\diff \bm z) \in \diagdom{K}$, then one can always choose the maximum norm on $\R^d$; see Lemma~\ref{lem:sufficient-condition-for-maximum-norm}.
    \item In view of Lemma~\ref{lem:stdf-spectral-dependence-measure-factor-model}, $L_{K,\loadn, \Psifactor}$ identifies the spectral dependence measure $\Psi_\bX$ with respect to the norm chosen in (1).
    \item Clearly, $\Psi_\bX$ identifies the TPDM $\TPDM=\TPDM_\bX$ by definition in \eqref{eq:definition-tpdm}.
    \item The number of factors $K$ and the pure variable indices $(I_a(\loadn))_{a\in[K]}$ can be reconstructed from $\TPDM$ (up to label permutations). Indeed, it can be shown that $I(\loadn)=\{j \in [d] \mid \forall \ell \in S_j(\TPDM):  m_j(\TPDM) = m_\ell(\TPDM)\}$, that $K$ is equal to the number of equivalence classes under the equivalence relation on $I(\loadn)$ given by 
    \[
    j \sim_\TPDM \ell \quad \Longleftrightarrow \quad  \ell \in S_j(\TPDM),
    \]
    and that each equivalence class corresponds to exactly one set $I_a:=I_a(\loadn)$ with $a \in [K]$; see Lemma~\ref{lem:identifiability_pure} for details. 
    \item The matrix $C=C_{\loadn, \Psifactor}$ can be reconstructed from $(I_a)_{a \in [K]}$ and $\TPDM$: for $a,b\in[K]$, we have $C(a,b) = \TPDM(j_a, j_b)$, where $j_a\in[d]$ denotes an arbitrary fixed element of $I_a$.
    \item The matrix $\loadn$ can be reconstructed from $(I_a)_{a \in [K]}$, $\TPDM$ and $C$: first, for each $j \in I(\loadn)$, there exists exactly one $a \in [K]$ such that $\loadn_{j\cdot} = \bm e_a^\top$; we have hence identified the rows of $\loadn$ corresponding to $j \in I(\loadn)$. Next, for each $j \in [d] \setminus I(\loadn)$, the $j$th row $\beta_j := \loadn_{j\cdot}^\top \in \R^K$ satisfies
    $
    C \beta_j =  
    (  \TPDM(j_1,j), \dots, \TPDM(j_K,j))^\top,
    $
    which can easily be solved for $\beta_j$ by the invertibility assumption on $C$.
    \item Finally, $\Psifactor$ can be identified from $(I_a)_{a \in [K]}$ and $L$. Indeed, let $\tilde \bx=(\tilde x_a)_{a \in [K]} \in [0,\infty)^K$ be arbitrary. Define $\bx \in [0,\infty)$ by $\tilde x_j=0$ if $j \notin\{j_1, \dots, j_K\}$ and by $x_{j_a}=\tilde x_a$ otherwise. Then 
    \[
    L(\bm x) = K  \int_{\simplexpos{K-1}} \bigvee_{j \in [d]} \Big( x_j\sum_{a\in [K]} \loadn_{ja} z_a\Big)  \, \Psifactor(\diff\bz)
    =
    K \int_{\simplexpos{K-1}}\bigvee_{a \in [K]} ( \tilde x_{a} z_a) \, \Psifactor(\diff z).
    \]
    The right hand side is the STDF $L_\bZ(\tilde \bx)$ associated with $\Psifactor$, see \eqref{eq:relation-stdf-spectral-dependence-measure},  which uniquely determines~$\Psifactor$.
\end{compactenum}
The main objects used in these seven steps have natural empirical counterparts, which give rise to estimators for $(K,\loadn,\Psifactor)$ that we introduce in Section~\ref{sec:estimating-latent-linear-factor-models}.
\end{remark}

We end this section with a collection of general tail properties of a random vector $\bm X$ with latent linear factor tail dependence. The results are crucial for estimation of parameters such as \eqref{eq:tail-parameter-introduction} introduced in the introduction.

\begin{proposition}
    \label{prop:survival_stdf_lfm_tail_dep}
    Suppose $\bX$ has STDF $L=L_{K,\loadn, \Psifactor}$ from \eqref{eq:stdf-factor}. Let $\mathcal{J}$ be a non-empty set of non-empty subsets of $[d]$. Then, for any $\bx \in [0, \infty)^d$, 
    \begin{align*}
        &R^\cup_{\mathcal{J}}(\bx) 
        \equiv 
        \lim_{x \rightarrow \infty} x \mathbb{P}\Big( \exists J \in \mathcal{J}\,  \forall j \in J : F_j(X_j) > 1-\frac{x_j}{x} \Big) 
        = 
        K\int_{\simplexpos{K-1}} \bigvee_{J \in \mathcal{J}} \bigwedge_{j \in J} \Big( \sum_{a\in[K]} \loadn_{ja} z_a\Big) x_j \Psifactor(\diff \bz),
        \\
        &L^\cap_{\mathcal{J}}(\bx) 
        \equiv 
        \lim_{x \rightarrow \infty} x \mathbb{P}\Big( \forall J \in \mathcal{J}\, \exists j \in J : F_j(X_j) > 1-\frac{x_j}{x} \Big) 
        = 
        K\int_{\simplexpos{K-1}} \bigwedge_{J \in \mathcal{J}} \bigvee_{j \in J} \Big(\sum_{a\in[K]} \loadn_{ja} z_a \Big) x_j \Psifactor(\diff \bz).
    \end{align*}
    In particular, for all $\emptyset \ne J \subset [d]$, the tail copula and the STDF of the subvector $\bX_J = (X_j)_{j \in J}$ exist and satisfy
    \begin{align*}
        &R_J(\bx_J) 
        \equiv 
        \lim_{x \rightarrow \infty } x \mathbb{P}\Big( \forall j \in J : F_j(X_j) > 1-\frac{x_j}{x} \Big) 
        = 
        K\int_{\simplexpos{K-1}} \bigwedge_{j \in J} \Big( \sum_{a\in[K]} \loadn_{ja}z_a \Big) x_j \Psifactor(\diff \bz) \\
        &L_J(\bx_J) 
        \equiv 
        \lim_{x \rightarrow \infty } x \mathbb{P}\Big( \exists j \in J : F_j(X_j) > 1-\frac{x_j}{x} \Big) 
        =
        K\int_{\simplexpos{K-1}} \bigvee_{j \in J} \Big( \sum_{a\in[K]} \loadn_{ja}z_a \Big) x_j \Psifactor(\diff \bz).
    \end{align*}
    Moreover, for each $j,\ell \in [d]$ with $j\ne \ell$, the tail correlation between $X_j$ and $X_\ell$ satisfies:
    \begin{align} \label{eq:model-tail-correlation}
        \chi(j,\ell) = R_{\{j,\ell\}}(1,1) = K\int_{\simplexpos{K-1}} \Big(\sum_{a\in[K]} \loadn_{ja} z_a\Big) \wedge \Big( \sum_{a\in[K]} \loadn_{\ell a}z_a \Big) \Psifactor(\diff \bz). 
    \end{align}
\end{proposition}

\section{Estimation of latent linear factor models for tail dependence}
\label{sec:estimating-latent-linear-factor-models}

In this section, we assume that an i.i.d.\ sample $\bX_1, \dots, \bX_n$ from a random vector $\bX$ is observed, where $\bX$ has continuous marginal cdfs and STDF $L = L_{K,\loadn,\Psifactor}$ for some unknown parameter triple $(K, \loadn, \Psifactor) \in \paraspacetaildep$, with $\paraspacetaildep$ defined in \eqref{eq:parameter-space-tail-dependence}.
The primary objective, addressed in Sections~\ref{subsec:estimating-K-and-pure-variables} and \ref{subsec:estimating-loading-matrix}, is to estimate $K$ and $ \loadn$. Once these parameters are available, the lower-dimensional spectral dependence measure $\Psifactor$ can be estimated by a method of choice; a nonparametric approach is described in Section~\ref{subsec:estimation-latent-spectral-measure}. 
Section~\ref{subsec:estimation-derived-tail-parameters} discusses the estimation of derived tail characteristics, such as the quantity introduced in \eqref{eq:tail-parameter-introduction}. Finally, Section~\ref{subsec:hyperparameter-and-model-validation} addresses the selection of tuning parameters and issues of model validation.

\subsection{Estimation of the number of factors and the pure variables}
\label{subsec:estimating-K-and-pure-variables}

As in \cite{bingadaptive2020} and \cite{boulin2026structured}, we start by estimating three key components: the number of factors $K$, the pure variable index set $I=I(\loadn)$, and its partition $\bm \Ic = (I_a(\loadn))_{a \in [K]}$. To achieve this, we use an empirical version of the construction detailed in Remark~\ref{rem:reconstructing-parameters-tail-dependence}, particularly in Step (3). The starting point is the choice of a norm $\| \cdot\|$ on $\R^d$, for which we then calculate the \textit{empirical tail pairwise dependence matrix}, an empirical version of the conditional expectation on the right-hand side of \eqref{eq:tail-pairwise-dependence-matrix-as-limit-of-conditional-expectation}. Formally,  for a threshold parameter $k \in [n]$ to be chosen by the statistician, the empirical TPDM is defined as
\begin{equation}
    \label{eq:empirical_tpdm}
    \hat \TPDM = \hat{\TPDM}_{n,k} 
    = 
    \frac{1}{k} \sum_{i=1}^n \frac{\hat\bY_i \hat \bY_i^\top}{\hat R_i^2} \bm 1_{ \{ \hat R_i > \hat R_{n-k:n} \} },
\end{equation}
where $\hat \bY_i=(\hat Y_{i1}, \dots, \hat Y_{id})$ has coordinates $\hat Y_{ij} = 1 / \{1 - \hat F_{nj}(X_{ij})\}$ with $\hat F_{nj}(x_j) = (n+1)^{-1} \sum_{i \in [n]} \bm 1(X_{ij} \le x_j)$ the (rescaled) empirical cdf of $X_{1j}, \dots, X_{nj}$, 
where $\hat R_i = \|\hat \bY_i\|$ and where $\hat R_{1:n}  \le \dots \le \hat R_{n:n}$ are the order statistics of $\hat R_1,\dots,\hat R_n$. Now, inspired by Step (3) in Remark~\ref{rem:reconstructing-parameters-tail-dependence}, we propose to proceed as follows: iterate through the index set $[d]$, and use the sample version $\hat \TPDM$ of $\TPDM$ to decide whether an index $j$ is a pure variable index. For that purpose, we use a small threshold parameter $\kappa>0$, and declare $j$ as a pure variable index if, for each $\ell \in \hat S_j = \{ \ell \in [d]: \hat{\TPDM}(j,\ell) \ge  \max_{m \in [d]} \hat \TPDM(j,m) - 2\kappa\}$, we have  $\hat{\TPDM}(j,\ell) \ge \max_{m \in [d]} \hat \TPDM(\ell,m) - 2\kappa$; this amounts to checking whether $m_j(\TPDM) = m_\ell(\TPDM)$ for all $\ell \in S_j(\TPDM)$. The parameter $\kappa$ should be chosen small, but, heuristically, not smaller than the estimation uncertainty of $\hat \TPDM$.
Next, if $j$ is declared a pure variable index, we add $\hat{S}_j$ to the estimated partition $\hat{\bm \Ic}$, potentially merging it with previously identified elements of $\hat{\bm \Ic}$. The method is summarized in Algorithm~\ref{alg:purevar}.

\begin{algorithm}[t]
\caption{PureVar Algorithm}
\label{alg:purevar}
\begin{algorithmic}[1]
\Require{Estimated TPDM $\hat{\TPDM} \in [0,\infty)^{d \times d}$, threshold parameter $\kappa>0$.}
\Ensure{Estimated number of pure variables $\hat K$, pure variable index set $\hat I$, and partition of the pure variable index set $\hat{\bm \Ic} = (\hat I_1, \dots, \hat I_{\hat K})$.}
\Function{PureVar}{$\hat{\TPDM}, \kappa$}
    \State $\hat{\bm \Ic} \gets \emptyset$
    \For{$j = 1,\dots,d$}
        \State $\hat m_j \gets \max \{ \hat \TPDM(j,\ell) \mid  \ell \in [d] \}$ \Comment{Estimator of $m_j$ from \eqref{eq:m_j-and_S_j}.}
        \State $\hat{S}_j \gets \{ \ell \in [d] : \hat m_j - \hat{\TPDM}(j,\ell) \le 2\kappa \, \}$ \Comment{Estimator of $S_j$ from \eqref{eq:m_j-and_S_j}.}
        \If{$\max \{   \hat m_\ell - \hat \TPDM(j,\ell)  : \ell \in \hat S_j \}\le 2 \kappa$} 
            \Comment{Check if $j$ is a pure variable index.}
            \If{$\hat S_j \cap (\bigcup_{I \in \hat{\bm \Ic}}I ) = \emptyset$} \Comment{Add $\hat S_j$ as a new set of pure variables.} 
                \State $\hat{\bm \Ic} \gets \hat{\bm \Ic} \cup \{ \hat S_j\}$
            \Else \Comment{Merge $\hat S_j$ with an existing  set of pure variables.} 
                \For{$I \in \hat{\bm \Ic}$}
                \If{$\hat S_j \cap I \ne \emptyset$}
                    \State $I \gets I \cup \hat S_j$ 
                \EndIf
            \EndFor 
            \EndIf
        \EndIf
    \EndFor
    \State \Return $(\hat K,\hat I, \hat{\bm \Ic}) := (|\hat{\bm \Ic}|, \bigcup_{I \in \hat{\bm \Ic} }I, \hat{\bm \Ic})$
\EndFunction
\end{algorithmic}
\end{algorithm}

\subsection{Estimating the loading matrix}
\label{subsec:estimating-loading-matrix}

In view of Step (5) of Remark~\ref{rem:reconstructing-parameters-tail-dependence}, we may use the estimator $\hat \TPDM$ from \eqref{eq:empirical_tpdm} and $\hat{K}$, $\hat{I}$ and $\hat{\bm \Ic}$ obtained from Algorithm \ref{alg:purevar} to construct an estimator $\hat \loadn \in [0,\infty)^{d \times \hat K}$ 
for the loading matrix $\loadn$. We introduce an additional hyperparameter $\lambda\ge 0$, with $\lambda>0$ enforcing sparsity in the rows of $\hat \loadn$. The latter property is referred to as \textit{row-sparsity}, which, on the population level, means that any variable $Y_j$ is only connected to a small number of latent factor. 

The estimator $\hat\loadn$ is constructed as follows. First of all, for each $j \in \hat I$, there exists a unique $a \in [\hat{K}]$ such that $j \in \hat{I}_a$, and we define $\hat \loadn_{j \cdot} = \bm e_a^\top$. Next, for each $j \in [d] \setminus \hat I$, Step (v) of Remark~\ref{rem:reconstructing-parameters-tail-dependence} suggests to define $ \hat\loadn_{j\cdot}^\top$ as any solution $\bm \beta_j$ of the equation 
\[
\hat C \bm \beta_j = \Big( \frac{1}{|\hat I_1|} \sum_{\ell \in\hat  I_1}  \hat \TPDM(\ell,j), \dots, \frac{1}{|\hat I_K|} \sum_{\ell \in \hat I_K}  \hat \TPDM(\ell,j) \Big)^\top =: \hat {\bm \theta}_j,
\]
where $\hat C =(\hat C(a,b))_{a,b \in [\hat K]}$ has coordinates
\[
\hat C(a,b) = \frac{1}{|\hat I_a| \cdot |\hat I_b|} \sum_{j \in \hat I_a, \ell \in \hat I_b} \hat \TPDM(j,\ell).
\]
To potentially enforce sparsity of the estimator, we use the Lasso-approach, and define
\[
\hat{\bm \beta}_j^{\mathrm{LASSO}} (\lambda)
=
\argmin_{\bm\beta \in \R^{\hat K}} \Bigl\{ \frac{1}{2}\|\hat{C} \bm \beta - \hat{\bm \theta}_j\|_2^2 + \lambda \|\beta\|_1 \Bigr\},
\]
Note that the optimization problem can be efficiently solved by the Fast Iterative Shrinkage-Thres\-holding Algorithm (FISTA) from \cite{BeckTeboulle2009}. Additionally, one could employ any modification of the plain Lasso, such as OLS post-Lasso, relaxed Lasso, or de-sparsified Lasso. Based on simulation experience, we suggest to work with the OLS post-Lasso estimator as the standard choice. More specifically, for $j \in [d]$, let $\hat{\mathrm{supp}_j}
=\{ a \in [\hat K] : \hat{\beta}^{\mathrm{LASSO}}_{j,a} \neq 0 \}$, and let $\hat C_{j}$ be the principal submatrix of $\hat C$ obtained by restricting both rows and columns to $\hat{\mathrm{supp}_j}$. We then redefine the subvector $(\hat{\bm \beta}_{j} ^{\mathrm{LASSO}})_{\hat{\mathrm{supp}_j}}$
to be any solution of the restricted OLS equation system $\hat C_{j} \bm \beta_{\hat{\mathrm{supp}_j}} = (\hat {\bm \theta}_j)_{\hat{\mathrm{supp}_j}}$, where we write $\bm \beta_S\in \mathbb{R}^{|S|}$ for the vector obtained by restricting $\bm \beta$ to the coordinates in $S$.

The estimator $\hat{\bm \beta}_j^{\mathrm{LASSO}} (\lambda)$ is neither guaranteed to have coordinates in $[0,1]$ nor 1-norm equal to one. These constraints can for instance be enforced by the Simplex Projector from \cite{kyrillidis2013sparse}, applied to the non-zero coordinates of $\hat{\bm \beta}_j^{\mathrm{LASSO}} (\lambda)$, see also \cite{boulin2026structured}. Denoting the resulting estimate by $\hat \loadn_{j\cdot}$, we have fully estimated $\hat \loadn$. The entire algorithm is described in Algorithm~\ref{alg:lsp}, where the functions LASSO and PROJECTOR refer to a generic version of the Lasso and the projection method, respectively.

\begin{algorithm}[t]
\caption{LassoSimplexProjector (LSP) Algorithm }
\label{alg:lsp}
\begin{algorithmic}[1]
\Require{Estimated TPDM $\hat{\TPDM} \in [0,\infty)^{d \times d}$, number of pure variables $\hat K$, pure variable index set $\hat I$, partition of the pure variable index set $\hat{\bm \Ic} = (\hat I_1, \dots, \hat I_{\hat K})$, LASSO smoothing parameter $\lambda\ge 0$.}
\Ensure{Estimated loading matrix $\hat{\loadn}$.}
\Function{LSP}{$\hat{\TPDM}, \hat K, \hat I, \hat{\bm \Ic},\lambda$}
    \State Initialize $\hat{\loadn} \gets \bm{0}_{d \times \hat{K}}$
    \For{each $a,b \in [\hat{K}]$} \Comment{Estimate $\hat{C}$}
        \State $\hat{C}(a,b) \gets (|\hat{I}_a||\hat{I}_b|)^{-1} \sum_{j \in \hat{I}_a, \ell \in \hat{I}_b}  \hat{\TPDM}(j,\ell)$
    \EndFor
    \For{$a \in [\hat K]$} \Comment{Estimate pure variable rows of $\hat{\loadn}$}
        \For{$j \in \hat I_a$}
            \State $\hat \loadn_{j,\cdot} \gets \bm e_a^\top$ 
        \EndFor
    \EndFor
    \For{each $j \in [d] \setminus \hat{I}$} 
    \Comment{Estimate non-pure variable rows of $\hat{\loadn}$}
    \For{each $a \in [\hat{K}]$}
        \State $\hat{\theta}_{ja} \gets |\hat{I}_a|^{-1} \sum_{\ell \in \hat{I}_a}  \hat{\TPDM}(\ell,j)$; $\hat{\bm \theta}_j \gets (\hat{\theta}_{j1}, \dots, \hat{\theta}_{j\hat K})^\top$
    \EndFor
    \State $\hat{\loadn}_{j \cdot}  \gets$ \Call{LASSO}{$\hat{C}, \hat{\bm \theta}_j, \lambda$}  \Comment{Row estimate based on the Lasso}
    \State $\hat{\loadn}_{j \cdot} \gets \max(\hat{\loadn}_{j \cdot}, 0)$ \Comment{Take coordinate-wise maximum.}
    \State $\hat{\loadn}_{j \cdot} \gets$ \Call{PROJECTOR}{$\hat{\loadn}_{j \cdot}$}     \Comment{Optional: projection on the unit simplex.}
\EndFor
\State \Return $\hat{\loadn}$
\EndFunction
\end{algorithmic}
\end{algorithm}

\subsection{Estimation of the latent spectral measure}
\label{subsec:estimation-latent-spectral-measure}

As indicated in Step~(7) of Remark~\ref{rem:reconstructing-parameters-tail-dependence}, estimation of $\loadn$ entails a corresponding estimation of $\Psifactor$. The next result establishes the theoretical basis for this approach.

\begin{lemma} 
\label{lem:spectral-measure-of-pure-variables}
Suppose $\bX$ has STDF $L=L_{K, \loadn, \Psifactor}$ from \eqref{eq:stdf-factor} with $(K,\loadn, \Psifactor) \in \paraspacetaildep$. 
Let $\bm Z = (Z_a)_{a \in [K]}$ have coordinates $Z_a = |I_a|^{-1} \sum_{j \in I_a} Y_j$, with $I_a(A)$ from \eqref{eq:pure-variables} and with $Y_j = 1/(1-F_j(X_j))$. 
Then $\bm Z$ is regularly varying of order 1, and the spectral measure of $\bZ$ with respect to the 1-norm is given by $\Psifactor$. In particular,
\begin{equation} 
\label{eq:spectral_measure_pure-variables}
\mathbb{P} \Big( \frac{\bZ}{\|\bZ\|_1} \in \, \cdot \,\, \Big| \, \|\bZ\|_1 > x \Big) \rightsquigarrow \Psifactor(\cdot) \quad \textrm{ in } \simplexpos{K-1}(\| \cdot \|_1) \quad \text{as }  x \to \infty.
\end{equation} 
\end{lemma}

The weak convergence in \eqref{eq:spectral_measure_pure-variables} motivates the following estimator of $\Psifactor$. 
Recall the pseudo-obser\-vations $\hat Y_{ij} = \{1-\hat F_{nj}(X_{ij})\}^{-1}$ from \eqref{eq:empirical_tpdm}. 
Using the estimated index sets $\hat I_a$ from Algorithm~\ref{alg:purevar}, define aggregated pseudo-observations $\hat{\bm Z}_i \in \R^{\hat K}$ with components
\begin{align} \label{eq:def-hat-z}
\hat Z_{ia} = \frac{1}{|\hat I_a|} \sum_{j \in \hat I_a} \hat Y_{ij}, 
\qquad i \in [n], a \in [\hat K].
\end{align}
Let $\hat S_i = \|\hat{\bm Z}_i\|_1$ and denote by 
$\hat S_{1:n} \le \dots \le \hat S_{n:n}$ 
the order statistics of $\hat S_1,\dots,\hat S_n$. 
For a large threshold parameter $k' \in [n]$ (typically, we will choose $k'=k$ with $k$ from \eqref{eq:empirical_tpdm}), we define
\begin{align}\label{eq:empirical-spectral-measure}
\hat \Psifactor_{k'}(\cdot)
= \frac{1}{k'} \sum_{i=1}^n 
\bm 1\Big( 
\frac{\hat{\bm Z}_i}{\hat S_i} \in \,\cdot\, ,\ 
\hat S_i > \hat S_{n-k':n}
\Big).
\end{align}
By construction, $\hat \Psifactor_{k'}$ is a variant of the empirical spectral measure introduced in \cite{EinmahldeHaanPiterbarg2001,EinmahlSegers2009}.
Alternatively, one may assume that $\Psifactor$ belongs to a parametric family and estimate the corresponding parameters, for instance by likelihood-based or moment-based methods applied to the sample $\{ \hat{\bm Z}_i / \hat S_i : i \in [n] \text{ and } \hat S_i > \hat S_{n-k':n}\}$; see \cite{EinmahlKrajinaSegers2012}.

\subsection{Estimation of derived tail parameters}
\label{subsec:estimation-derived-tail-parameters}

We describe how estimation of $K$, $\loadn$, and $\Psifactor$ can be combined with marginal estimation to obtain estimators of tail parameters such as those introduced in \eqref{eq:tail-parameter-introduction}. 
Throughout, we assume that $\bm X$ is a $d$-dimensional observable random vector whose tail dependence is of linear latent $K$-factor form as in Definition~\ref{def:latent-linear-factor-tail-dependence}. Moreover, for each $j \in [d]$, the marginal distribution of $X_j$ is assumed to lie in the max-domain of attraction of an extreme-value distribution with index $\xi_j \in \R$. We observe an i.i.d.\ sample $\bm X_1, \dots, \bm X_n$ from $\bm X$.

Our goal is to estimate tail probabilities of the following form. Let $\mathcal{J}$ be a collection of subsets of $[d]$, for example, $\mathcal{J} = \{ J \subset [d] : |J| \ge 0.8d \}$, 
the class of subsets containing at least $80\%$ of the indices. For a vector of large deterministic thresholds $\bm x = (x_1, \dots, x_d)^\top$, define
\begin{align} \label{eq:derived-tail-parameter}
p(\bm x, \mathcal{J})
=
\mathbb{P}\Big( \exists\, J \in \mathcal{J} \text{ such that } X_j > x_j \ \forall j \in J \Big).
\end{align}
Under our above model assumption on $\bm X$, we can approximately express $p(\bm x, \mathcal{J})$ as a function of $K, \loadn, \Psifactor$ and $q_1(x_1), \dots, q_d(x_d))$, where $q_j(x_j) = 1-F_j(x_j)$ is a small tail probability. Indeed, by Proposition~\ref{prop:survival_stdf_lfm_tail_dep}, for $t$ sufficiently large and $q(x_j)$ sufficiently small (of the order $1/t$), we have
\begin{align} \label{eq:tail-parameter}
p(\bm x, \mathcal{J})
&= \nonumber
\mathbb{P}\Big\{ \exists J \in \mathcal{J} : F_j(X_j) > 1 - \frac{t q_j(x_j)}{t} \ \forall j \in J \Big\} 
\\&\approx 
\frac{1}{t} \, R^\cup_{\mathcal{J}}(t \, \bm{q}(\bm x)) 
= 
K \int_{\simplexpos{K-1}} 
\bigvee_{J \in \mathcal{J}} \bigwedge_{j \in J} 
\Big( \sum_{a \in [K]} \loadn_{ja} z_a \, q_j(x_j) \Big) \Psifactor(\diff \bz),
\end{align}
where $\bm q (\bm x)=(q_1(x_1), \dots, q_d(x_d))$. 

The vector $\bm q(\bm x)$ of tail exceedance probabilities can be estimated using standard methods from univariate extreme value statistics. More specifically, let 
\[
\mathrm{GPD}(y; \xi, \sigma) =
\begin{cases}
    1 - \left( 1 + \xi \frac{y}{\sigma} \right)_+^{-1/\xi}, & \xi \neq 0, \\
    1 - \exp\left(-\frac{y}{\sigma}\right), & \xi = 0,
\end{cases} \qquad (y>0)
\]
(with $a_+ = \max(a,0)$)
denote the cdf of the generalized Pareto distribution with shape parameter $\xi \in \R$ and scale parameter $\sigma>0$. For each $j \in [d]$, let $k = k(j)\in [n]$ be a threshold parameter (typically, $k$ is around $5\%$ of $n$ or smaller), and define $\hat u_j = X_{j,n-k:n}$.
We then fit the GPD to the sample $\{ X_{j,n-k+1:n} - \hat u_j, \dots, X_{j,n:n} - \hat u_j \}$ of threshold exceedances using maximum likelihood estimation \citep[Section 3.4]{DeHaanFerreira2006}, which yields estimated values $(\hat \xi_j, \hat \sigma_j)$. Finally, the tail probabilities are estimated using the semiparametric estimator
\begin{align} \label{eq:empirical-tail-exceedance-probabilities}
\hat{q}_j(x_j) =
\begin{cases}
    1 - \hat{F}_{nj}(x_j), & x_j \leq \hat u_j, \\
    \frac{k}n \cdot \left(1 - \mathrm{GPD}(x_j - \hat u_j; \hat{\xi}_j, \hat{\sigma}_j)\right), & x_j > \hat u_j,
\end{cases}
\end{align}
see \cite[Section 4.4]{DeHaanFerreira2006}. 

The marginal estimators from \eqref{eq:empirical-tail-exceedance-probabilities} can now be combined with the estimators from Sections~\ref{subsec:estimating-K-and-pure-variables}-\ref{subsec:estimation-latent-spectral-measure} to estimate the right-hand side of \eqref{eq:tail-parameter} using the plug-in-approach:
\begin{align} \label{eq:estimator-tail-parameter}
\hat p(\bm x, \mathcal{J}) 
:=&\, \nonumber
\hat K \int_{\simplexpos{\hat K-1}} 
\bigvee_{J \in \mathcal{J}} \bigwedge_{j \in J} 
\Big( \sum_{a \in [\hat K]} \hat \loadn_{ja} z_a \, \hat q_j(x_j) \Big) \hat \Psifactor_{k'}(\diff \bz)
\\=&\,
\frac{\hat{K}}{k'} 
\sum_{i \in [n]} 
\bigvee_{J \in \mathcal{J}} \bigwedge_{j \in J} 
\Big( \sum_{a \in [\hat K]} \frac{\hat{\loadn}_{ja} \hat Z_{ia} \hat{q}_j(x_j)}{\hat S_i} \Big)
\bm{1}(\hat S_i > \hat S_{n-k':n}),
\end{align}
where $\hat S_i  =\| \hat {\bm Z}_i\|_1$ with $\hat {\bm Z_i}$ from \eqref{eq:def-hat-z}. 
In applications of interest, the collection $\mathcal J$ is often exponentially large in $d$, which at first sight renders the estimator computationally infeasible. 
However, the key quantities of the form
$\bigvee_{J \in \mathcal{J}} \bigwedge_{j \in J} \hat v_{j}$ with $\hat v_{j} = \sum_{a \in [\hat K]}\hat{\loadn}_{ja} \hat Z_{ia} \hat{q}_j(x_j)/\hat S_i$
can often be evaluated efficiently.
For example, if 
$
\mathcal J 
= 
\{ J \subset [d] : |J|/d \ge \alpha \}$ for some 
$\alpha \in [d^{-1},1]$,
then
$\bigvee_{J \in \mathcal{J}_\alpha} \bigwedge_{j \in J} \hat v_j = \hat v_{d-\ell+1:d}$, where $\ell = \lfloor \alpha d\rfloor$.
Further computational simplifications for other choices of $\mathcal J$ are discussed in Section~\ref{sec:case-study}.

\subsection{Hyperparameter selection and model validation}
\label{subsec:hyperparameter-and-model-validation}

The estimators and algorithms introduced in the previous sections depend on several tuning parameters that must be specified by the practitioner. 

First, we encounter the classical threshold selection problem in extreme-value analysis, which in our setting concerns the parameters $k$, $k'$, and $k(j)$ governing the effective tail sample size. We do not elaborate on this issue further; in practice, these parameters are typically chosen as a small fraction of the total sample size, for example around $5\%$ of $n$, which is also the choice adopted here. In particular, we typically set $k=k'=k(j)$.

The PureVar Algorithm~\ref{alg:purevar} and the LSP Algorithm~\ref{alg:lsp} for estimating $\loadn$ involve several additional design choices: 
(i) the selection of a norm on $\R^d$ for computing the empirical TPDM, 
(ii) a tuning parameter $\kappa>0$ for identifying pure variables, 
(iii) a sparsity parameter $\lambda>0$ for regularizing the loading matrix $\loadn$, and 
(iv) the choice between the projected and non-projected estimator in Step~17 of Algorithm~\ref{alg:lsp}. Regarding the norm on $\R^d$, we found the maximum norm to yield stable and robust results. For the remaining tuning parameters, we propose the following data-adaptive strategy. First, select grids $S_\kappa = \{\kappa_1,\dots,\kappa_{L(\kappa)}\}$ and $S_\lambda = \{\lambda_1,\dots,\lambda_{L(\lambda)}\}$
with $\lambda_1=0$, and compute estimators of $K$, $\loadn$, and $\Psifactor$ for each triple
$
(\kappa,\lambda,\mathrm{proj})
\in
S_\kappa \times S_\lambda \times \{\mathrm{True},\mathrm{False}\},
$
where $\mathrm{proj}$ indicates whether the projected estimator is used in Step~17 of Algorithm~\ref{alg:lsp}. For each configuration, pairwise tail correlations can be estimated as in Section~\ref{subsec:estimation-derived-tail-parameters}. Specifically, for $j\neq \ell$, recalling the formula for the tail correlations from \eqref{eq:model-tail-correlation}, we define 
\[
\tilde \chi_{\kappa,\lambda,\mathrm{proj}}(j,\ell)
=
\min\bigg\{
1,\,
\frac{\hat K}{k'}
\sum_{i \in [n]}
\Big[
\Big(
\sum_{a \in [\hat K]}
\frac{\hat{\loadn}_{ja}\hat Z_{ia}}{\hat S_i}
\Big)
\wedge
\Big(
\sum_{a \in [\hat K]}
\frac{\hat{\loadn}_{\ell a}\hat Z_{ia}}{\hat S_i}
\Big)
\Big]
\bm 1(\hat S_i>\hat S_{n-k':n})
\bigg\}.
\]
These values can be compared with the empirical pairwise tail correlation
\[
\hat\chi(j,\ell)
=
\frac{1}{k'}
\sum_{i \in [n]}
\bm 1\big(
R_{ij}>n-k',\,
R_{i\ell}>n-k'
\big),
\]
where $R_{ij}$ denotes the rank of $X_{ij}$ among $X_{1j},\dots,X_{nj}$. This yields the point cloud $\{(\tilde \chi_{\kappa, \lambda,\mathrm{proj}}(j,\ell),  \allowbreak \hat{\chi}(j,\ell)): j,\ell \in [d] \text{ with } j \ne \ell\}$, which should concentrate around the identity line. We therefore select $(\kappa,\lambda,\mathrm{proj})$ by maximizing the coefficient of determination $R^2$ between $\tilde \chi_{\kappa,\lambda,\mathrm{proj}}(j,\ell)$ and $\hat\chi(j,\ell)$ over the grid
$S_\kappa \times S_\lambda \times \{\mathrm{True},\mathrm{False}\}$. 

Finally, the point cloud corresponding to the selected tuning parameters serves as a graphical diagnostic for assessing model adequacy: it should closely cluster around the identity line. Based on simulation experiments, a respective $R^2$ value exceeding 0.8 typically indicates a satisfactory model fit.

\section{Case study}
\label{sec:case-study}

In this section, we illustrate our new methods with a case study related to wind power production in Germany. The focus lies on the spatial occurrence of winter low-wind (``Flaute'') events, which are particularly critical in an electricity system with a high share of weather-dependent renewables. In Germany, wind energy accounts for a substantial fraction of total generation, especially during winter months when electricity demand is high and solar production is limited. A spatial analysis of low-wind conditions therefore provides important insights into system vulnerability and the probability of widespread shortfalls in wind power generation \citep{ohlendorf2020frequency}.

As a basis for the analysis, we use the renewable\_power\_plants dataset from Open Power System Data (OPSD), which contains plant-level information on renewable installations across Germany at the NUTS-3 level.\footnote{\url{https://data.open-power-system-data.org/renewable_power_plants/}} Wind capacity is heavily concentrated in the northern part of Germany, with especially high densities along the western coast of Schleswig-Holstein. We therefore focus on two coastal NUTS-3 regions with particularly large installed capacity. To obtain precise turbine locations, we complement the OPSD data with the Windkraftanlagen dataset from the Schleswig-Holstein Open Data Portal.\footnote{\url{https://opendata.schleswig-holstein.de/dataset/windkraftanlagen2}} In total, 1,917 turbines in the selected area are considered, including operating, approved, and planned installations, enabling a comprehensive assessment of potential risk scenarios; see Panel (a) in Figure~\ref{fig:wind_combined}.

\begin{figure}[t!]
    \centering    
    \begin{subfigure}[b]{0.42\textwidth}
        \centering
        \includegraphics[width=.95\textwidth, keepaspectratio]{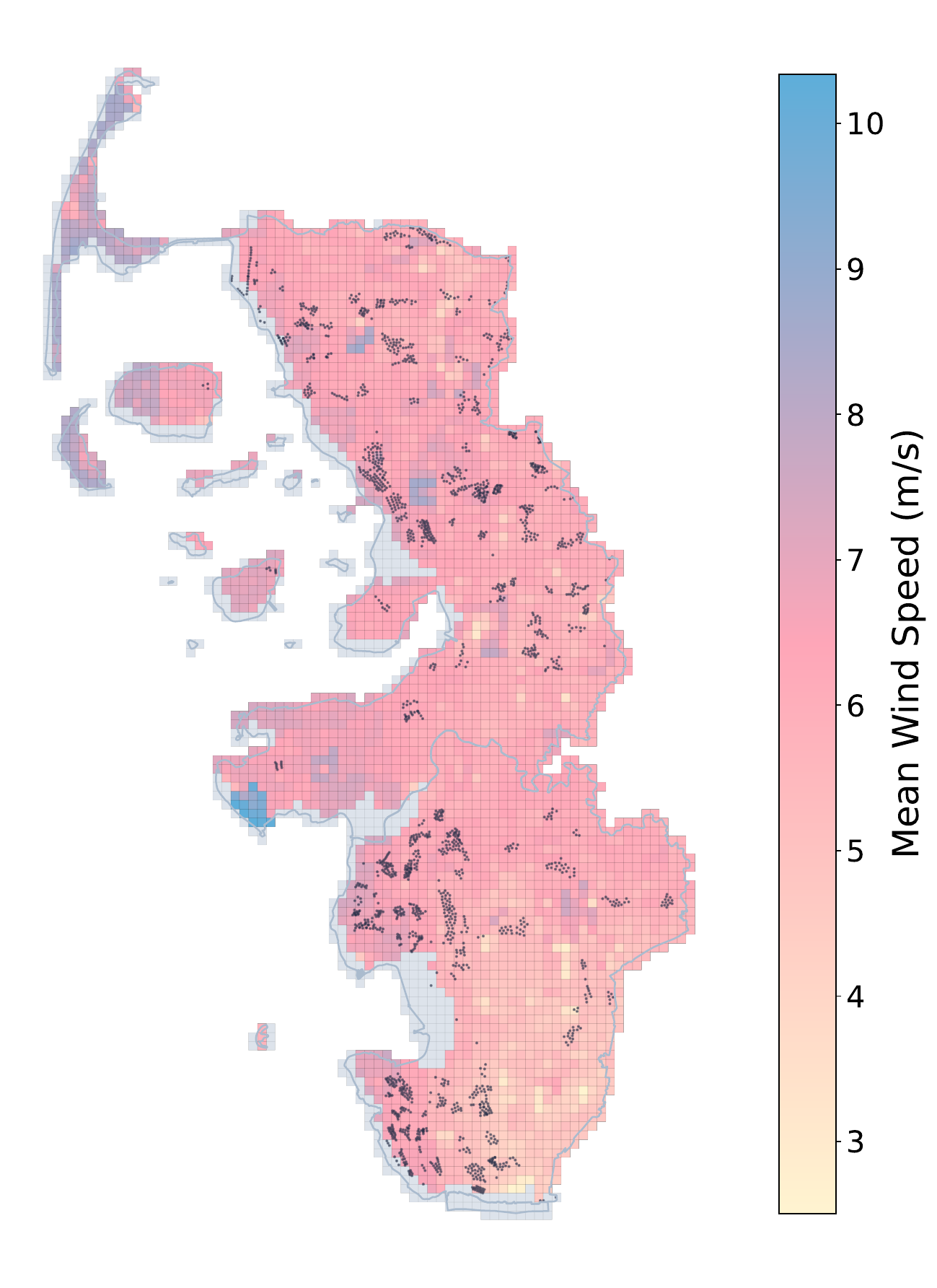}
        \caption{Wind turbine locations.}
        \label{fig:wind_turbine_wescoast}
    \end{subfigure}
    \hfill
    \begin{subfigure}[b]{0.57\textwidth}
\centering
\begin{tikzpicture}
\begin{axis}[
    width=7cm,
    height=8.5cm, 
    font=\small,
    xmin=0, xmax=30,
    ymin=0, ymax=1.2,
    xlabel={Wind speed (m/s)},
    ylabel={Capacity factor ratio},
    xtick={0,5,10,15,20,25,30},
    ytick={0,0.2,0.4,0.6,0.8,1.0},
    axis lines=left,
    samples=400,
    thick,
]

% --- Parameters ---
\def\vci{3}    % cut-in
\def\vr{14}    % rated speed
\def\vco{25}   % cut-out

% --- Vertical shading (full height [0,1]) ---

% Non-operational: below cut-in
\addplot[
    draw=none,
    fill=red,
    fill opacity=0.12
] coordinates {
    (0,0) (\vci,0) (\vci,1) (0,1)
};

% Operational range
\addplot[
    draw=none,
    fill=green!60!black,
    fill opacity=0.12
] coordinates {
    (\vci,0) (\vco,0) (\vco,1) (\vci,1)
};

% Non-operational: above cut-out
\addplot[
    draw=none,
    fill=red,
    fill opacity=0.12
] coordinates {
    (\vco,0) (30,0) (30,1) (\vco,1)
};

% --- Power curve ---
\addplot[
    domain=0:30,
    blue,
    very thick
] expression {
    (x<\vci) * 0
  + (x>=\vci && x<\vr) * ( pow((x-\vci)/(\vr-\vci),2) )
  + (x>=\vr && x<=\vco) * 1
  + (x>\vco) * 0
};

% Cut-in / cut-out markers
\addplot[dashed] coordinates {(\vci,0) (\vci,1)};
\addplot[dashed] coordinates {(\vr,0) (\vr,1)};
\addplot[dashed] coordinates {(\vco,0) (\vco,1)};

\node[anchor=south, align=center] at (axis cs:{\vci+0.5},1.02) {Cut-in: \\ $\vci$ m/s};
\node[anchor=south, align=center] at (axis cs:{\vr+0.1},1.02) {Rated speed: \\ $\vr$ m/s};
\node[anchor=south, align=center] at (axis cs:\vco,1.02) {Cut-out: \\ $\vco$ m/s};

\end{axis}
\end{tikzpicture}
\caption{Typical wind turbine power curve.}
	\label{fig:wind_turbine}
    \end{subfigure}
    \caption{(a) Detailed turbine locations in Schleswig-Holstein (black dots) with January 2023 mean wind speeds from \textit{HOSTRADA} (grid cells of size 1km x 1km). Pixels outside \textit{HOSTRADA} resolution data are shown in grey. (b) A typical wind turbine power curve. Turbines produce no power below the cut-in speed (3 m/s), operate at rated capacity between 14 and 25 m/s, and shut down above the cut-out speed (25 m/s).}
    \label{fig:wind_combined}
\end{figure}
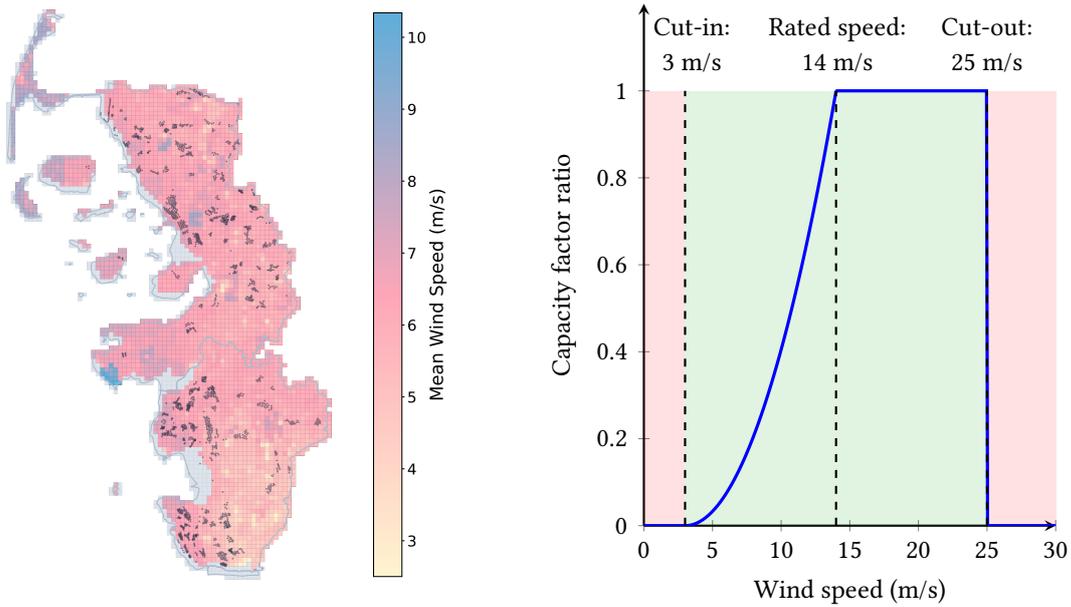

To translate wind speeds into power production, we rely on standard wind turbine power curves, which describe the nonlinear relationship between wind speed and electrical output. In particular, turbines generate no power below a cut-in speed (around $3ms^{-1}$), reach rated output above a certain threshold, and shut down at very high wind speeds (around $25ms^{-1}$); see Panel (b) in Figure~\ref{fig:wind_combined} for an illustration. Since our interest lies in system-wide shortfalls, we focus on the lower tail of the wind speed distribution and identify days on which turbines operate below the cut-in speed, implying zero production. Wind speeds are obtained from the ERA5 reanalysis (available at a spatial resolution of roughly 27.5km, and at 10m and 100m above surface, see \citealp{hersbach2020era5}) and the high-resolution DWD HOSTRADA dataset (spatial resolution of 1km, and at 10m above surface only).\footnote{\url{https://opendata.dwd.de/climate_environment/CDC/grids_germany/hourly/hostrada/}} HOSTRADA 10m winds are extrapolated to a 100 m hub height using the power law \citep{cai2021wind, jourdier2020evaluation, schallenberg2013methodological}, with the Hellmann exponent derived from ERA5 and averaged by hour of day and month to capture diurnal variation; details are provided in Section~\ref{sec:details-on-case-study}.
For each HOSTRADA grid cell containing at least one wind turbine, we extract daily winter maxima, yielding $n=2{,}708$ winter days over the past 30 years and $d=563$ grid cells. Each grid cell $j$ is assigned a weight $c_j$, defined as the proportion of total capacity of the wind power plants in that cell relative to the  total capacity in the entire study region (which is roughly 6 gW),
so that $\sum_{j\in[d]} c_j = 1$.
On this basis, our primary quantities of interest are the probabilities that wind turbines accounting for at least an $\alpha\cdot 100\%$-fraction of the total capacity (with $\alpha \in \{0.1,0.2,\dots,0.9\}$) are located in areas where the wind speed simultaneously remains below the cut-in speed throughout an entire winter day, implying zero aggregate wind power production for those turbines.

More precisely, let $W_{ij}$ denote the (daily maximal) wind speed at 100\,m above surface at time $i \in [n]$ and pixel $j \in [d]$, where $n = 2708$ and $d = 563$. 
For each time point $i$, define the $d$-dimensional vector $\bm W_i = (W_{i1}, \dots, W_{id})^\top$, 
and assume that $\bm W_1, \dots, \bm W_n$ are identically distributed and approximately independent (in the tail); the independence assumption has been validated empirically. For $\alpha\in(0,1)$, let
$\mathcal{J}_\alpha = \{ J \subset [d] : \sum_{j \in [J]} c_j \ge \alpha  \}$
denote the collection of all subregions of the study area such that the total capacity of the wind power plants in the subregion accounts for at least an $\alpha\cdot 100\%$-fraction of the total capacity.
Given a small threshold vector $\bm w = (w_j)_{j \in [d]} \in (0,\infty)^d$, our quantity of interest is
\begin{equation}
\label{eq:target_parameter_wind}
p_\alpha(\bm w) 
= 
\mathbb P \Big( \exists\, J \in \mathcal{J}_\alpha 
\ \text{s.t.}\ W_j < w_j \ \forall j \in J \Big).
\end{equation}
We are particularly interested in the homogeneous case $w_j = 3$ for all $j \in [d]$, which corresponds to the probability of the event described at the end of the last paragraph.

While our methodology focuses on right-tail events, the target parameter from \eqref{eq:target_parameter_wind} concerns left-tail events. We therefore consider the reciprocal transformation $x_j = w_j^{-1}$ and $X_{j} = W_{j}^{-1}$
for $j \in [d]$, such that
\begin{align*}
p_\alpha (\bm w)
&= 
\mathbb{P}\Big\{ \exists J \in \mathcal{J}_\alpha : X_j > x_j \ \forall j \in J \Big\} 
\end{align*}
is exactly of the form in \eqref{eq:derived-tail-parameter}. The derivations in Section~\ref{subsec:estimation-derived-tail-parameters} show that, under the assumption that the tail dependence of $\bm{X}$ is of latent linear $K$-factor form as introduced in Definition~\ref{def:latent-linear-factor-tail-dependence},  
and provided that the exceedance probability $q_j(x_j)=\Prob(X_j > x_j)$ is sufficiently small,  we have 
\begin{align*}
p_\alpha (\bm w)
&\approx  K \int_{\simplexpos{K-1}} 
\bigvee_{J \in \mathcal{J}_\alpha} \bigwedge_{j \in J} 
\Big( \sum_{a \in [K]} \loadn_{ja} z_a \, q_j(x_j) \Big) \Psifactor(\diff \bz),
\end{align*}
where $\bm q (\bm x)=(q_1(x_1), \dots, q_d(x_d))$; see \eqref{eq:tail-parameter}.
As shown at the end of Section~\ref{subsec:estimation-derived-tail-parameters}, the preceding expression, and hence the corresponding estimator, appears to pose a combinatorial challenge, since it involves terms of the form $\bigvee_{J \in \mathcal{J}} \bigwedge_{j \in J} v_{j}$, where $v_j$ is defined for every $j \in [d]$ and where $|\mathcal J|$ can be extremely large as it scales exponentially with $d$.
This difficulty can be avoided as follows: let $\pi:[d] \to [d]$ be a permutation such that $v_{\pi(1)} \ge v_{\pi(2)} \ge \dots \ge v_{\pi(d)}$, and define the cumulative capacity weights $C_\ell = \sum_{j \in [\ell]} c_{\pi(j)}$ for each $\ell \in [d]$. We then have that $\bigvee_{J \in \mathcal{J}_\alpha} \bigwedge_{j \in J} v_j = v_{\pi(m)}$, where $m:= \min\{ \ell \in [d]: C_\ell \ge \alpha\}$.

Under the additional assumption that the marginal cdfs $F_j$ of $X_j$ lie in the max-domain of attraction of an extreme-value distribution with extreme-value index $\xi_j$, we estimate $p_\alpha(\bm w)$ via the estimator $\hat p_\alpha (\bm w) := \hat p(\bm x, \mathcal J_\alpha)$ from \eqref{eq:estimator-tail-parameter}, using the combinatorial trick explained in the previous paragraph and the hyperparameter selection scheme of Section~\ref{subsec:hyperparameter-and-model-validation}. 
Specifically, we set $k=k'=k(j)=135$ (so that $k/n \approx 5\%$) and use the maximum norm to compute the empirical TPDM. We consider the grids $G_\kappa=\{0.002,0.0025,\dots,0.008\}$ and $G_\lambda=\{0.00001,0.00002,\dots,0.001\}$, yielding $13 \cdot 100 \cdot 2 = 2{,}600$ configurations for $(\kappa,\lambda,\mathrm{proj}) \in G_\kappa \times G_\lambda \times \{\mathrm{True},\mathrm{False}\}$. The best-performing configuration uses the projected estimator with $\kappa^*=0.006$ and $\lambda^*=0.00084$, attaining $R^2=0.7366$. The corresponding comparison of model-implied and empirical tail correlations (see Section~\ref{subsec:hyperparameter-and-model-validation}) is shown in Panel~(a) of Figure~\ref{fig:empirical_fit}. Panel~(b) repeats the same diagnostic with $k'=72$ (i.e., $3\%$ of the sample), for which the coefficient of determination increases, indicating an even closer agreement further out in the tail. The results indicate that the model provides a reasonable fit.

\begin{figure}[t!]
    \centering  
    \begin{minipage}{0.48\textwidth}
        \centering
        \includegraphics[width=\linewidth]{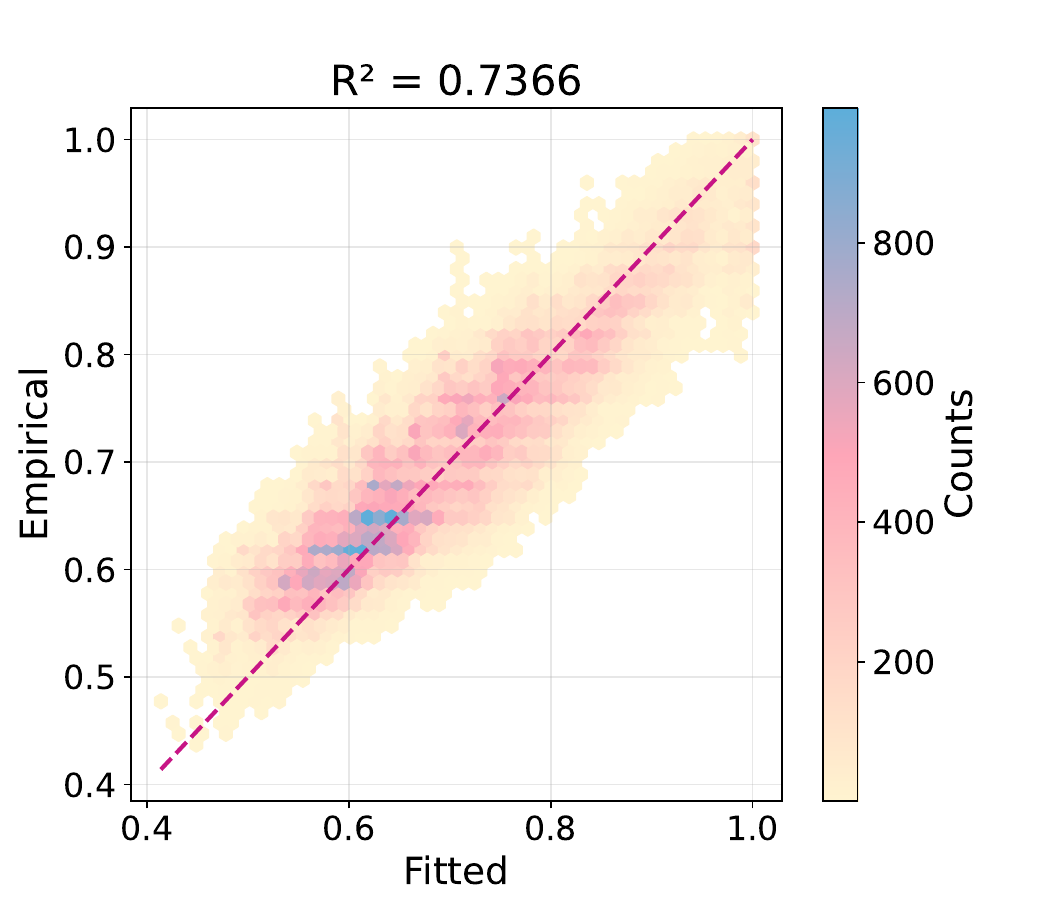}
        \caption*{(a) $k'=135$}
    \end{minipage}
    \hfill
    \begin{minipage}{0.48\textwidth}
        \centering
        \includegraphics[width=\linewidth]{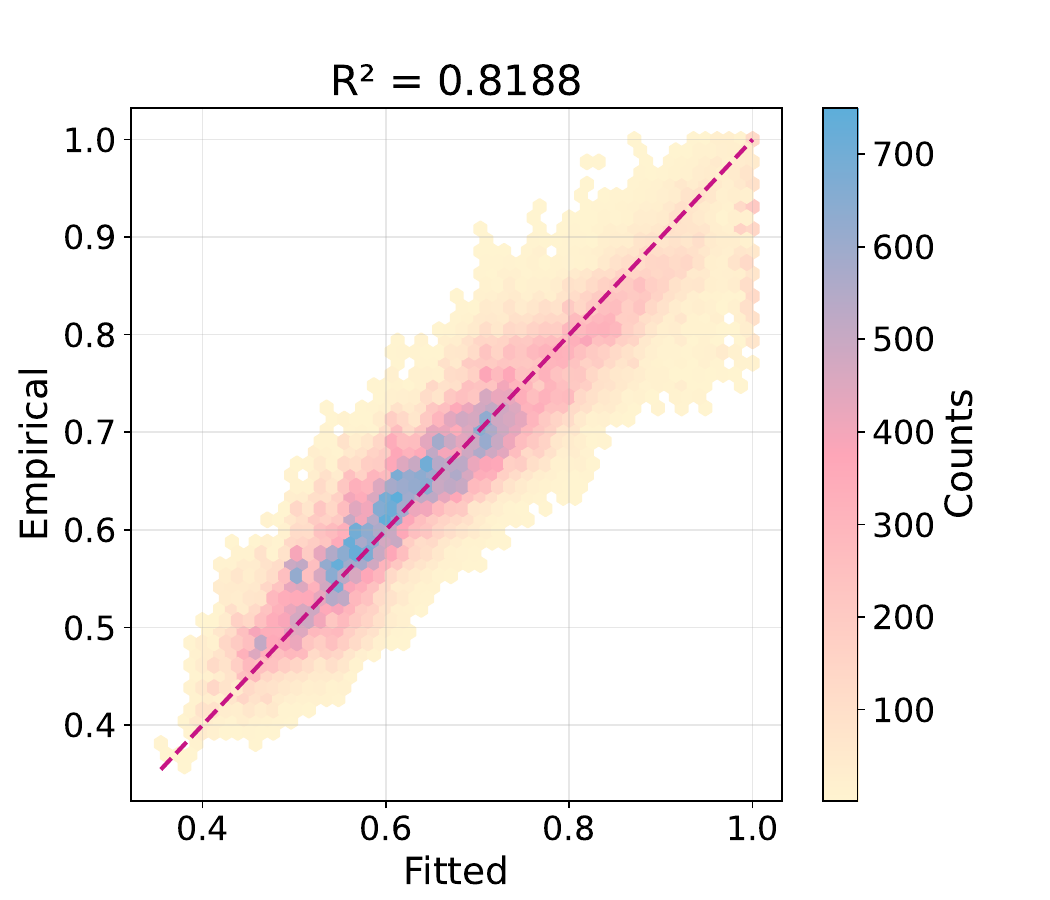}
        \caption*{(b) $k'=72$}
    \end{minipage}
    \caption{Implied tail correlations from the fitted model against empirical tail correlations. Both plots are based on $k=135$ ($5\%$ of the total sample size), while $k'$ is either $k'=135$ in Panel (a), or $k'=72$ ($3\%$ of the total sample size) in Panel (b).} 
    \label{fig:empirical_fit}
\end{figure}

We next summarize and interpret the estimates obtained under the best-performing configuration. First, $\hat{K} = 20$, indicating that the $563$-dimensional tail dependence structure can be represented by only $20$ latent factors. The estimated loading matrix $\hat\loadn$ is sparse: it contains $2{,}251$ strictly positive entries, corresponding to approximately $20\%$ of all entries, and the maximum number of non-zero loadings in any row is $9$. 

The resulting estimates $\hat p_\alpha(\bm w)$ are illustrated in Figure~\ref{fig:p_fitted_p_emp} for $\bm w = (w,\dots,w)^\top$ with $w \in [2,7]$ and $\alpha \in \{0.5,0.7\}$ (Panel (a) and (b)) and for $w\in\{3,5\}$ and $\alpha \in [0.1,0.9]$ (Panel (c) and (d)). In addition to the model-based estimator described above, we report the empirical benchmark
\begin{align} \label{eq:estimator-tail-parameter-empirical-benchmark}
\hat p_\alpha^{\mathrm{emp}}(\bm w) 
= 
\frac{1}{n} \sum_{i \in [n]} \bm{1}\Big(  \sum_{j \in [d]: X_{ij} > x_j} c_j  \ge \alpha  \Big).
\end{align}
By construction, this empirical estimator cannot extrapolate beyond the range of the observed data and is therefore expected to perform poorly near the boundary of that range. Next to the two estimators, we also depict $95\%$ basic bootstrap confidence bands; their precise calculation is explained in Section~\ref{sec:details-on-case-study}.

\begin{figure}[t!]
    \centering
    \begin{minipage}{0.45\textwidth}
        \centering
        \includegraphics[width=\linewidth]{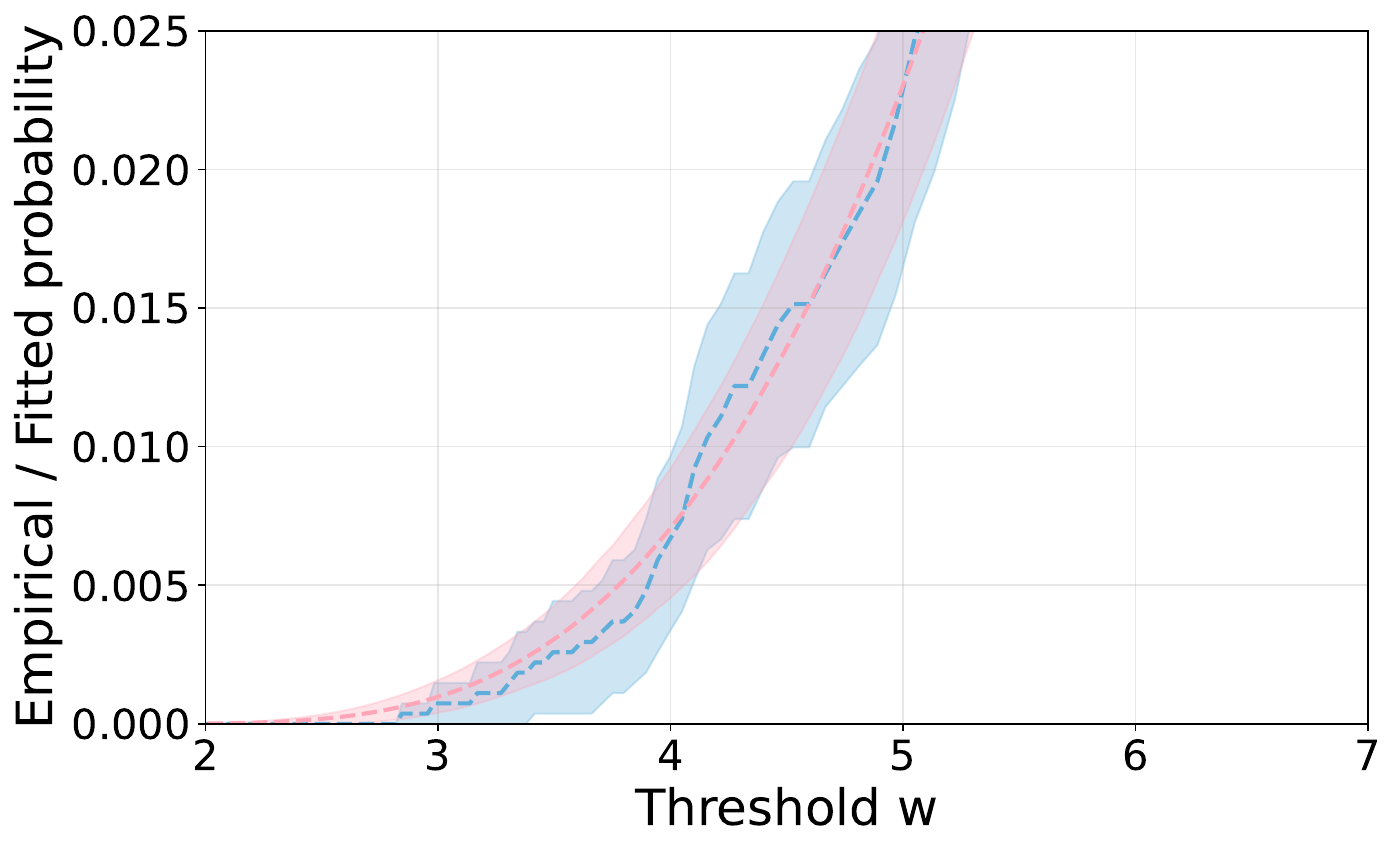}
        \vspace{-.9cm}
        \caption*{(a) $\alpha = 0.5$}
        \vspace{.3cm}
        \includegraphics[width=\linewidth]{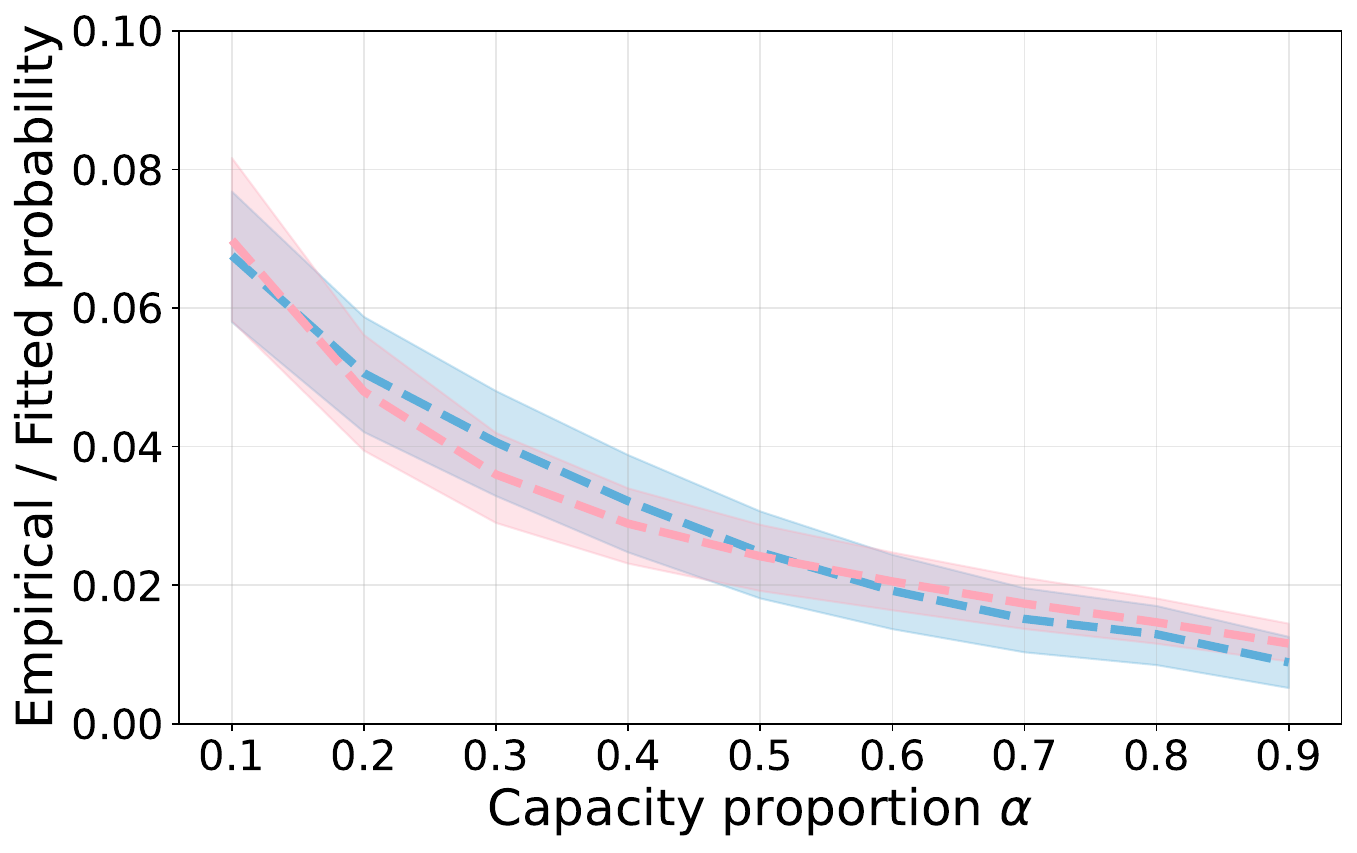}
        \vspace{-.9cm}
        \caption*{(c) $w = 5$}
    \end{minipage}
    \hspace{.2cm}
    \begin{minipage}{0.45\textwidth}
        \centering
        \includegraphics[width=\linewidth]{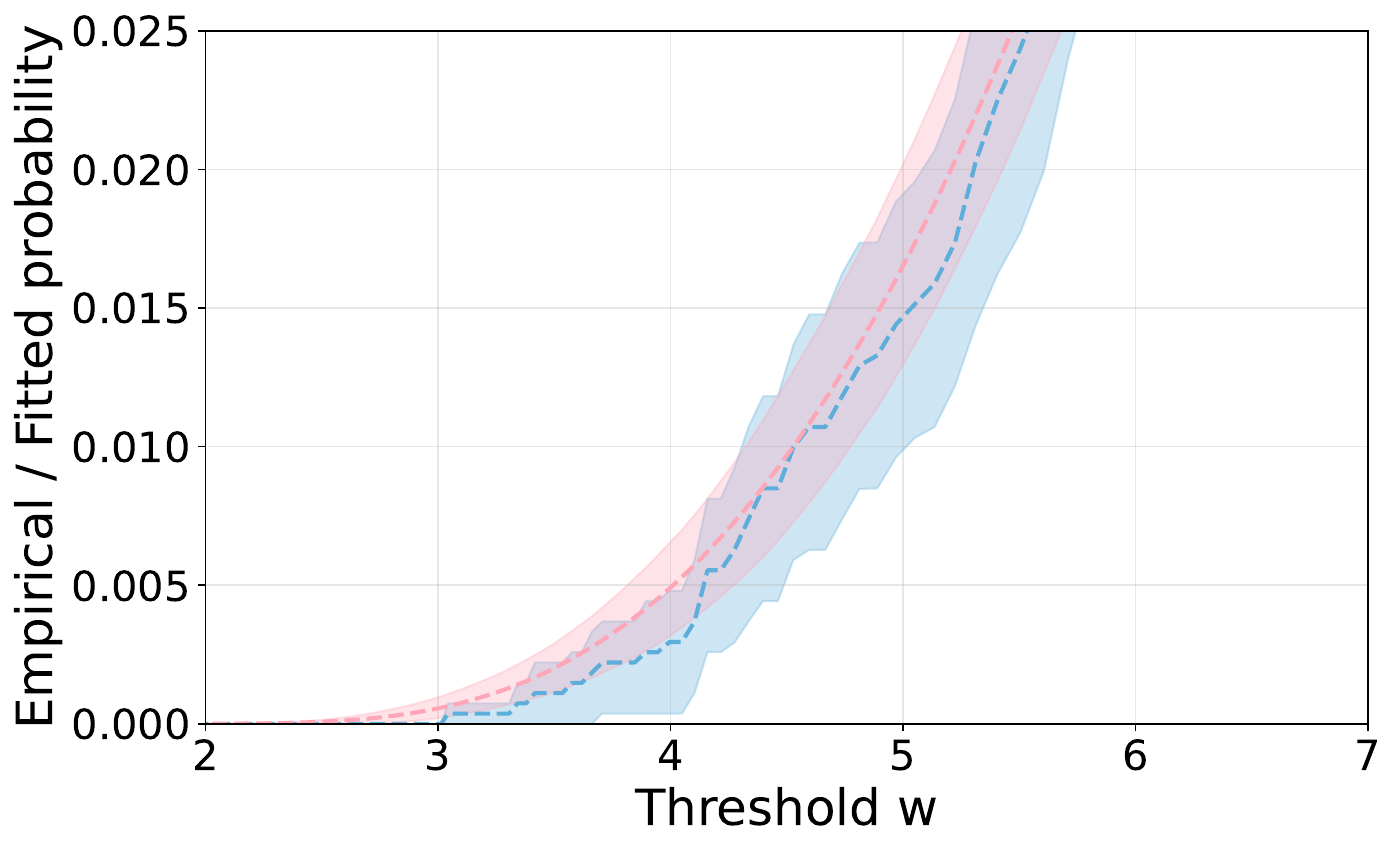}
        \vspace{-.9cm}
        \caption*{(b) $\alpha = 0.7$}
        \vspace{.3cm}
        \includegraphics[width=\linewidth]{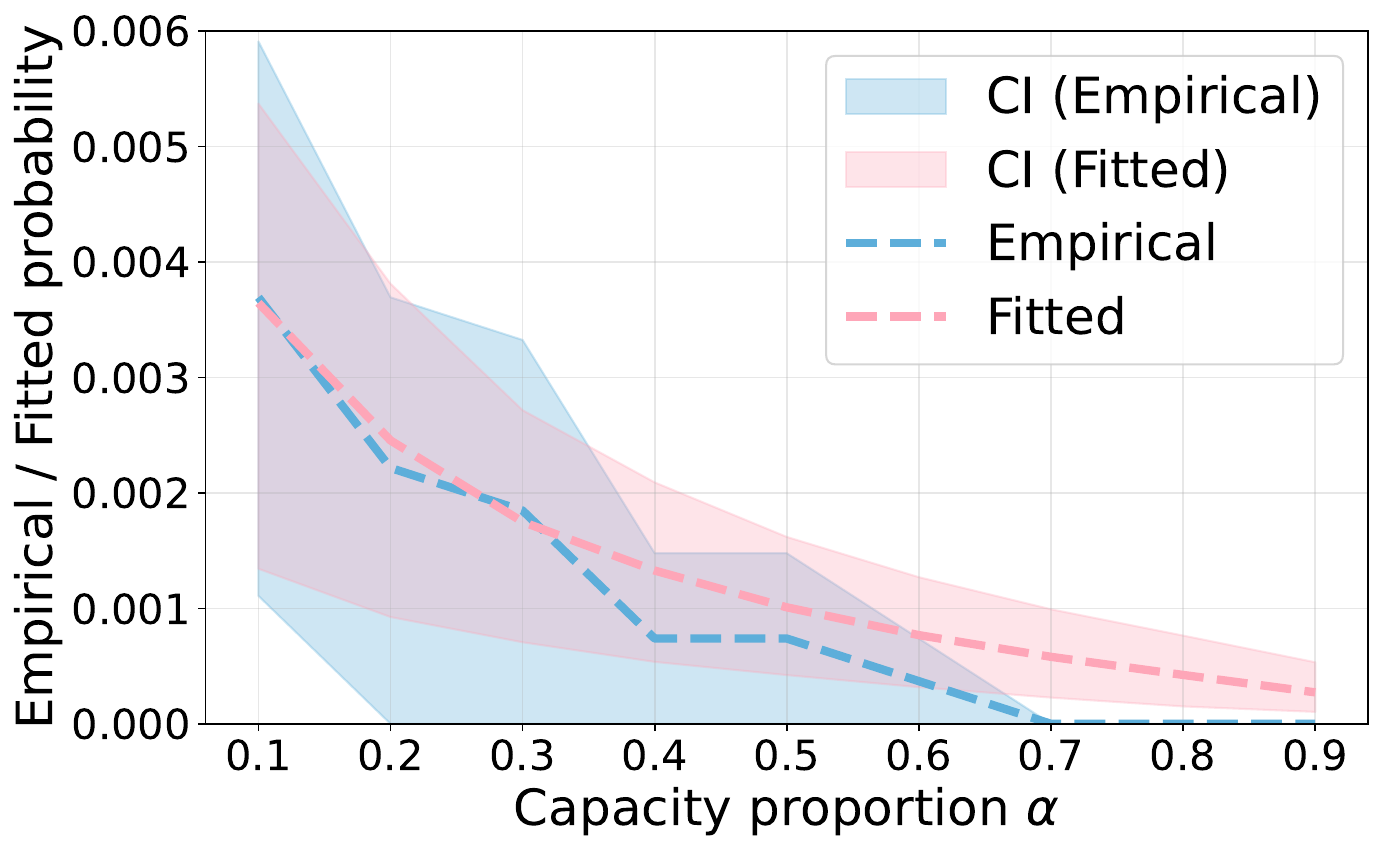}
        \vspace{-.9cm}
        \caption*{(d) $w = 3$}
    \end{minipage}
    \caption{Estimated probabilities $p_\alpha(\bm w)$ from \eqref{eq:target_parameter_wind}, with $\bm w=(w,\dots, w)^\top$, as a function of $w$ and for fixed $\alpha \in \{0.5, 0.7\}$ (Panel (a) and (b)), and as a function of $\alpha$ for fixed $w \in \{3,5\}$ (Panel (c) and (d)). The case $w=3$ corresponds to zero aggregate daily wind power production for turbines accounting for at least $\alpha \cdot 100\%$ of total capacity in the study region.}
    \label{fig:p_fitted_p_emp}
\end{figure}

The results show good agreement between the empirical and model-based estimators when the threshold $w$ is not too small and $\alpha$ is not too large, corresponding to estimated exceedance probabilities of about $0.01$ or higher. In this range, the $95\%$ confidence intervals largely overlap or coincide.
Larger discrepancies emerge for exceedance probabilities below $0.001$, which corresponds to roughly three events among the $n = 2{,}708$ observations. In this regime, empirical estimators are known to be unreliable, whereas the model-implied estimators still permit reasonable extrapolation. In particular, the empirical exceedance probability $\hat p_\alpha^{\mathrm{emp}}(\bm {3})$ is zero for all $\alpha \ge 0.7$, while the model-based estimators remain non-zero in the entire range under consideration. For instance, for $\alpha = 0.9$ we obtain $\hat p_\alpha(\bm {3}) = 0.000272$, meaning that zero aggregate daily wind power production for turbines accounting for at least $90\%$ of total capacity is expected to occur about once every $1/\hat p_\alpha(\bm {3}) \approx 3{,}676$ winter days (for the current wind park).

\section{Conclusion and Outlook}
\label{sec:conclusion}

In summary, we have introduced a latent linear factor framework for multivariate extremes that yields an explicit and interpretable form of dimension reduction. The model admits a clear geometric characterization, allows for sparse loading structures, and is identifiable under mild conditions. A constructive recovery procedure based on a margin-free tail pairwise dependence matrix enables practical estimation and facilitates applications in high-dimensional settings, as illustrated by a case-study on wind energy production.

Several directions for future research arise from our work. A natural next step is the development of flexible parametric models for the latent spectral measure, enabling more structured and parsimonious inference in applications. From a theoretical perspective, establishing statistical guarantees for the proposed recovery and estimation procedures remains an important challenge. In particular, deriving consistency and convergence results for the rank-based empirical tail pairwise dependence matrix, which underpins our methodology, is a key open problem.

\section*{Declarations}

\begin{itemize}
\item\textbf{Ethical Approval.} Not applicable.
\item\textbf{Availability of supporting data. } Code for the case study is available in this Github repository \url{https://github.com/Aleboul/dim_ref_llfm}. The datasets used in Section~\ref{sec:case-study} are also publicly available on the websites mentioned in the main text.
\item\textbf{Competing interests.} The authors declare that they have no conflict of interest.
\item\textbf{Funding.} Both authors were supported by the Deutsche For\-schungsgemeinschaft (DFG, German Research Foundation; Project-ID 520388526;  TRR 391:  Spatio-temporal Statistics for the Transition of Energy and Transport) which is gratefully acknowledged. Calculations for this publication were performed on the HPC cluster Elysium of the Ruhr University Bochum, subsidised by the DFG (INST 213/1055-1).
\item \textbf{Use of AI.} During the preparation of this manuscript, the authors used ChatGPT 5.2 to improve the language and clarity of the text and to improve the figures. The tool also suggested the construction of the norm $\|\cdot\|_v$ in the proof of Theorem~\ref{thm:identifieability-tail-dependence}. The authors carefully reviewed and revised all outputs and take full responsibility for the content of the publication.
\end{itemize}

\appendix

\section{Latent linear factor modeling for regularly varying random vectors}
\label{sec:latent-linear-factor-models-regular-variation}

The developments in Sections~\ref{sec:latent-linear-factor-models} and \ref{sec:estimating-latent-linear-factor-models} were based on the assumption that the random vector $\bY$ obtained by marginal standardization of an observable vector $\bX$ to unit Pareto distributions is regularly varying, and that the spectral measure of $\bY$ (i.e., the spectral dependence measure of $\bX$) is of a certain latent linear factor form. In this section, we adopt a different perspective, which is more general from a probabilistic perspective, but more restrictive from a statistical perspective:
\begin{compactitem}
    \item On the probabilistic side, we do not restrict ourselves to spectral measure models that arise for regularly varying random vectors whose marginal distributions are fixed and all equal to the unit Pareto distribution. Instead, this section will be based on the assumption that the observable $\bX$ itself is regularly varying (of general index $\alpha>0$, and with otherwise arbitrary marginal distributions), and that the spectral measure $\Phi_\bX$ of $\bX$ is of latent linear factor form. Unlike the spectral dependence measure $\Psi_\bX$, the spectral measure $\Phi_\bX$ is not needed to satisfy any moment constraints, and we also have to take care of the additional parameter $\alpha$ for model identifiability.
    \item On the statistical side, the assumption that an observable random vector $\bX$ is regularly varying is very restrictive, in particular in high dimensions. Indeed, non-trivial and meaningful spectral measures will only be obtained if all (or at least a large proportion) of the coordinates $X_j$ of $\bX$ have the same tail index $\alpha_j=\alpha$. This can be the case in some financial applications \citep{kiriliouk2022estimating}, but in general, is appears very restrictive. 
\end{compactitem}
In view of the larger mathematical generality, all results obtained in Section~\ref{sec:latent-linear-factor-models} for spectral dependence measures will more or less be immediate corollaries of the results in this section.

As in Section~\ref{sec:latent-linear-factor-models}, we start from a motivating model that is of plain linear factor form. Throughout this section, unless stated otherwise, $\R^d$ will be equipped with some arbitrary fixed norm $\| \cdot \|$, and $\R^K$ (with $K\in[d]$) will be equipped with some fixed norm $| \cdot |$. The respective unit spheres are denoted by $\simplex{d-1}=\simplex{d-1}(\|\cdot\|)=\{ \bx \in \R^d: \| \bx\| = 1\}$ and $\simplex{K-1}=\simplex{K-1}(|\cdot|)=\{ \bz \in \R^K: | \bz| = 1\}$.

\begin{proposition}
\label{prop:genuine-linear-factor-dependence-regular-variation}
    Consider the random vector
    \[
    \bX := \load \bZ + \bE,
    \]
    where $\load \in \R^{d \times K}$ is a deterministic full column rank matrix,
    where $\bZ \in \R^K$ is regularly varying with index $\alpha>0$ and spectral measure $\Phifactor$, and where $\bE \in \R^d$ is a light-tailed noise in the sense that $\Prob(\|\bE\| > x)=o(x^{-\alpha})$ for $x \to \infty$.
    Then $\bX$ is regularly varying with index $\alpha$ and with spectral measure $\Phi_\bX$ given by $\Phi_\bX=\Phi_{\alpha, K, \load,\Phifactor}$, where
    \begin{align}
    \label{eq:phi-linear-factor}
    \Phi_{\alpha, K, \load,\Phifactor} = f_\load \, \# \, \Phifactor^{\alpha,\load} = \Phifactor^{\alpha,\load}(f_\load^{-1}(\cdot)),
    \end{align}
    with $f_\load:\simplex{K-1} \to \simplex{d-1}$ defined by $f_\load(\bz) = \load \bz / \| \load \bz\|$ and with $\Phifactor^{\alpha, \load}$ the (tilted) probability measure on $\simplex{K-1}$ defined by
    \begin{align}
    \label{eq:cZB}
    \Phifactor^{\alpha,\load}(\diff \bz) 
    &= c_{\alpha, \load, \Phifactor}^{-1} \, \| \load \bz \|^\alpha \,  \Phifactor(\diff z) 
    \quad \text{ with } \quad
    c_{\alpha, \load, \Phifactor} = \int_{\simplex{K-1}}\| \load \bz \|^\alpha \, \Phifactor(\diff z).
\end{align}
\end{proposition}

Similar to the approach in Section~\ref{sec:latent-linear-factor-models}, the previous result gives rise to the following definition.

\begin{definition}[Regular variation of latent linear $K$-factor form]
\label{def:linear-factor-model-rv}
Let $\bX \in \R^d$ be regularly varying with index $\alpha>0$ and let $K \in [d]$. Fix some norm $\| \cdot \|$ on $\R^d$, and let $\Phi_\bX$ be the spectral measure of $\bX$ with respect to $\|\cdot \|$. The spectral measure $\Phi_\bX$ is said to be of \textit{latent linear $K$-factor form} (with respect to the norm $| \cdot |$ on $\R^K$) if there exists a full column rank matrix $\load \in \R^{d \times K}$, the \textit{loading matrix}, and a spectral measure $\Phifactor$ on $\simplex{K-1}$ such that $\Phi_\bX=\Phi_{\alpha, K,\load,\Phifactor}$ from \eqref{eq:phi-linear-factor}.
\end{definition}

By construction, the model in Definition~\ref{def:linear-factor-model-rv} depends on the choice of two norms, one on the ambient space $\R^d$ and one on the latent space $\R^K$. The behavior of the model under changes of theses norms is discussed in Lemma~\ref{lem:geometry}.

\begin{definition}[Tail spectral covariance matrix]
\label{def:tail-pairwise-dependence-matrix-regular-variation}
    Let $\bX \in \R^d$ be regularly varying. The \textit{tail spectral covariance matrix} (TSCM) of $\bX$ (with respect to the norm $\| \cdot\|$) is defined as 
    \begin{align} \label{eq:definition-tpdm2}
    \shittyTPDM_\bX = \int_{\simplex{d-1}} \bx \bx^\top \Phi_\bX(\diff \bx),
    \end{align}
    where $\Phi_\bX$ denotes the spectral measure of $\bX$ with respect to the norm $\| \cdot\|$. 
\end{definition}

Unlike other authors (e.g., \citealp{larsson2012extremal}), we refrain from calling $\shittyTPDM_\bX$ a tail pairwise dependence matrix: unlike $\TPDM_\bX$, it depends on the marginal behavior of $\bX$ and therefore does not constitute a genuine measure of dependence. 
Note that, in general, $\shittyTPDM_\bX \neq \TPDM_\bX$.

\begin{proposition}[Tail spectral covariance matrix for latent linear factor regular variation]
\label{prop:tpdm-latent-factor-linear-model-regular-variation}
    Suppose $\bX \in \R^d$ is regularly varying with index $\alpha>0$ and spectral measure $\Phi_\bX$ of latent linear $K$-factor form, i.e., $\Phi_\bX=\Phi_{\alpha, K,\load,\Phifactor}$ from \eqref{eq:phi-linear-factor}. Then, with 
    $c_{\alpha, \load, \Phifactor}$ from \eqref{eq:cZB}, we have
    \begin{align} \label{eq:B_Z}
    \shittyTPDM_{\bX} 
    = 
    \load C_{\alpha, \load, \Phifactor} \load^\top 
    \quad \text{ with } \quad 
    C_{\alpha, \load, \Phifactor} 
    = 
    c_{\alpha, \load, \Phifactor}^{-1} \int_{\simplex{K-1}} \frac{\bz \bz^\top}{\| \load \bz\|^{2-\alpha}} \, \Phifactor(\diff \bz).
    \end{align}
\end{proposition}

The subsequent identifiability result is based on the following parameter space which naturally extends $\paraspacetaildep$ from \eqref{eq:parameter-space-tail-dependence}:
\begin{multline}
    \label{eq:parameter-space-regular-variation}
    \paraspacespectral
    = 
    \big\{ (\alpha, K, \load,  \Phifactor) \mid K \in [d] \text{ and } 
     (\alpha, \load, \Phifactor) \in  (0,\infty) \times \setloadpurerv \times \setspectralK{|\cdot|} 
     \\ 
     \textrm{ satisfies }  
     C_{\alpha, \load, \Phifactor} \in \diagdom{K} \big\},
\end{multline}
where $\setspectralK{|\cdot|}$ denotes the set of spectral measures on $\simplex{K-1}=\simplex{K-1}(|\cdot|)$ (which coincides with the set of probability measure on $\simplex{K-1}$; see Lemma 2.2.2 in \citealp{kulik2020heavy}) and where
\begin{align}
\label{eq:pure-loading-matrices-rv}
\setloadpurerv = \big\{ \load \in \R^{d \times K} \mid \forall j \in [d]: \|\load_{j,\cdot} \|_1 \le 1 ~~\text{ and }~~ \forall a \in [K] \,\exists j \in [d]: \load_{j,\cdot} \in \{ \pm \be_a^\top\} \big\}.
\end{align}

As in Section~\ref{sec:latent-linear-factor-models}, 
two pairs $(\load, \Phifactor)$ and $(\load', \Phifactor')$ from $\setloadpurerv  \times \setspectralK{|\cdot|}$ are called equivalent, notation  $(\load, \Phifactor) \sim (\load', \Phifactor')$ if there exists a signed permutation matrix $P \in \{-1,0,1\}^{K \times K}$ (i.e., $P=\Pi \cdot D$ with $\Pi$ a permutation matrix and $D$ a diagonal matrix with diagonal entries $\pm1$)
with $P^{-1}=P^\top$ such that
\[
\load = \load'P \quad \text{ and } \quad  \Phifactor = L_{P^\top}\, \#\, \Phifactor' = \Phifactor'(L_{P^\top}^{-1}(\cdot)),
\]
where $L_{P^\top}:\simplex{K-1} \to \simplex{K-1}$ is defined by $ \bz \mapsto L_{P^\top}(\bz) = P^\top \bz$. In terms of random variables, the second equation in the last display means that $\bm \Lambda =_d P^\top \bm \Lambda'$ for $\bm \Lambda \sim \Phifactor$ and $\bm \Lambda' \sim \Phifactor'$, and together with the first equation, we get that $\load \bm \Lambda =_d \load'P P^\top \bm \Lambda' = \load'\bm \Lambda'$. We can now state the formal result.

\begin{theorem}
\label{thm:identifieability-regular-variation}
Recall $\Phi_\bX=\Phi_{\alpha, K,\load, \Phifactor}$ from \eqref{eq:phi-linear-factor},
and let $\setspectrald{\|\cdot\|} $ denote the set of spectral measures on $\simplex{d-1}$.
The mapping 
\[
T: \paraspacespectral \mapsto (0,\infty) \times \setspectrald{\|\cdot\|}, \quad (\alpha, K, \load, \Phifactor) \mapsto (\alpha, \Phi_{\alpha, K,\load, \Phifactor})
\] 
is injective up to signed label permutations, that is, 
$T(\alpha^{(1)}, K^{(1)},\load^{(1)}, \Phifactor^{(1)}) = T(\alpha^{(2)}, K^{(2)},\load^{(2)}, \Phifactor^{(2)})$ implies $\alpha^{(1)}=\alpha^{(2)}$, $K^{(1)}=K^{(2)}$ and $(\load^{(1)},\Phifactor^{(1)}) \sim (\load^{(2)},\Phifactor^{(2)})$.
\end{theorem}

\section{Proofs for Section~\ref{sec:latent-linear-factor-models}}

\begin{proof}[Proof of Proposition~\ref{prop:genuine-linear-factor-dependence}]
A straightforward calculation shows that $\bZ$ is regularly varying with index $\alpha=1$ and with spectral measure $\Psifactor$ with respect to the 1-norm (moreover, its margins are standard Pareto, whence its spectral dependence measure is also given by $\Psifactor$).
We may hence apply Proposition~\ref{prop:genuine-linear-factor-dependence-regular-variation} with $\bE=\bm 0$ and $\| \cdot \|$ an arbitrary fixed norm on $\R^d$ and $|\cdot|$ the 1-norm to obtain that $\bY'$ has spectral measure (with respect to $\| \cdot \|$) $\Phi_{\bY'} = \Psifactor^\loadn(f^{-1}_{\loadn}(\cdot) )$ as in \eqref{eq:spectral-dependence-measure-factor}.  
Next, the margins of $\bY'$ are standard Pareto: indeed, since $R$ and $\bm \Lambda$ are independent with $\Exp[\Lambda_a]=K^{-1}$ and since $\loadn$ has row-sums equal to one, we have, for any $j \in [d]$,
\begin{align*}
\Prob(Y_j'>x) 
= 
\Prob\Big( KR \sum_{a\in [K]} \loadn_{ja} \Lambda_a>x\Big)
&=
\Exp\Big[ \Prob\Big( R > \Big(K\sum_{a\in [K]} \loadn_{ja} \Lambda_a\Big)^{-1} x \, \Big| \, \bm \Lambda \Big) \Big]
\\&=
\frac{K}x \Exp\Big[ \sum_{a\in [K]} \loadn_{ja} \Lambda_a \Big]= \frac1x, \qquad x \ge 1.
\end{align*}
As a consequence, the spectral dependence measure of $\bY'$ with respect to $\| \cdot \|$ is also given by $\Phi_{\bY'}$, that is, $\Psi_{\bY'}=\Phi_{\bY'}$.
Apply Lemma~\ref{lem:stdf-spectral-dependence-measure-factor-model} to conclude the formula for $L$.
\end{proof}

\begin{proof}[Proof of Lemma~\ref{lem:stdf-spectral-dependence-measure-factor-model}]
As argued in the proof of Proposition~\ref{prop:genuine-linear-factor-dependence}, $\Psifactor^\loadn(f^{-1}_{\loadn}(\cdot) )$ is a spectral dependence measure with respect to the norm $\| \cdot \|$; more precisely, it is both the spectral measure and the spectral dependence measures of $\bm Y'$ defined in that proposition.

Suppose first that $\bm{X}$ has spectral dependence measure $\Psi_\bX = \Psifactor^\loadn(f^{-1}_{\loadn}(\cdot) )$ with respect to the norm $\| \cdot \|$. Then the STDF of $\bm{X}$ exists, and by \eqref{eq:relation-stdf-spectral-dependence-measure}, we have for $\bx \in [0,\infty)^d$,
\begin{align*}
    L(\bx) 
    = 
    \varsigma  \int_{\simplexpos{d-1}(\| \cdot \|)} \max_{j \in [d]} (\lambda_j x_j) \Psi_\bX(d\bm{\lambda}) 
    &= 
    \varsigma \int_{\simplexpos{K-1}(\| \cdot \|_1)} \max_{j \in [d]}\big(  x_j (f_\loadn(\bz))_j \big) \, \Psifactor^\loadn(\diff \bz) \\
    &=
    \frac{\varsigma}{c_{\loadn, \Psifactor}}  \int_{\simplexpos{K-1}(\| \cdot \|_1)} \max_{j \in [d]}\Big( x_j \sum_{a \in [K]} \loadn_{ja} z_a\Big) \, \Psifactor(\diff z),
\end{align*}
where $\varsigma>0$ is some constant.
Since $\loadn$ has row sums 1 and since $\Psifactor \in \setspectraldepposK{\|\cdot\|_1}$ satisfies $\int_{\simplexpos{K-1}(\| \cdot \|_1)} z_a  \, \Psifactor(\diff z)=K^{-1}$ for any $a \in [K]$ by \eqref{eq:set-of-spectral-dependence-measures}, we obtain that, for any $\ell \in [d]$,
\[
    1 = L(\bm{e}_\ell) = \frac{\varsigma}{c_{\loadn,  \Psifactor} }
    \sum_{a \in [K]} \loadn_{\ell a}\int_{\simplexpos{K-1}(\| \cdot \|_1)} z_a  \, \Psifactor(\diff z)
    =
    \frac{\varsigma}{c_{\loadn,  \Psifactor} \cdot K}.
\]
Hence $\varsigma = c_{\loadn, \Psifactor} \cdot K$ and the penultimate display yields
\[
    L(\bx) = K \int_{\simplexpos{d-1}(\| \cdot \|)} \max_{j \in [d]}\Big( x_j \sum_{a \in [K]} \loadn_{ja} z_a\Big) \, \Psifactor(\diff z) =L_{K, \loadn, \Psifactor}(\bx).
\]

For the converse, we may use the fact that existence of the STDF implies existence of the spectral dependence measure, and that the latter is uniquely determined by the former; see, e.g., Lemma 2.1 in \cite{boulin2026structured}.
\end{proof}

\begin{proof}[Proof of Proposition~\ref{prop:tpdm-latent-factor-linear-model}]
In view of Lemma~\ref{lem:stdf-spectral-dependence-measure-factor-model}, the result is an immediate consequence of Proposition~\ref{prop:tpdm-latent-factor-linear-model-regular-variation} applied to the random vector $\bY$. Indeed, the latter amounts to setting $\alpha=1, \load=\loadn$ and $\Phifactor=\Psifactor$ in \eqref{eq:B_Z}, which is exactly \eqref{eq:C_ZA}.
\end{proof}

\begin{proof}[Proof of Theorem~\ref{thm:identifieability-tail-dependence}]
Recall that a function is injective iff it has a left-inverse. In the given context this means that we need to construct a function $h:\mathcal L(d) \to \paraspacetaildep$ such that, for any $(K, \loadn, \Psifactor) \in \paraspacetaildep$, $(K', \loadn', \Psifactor') := h\circ T(K, \loadn, \Psifactor)$ satisfies $K=K'$ and $(\loadn, \Psifactor) \sim (\loadn', \Psifactor')$. 

Take an arbitrary $L \in \mathcal L(d)$. If $L \notin T(\paraspacetaildep)$, then define $h(L) \in\paraspacetaildep$ arbitrary. If $L\in T(\paraspacetaildep)$, then there exists $(K,\loadn, \Psifactor) \in\paraspacetaildep$ such that $L=L_{K,\loadn,\Psifactor}$. Our goal is to apply Theorem~\ref{thm:identifieability-regular-variation}, and for that purpose, we start by constructing a norm $\| \cdot \|_v$ on $\R^d$ such that $C_{\loadn, \Psifactor} \in \diagdom{K}$, with $C_{\loadn, \Psifactor}$ from \eqref{eq:C_ZA} depending on the chosen norm on $\R^d$.

First, the required positive definiteness $C_{\loadn, \Psifactor}$ holds for any norm on $\R^d$. Indeed, for any $\bm v \in \R^K$ that is non-zero, we have
\[
\bm v^\top C_{\loadn, \Psifactor} \bm v = c_{\loadn, \Psifactor}^{-1}\int_{\simplexpos{K-1}(\|\cdot\|_1)} \frac{(\bm v^\top \bz)^2}{\| \loadn \bz\|} \, \Psifactor(\diff \bz).
\]
Since $\loadn  \bm z = \sum_{a \in [K]}z_a \loadn_{\cdot a}$ is a convex combination of the rows of $\loadn$, we have $\|\loadn  \bm z\| \le \sum_{a \in [K]} z_a \| \loadn_{\cdot a}\| \le \max_{a \in [K]} \| \loadn_{\cdot a} \| =: \loadn_\infty$ for any $\bm z \in \simplexpos{K-1}(\|\cdot\|_1)$. Hence, 
\[
\bm v^\top C_{\loadn, \Psifactor} \bm v 
\ge 
(\loadn_\infty c_{\loadn, \Psifactor})^{-1} \int_{\simplexpos{K-1}(\|\cdot\|_1)} (\bm v^\top \bz)^2 \, \Psifactor(\diff \bz)
=
(\loadn_\infty c_{\loadn, \Psifactor})^{-1} \bm v^\top \Sigma_{\Psifactor} \bm v >0,
\]
where the last inequality follows from the assumption that $\Sigma_\Psifactor$ is positive definite. Hence, $C_{\loadn, \Psifactor}$ is positive definite.

We next construct a norm such that $\Delta(C_{\loadn, \Psifactor})>0$. Recall the index set of the pure variables $I=I(\loadn)$ and $I_a=I_a(\loadn)$ from \eqref{eq:pure-variables}, and let $J=[K] \setminus I$ denote the impure variable indices.
For each $a \in [K]$, define $T_a = \sum_{j \in J} \loadn_{ja}$. Choose $\gamma>0$ sufficiently small such that $\gamma T_a < 1$ for all $a \in [K]$. Define $r_a = 1- \gamma T_a>0$. Now define a vector $\bm v \in \R^d$ by $v_j = r_a/|I_a|$ if $j \in I_a$ and $v_j=\gamma$ if $j \in J$. Then, for every $a \in [K]$, we have
\begin{align*}
    (\loadn^\top \bm v)_a  
    =  
    \sum_{j \in [d]} v_j \loadn_{ja}
    &=
    \sum_{j \in I_a}v_j \loadn_{ja} + \sum_{b \ne a} \sum_{j \in I_b} v_j \loadn_{ja} + \sum_{j \in J}v_j \loadn_{ja}
    \\&=
    \sum_{j \in I_a} \frac{r_a}{|I_a|} + \sum_{b \ne a} \sum_{j \in I_b} v_j \cdot 0 +  \gamma T_a
    =
    r_a+\gamma T_a=1,
\end{align*}
i.e., $\loadn^\top \bm v = \bm 1 \in \R^{K}$.
Define a weighted $\ell_1$-norm on $\R^d$ by $\| \bm x\|_v = \sum_{j \in [d]} v_j |x_j|$.
Then, for all $\bm z \in \simplexpos{K-1}(\| \cdot \|_1)$, we have
$
\| \loadn \bm z \|_{v}  = \bm v^\top \loadn \bm z = \bm 1^\top \bm z = 1, 
$
since $\bm z \in \simplexpos{K-1}(\|\cdot\|_1)$. As a consequence, for the weighted $\ell_1$-norm $\| \cdot \|_v$, we have $C_{\loadn, \Psifactor} = c_{\loadn, \Psifactor}^{-1} \Sigma_\psi$, and the assumption $\Delta(\Sigma_\psi)>0$ hence yields $\Delta(C_{\loadn, \Psifactor})>0$. 

Let $\Theta_{\Psi_v}$ denote the set in \eqref{eq:parameter-space-regular-variation}, with the implicit norm in the definition of that set being the $\|\cdot\|_v$ norm constructed above. By Theorem~\ref{thm:identifieability-regular-variation}, the function 
\[
T_v: \Theta_{\Psi_v} \to \setspectrald{\| \cdot\|_v}, 
\quad
(\alpha', K',\loadn', \Psifactor') \mapsto f_{\loadn'} \,\#\, (\Psifactor')^{\alpha', \loadn'}
\]
has a left-inverse up to signed label permutations, say $g_v$. The latter means that, for any fixed $(\alpha', K',\loadn', \Psifactor') \in  \Theta_{\Psi_v}$,  $(\alpha'', K'',\loadn'', \Psifactor'') := g_v \circ T_v(\alpha', K',\loadn', \Psifactor')$  satisfies $(\alpha', K')=(\alpha'', K'')$ and $(\loadn', \Psifactor') \sim (\loadn'', \Psifactor'')$.

Now, define $h(L) = p_{(2,3,4)} \circ g_v \circ S_{\| \cdot \|_v}(L)$, where $p_{(2,3,4)}$ is the coordinate projection defined by $p_{(2,3,4)}(\alpha, K',\loadn', \Psifactor')=(K',\loadn', \Psifactor')$ and where $S_{\| \cdot \|_v}$ is the bijection from $\setstdf{d}$ to $\setspectraldepposd{\| \cdot\|_v}$ from \eqref{eq:stdf-spectral-dependence-measure-bijection}. It then follows that, for every $(K,\loadn, \Psifactor) \in\paraspacetaildep$,
\begin{align*}
(K', \loadn', \Psifactor') := h \circ T(K,\loadn, \Psifactor) 
&= 
p_{(2,3,4)} \circ g_v \circ S_{\| \cdot \|_v}(L_{K,A,\Psifactor})
\\&=
p_{(2,3,4)} \circ g_v(f_\loadn \, \# \, \Psifactor^{1,\loadn})
\\&=
p_{(2,3,4)} \circ g_v \circ T_v(1,K,\loadn, \Psifactor) 
\end{align*}
where the second equality follows from Lemma~\ref{lem:stdf-spectral-dependence-measure-factor-model}. In view of the fact that $g_v$ is a left-inverse of $T_v$ up to signed label permutations, we obtain that $K=K'$ and $(\loadn, \Psifactor) \sim (\loadn', \Psifactor')$. This finalizes the proof.
\end{proof}

\begin{proof}[Proof of Proposition~\ref{prop:survival_stdf_lfm_tail_dep}]
The result is a generalization of \cite[Lemma 2.7]{boulin2026structured}, and the proof follows along similar lines.
Since $\bm{X}$ has STDF $L_{K,\loadn, \Psifactor}$, $\bm Y$ is regularly varying with index 1. As a consequence,  $R^\cup_{\mathcal{J}}$ exists and by \eqref{eq:polar-decomposition}, we have for any $\bm x \in [0,\infty)^d$ and any norm $\|\cdot\|$ in $\mathbb{R}^d$,
\begin{align*}
    R^\cup_{\mathcal{J}}(\bm x) 
    &= 
    \lim_{x \rightarrow \infty} x \mathbb{P}\Big( \exists J \in \mathcal{J}\,  \forall j \in J : Y_j > \frac{x}{x_j} \Big)  \\
    &= 
    \varsigma_{\bY} \cdot (\nu_\alpha \otimes \Psi_{\bm{X}}) \left( \left\{ (r,\bm \lambda) \in [0,\infty) \times \simplexpos{d-1}(\|\cdot\|) \, \Big| \, \exists J \in \mathcal{J}\, \forall j \in J: r\lambda_j > x_j^{-1} \right\} \right) \\
    &= 
    \varsigma_{\bY} \int_{\simplexpos{K-1}(\|\cdot\|_1)}\bigvee_{J \in \mathcal{J}} \bigwedge_{j \in J} (x_j \lambda_j) \, \Psi_{\bm{X}}(\diff\bm{\lambda}).
\end{align*}
Using that $\Psi_{\bm{X}} = \psi^\loadn (f_\loadn^{-1}(\cdot))$ by Lemma~\ref{lem:stdf-spectral-dependence-measure-factor-model}, we obtain that
\begin{align*}
    R^\cup_{\mathcal{J}}(\bm x) 
    = 
    \frac{\varsigma_{\bY}}{c_{\loadn, \Psifactor}} \int_{\simplexpos{K-1}(\|\cdot\|_1)} \bigvee_{J \in \mathcal{J}} \bigwedge_{j \in J} \Big(x_j \sum_{a \in [K]} \loadn_{ja} z_a\Big) \Psifactor(\diff \bz).
\end{align*}
Since $\loadn$ has row sums $1$ and since $\Psifactor \in  \setspectraldepposK{\|\cdot\|_1}$ satisfies $\int_{\simplexpos{K-1}(\|\cdot\|_1)} z_a \Psifactor(d\bm{z}) = K^{-1}$ for any $a\in [K]$ by \eqref{eq:set-of-spectral-dependence-measures}, we obtain
\begin{align*}
    1 
    = 
    \lim_{x \rightarrow \infty} x \mathbb{P}\left\{ Y_1 > x \right\} 
    = 
    R^\cup_{\{1\}}(\bm{e}_1)  
    = 
    \frac{\varsigma_{\bY}}{c_{\loadn, \Psifactor}} \sum_{a \in [K]} \loadn_{1a} \int_{\simplexpos{K-1}(\|\cdot\|_1)} z_a \Psifactor(\diff \bm{z}) 
    = 
    \frac{\varsigma_{\bY}}{c_{\loadn, \Psifactor} \cdot K}.
\end{align*}
Hence $\varsigma_{\bY} = c_{\loadn, \Psifactor} \cdot K$ and the penultimate display yields:
\[
    R^\cup_{\mathcal{J}}(\bm x) 
    = 
    K\int_{\simplexpos{K-1}(\|\cdot\|_1)} \bigvee_{J \in \mathcal{J}} \bigwedge_{j \in J} \Big(x_j \sum_{a \in [K]} A_{ja} z_a\Big) \Psifactor(\diff \bm{z})
\]
as asserted.
The claim for $L^\cap_{\mathcal{J}}$ follows along similar lines.
\end{proof}

\begin{lemma} \label{lem:sufficient-condition-for-maximum-norm}
    Suppose that $\int \bm z \bm z^\top / \|\bm z \|_\infty \psi(\diff \bm z) \in \diagdom{K}$. If $C=C_{\loadn, \Psifactor}$ from \eqref{eq:C_ZA} is calculated using the maximum norm $\| \cdot \|_\infty$ on $\R^d$, then $C \in \diagdom{K}$.
\end{lemma}

\begin{proof}
    Clearly, $\| \loadn \bm z\|_\infty = \max_{j \in [d]} |\sum_{a \in [K]}A_{ja} z_a| \le \|\bm z\|_\infty \max_{j \in [d]} \sum_{a \in [K]}|A_{ja}|=\|\bm z\|_\infty$, where we have used that $\loadn$ has rows sums 1. On the other hand, since each unit vector appears at least once among the rows of $\loadn$, we have $\| \loadn \bm z\|_\infty \ge \max_{a \in [K]}|\bm e_a^\top \bm z| = \| \bm z \|_\infty$. Hence,  $\| \loadn \bm z\|_\infty= \| \bm z \|_\infty$, which yields $C_{\loadn,\Psifactor}=c_{\loadn, \Psifactor}^{-1} \int \bm z \bm z^\top / \|\bm z \|_\infty \psi(\diff \bm z) \in \diagdom{K}$ by assumption.
\end{proof}

\section{Proofs for Section~\ref{sec:estimating-latent-linear-factor-models}}

\begin{proof}[Proof of Lemma~\ref{lem:spectral-measure-of-pure-variables}]
    Throughout, we equip both $\R^d$ and $R^K$ with the 1-norm.
    Write $\Phi_\bY$ for the spectral measure of $\bm Y$, and recall that $\Phi_\bY = \Psi_\bX$.
    For $a \in [K]$, let $\bm m_a =(m_{a1}, \dots, m_{ad})$ have coordinates $m_{aj} = |I_a|^{-1}\bm 1( j \in I_a)$. Then, with $M\in \R^{K \times d}$ having rows $\bm m_1, \dots, \bm m_K$, we have $\bm Z = M \bm Y$. Exercise 2.4 in \cite{kulik2020heavy} implies that $\bm Z$ is regularly varying with tail index 1 and spectral measure 
    \[
    \Phi_{\bm Z} 
    = 
    \Big( \int_{\simplexpos{d-1}} \| M \bm y \|_1 \Phi_{\bY}(\diff \by) \Big)^{-1} \Big(\int_{\simplexpos{d-1}} \bm 1\Big(\frac{M \by}{\| M\by\|_1} \in \cdot \Big) \| M \by\|_1 \Phi_{\bY}(\diff \by) \Big).
    \]
    Next, by Lemma~\ref{lem:stdf-spectral-dependence-measure-factor-model}, we have $\Phi_\bY = \Psi_\bX = f_{\loadn}  \, \# \, \Psifactor^\loadn =  \Psifactor^\loadn(f^{-1}_{\loadn}(\cdot) )$,
    which yields 
    \[
    \Phi_{\bm Z} = \Big( \int_{\simplexpos{K-1}} \| M f_{\loadn}(\bm z) \|_1 \Psifactor^A(\diff \bm z) \Big)^{-1} \Big(\int_{\simplexpos{K-1}} \bm 1\Big(\frac{M f_{\loadn}(\bm z)}{\| M f_{\loadn}(\bm z)\|_1} \in \cdot \Big) \| M f_{\loadn}(\bm z)\|_1 \Psifactor^A(\diff \bm z) \Big)
    \]
    by the change of variable formula.
    Since
    $(MA\bm)_{ab}  = \sum_{j \in [d]} m_{aj}A_{jb}=|I_a|^{-1} \sum_{j \in I_a} A_{jb} = \bm 1(a=b)$, we have $Mf_A(\bm z) = MA \bm z / \| A\bz\|_1 = \bz / \|A\bz\|_1$ for every $\bz\in\simplexpos{K-1}$, which yields
    \[
    \Phi_{\bm Z} = \Big( \int_{\simplexpos{K-1}} \frac1{\| A\bm z \|_1} \Psifactor^A(\diff \bm z) \Big)^{-1} \Big(\int_{\simplexpos{K-1}} \bm 1\big(\bz\in \cdot \big)\frac1{\| A\bm z \|_1} \Psifactor^A(\diff \bm z) \Big) = \Psifactor.
    \]
Here, the last equality follows from the definition of $\Psifactor^A$ in \eqref{eq:cZA}.    
\end{proof}

\section{Proofs for Section~\ref{sec:latent-linear-factor-models-regular-variation}}

\begin{proof}[Proof of Proposition~\ref{prop:genuine-linear-factor-dependence-regular-variation}]
    In view of Lemma~\ref{lem:vanishing_noise}, it is sufficient to show the assertion for $\bE=\bm 0$. 
    
    Let $g_\bZ$ be an auxiliary function and let $\nu_\bZ$ be the associated exponent measure of $\bZ$ such that $g_\bZ(x) \Prob(x^{-1}\bZ \in \cdot) \to \nu_\bZ(\cdot)$ in $\Mb(\Eb_0^K)$ as $x \to \infty$. Define $h:\R^K \to \R^d$ by $h(\bz) = \load \bz$.
    Regular variation of $\bX=h(\bZ)$ (with tail index $\alpha$, auxiliary function $g_{h(\bZ)} = g_\bZ$, and associated exponent measure $\nu_{h(\bZ)} :=\nu_\bZ(h^{-1}(\cdot))$)
    then follows from Proposition 2.1.12 in \cite{kulik2020heavy}, once we show that $\nu_{h(\bZ)}$ is not the null measure on $\Eb_0^d$. For that purpose, fix $\eps>0$. By the polar decomposition $\nu_{\bZ}\circ T^{-1} = \varsigma_{\bZ} \cdot (\nu_\alpha \otimes \Phifactor)$ from \eqref{eq:polar-decomposition} and Proposition 2.2.3 in \cite{kulik2020heavy}, we obtain that
    \begin{align*}
    \nu_{h(\bZ)}\left(\left\{  \bx \in \R^d : \|\bx\| > \eps \right\}\right) 
    = 
    \nu_{\bZ} \left( \left\{  \bz \in \R^K : \|\load \bz\| > \eps \right\}\right) 
    &= 
    \varsigma_{\bZ} \int_{\simplex{K-1}} \int_{ r > \eps / \| \load \bm \lambda\|} \alpha r^{-\alpha-1} \, \diff r \, \Phifactor(\diff\bm \lambda) \\ 
    &= 
    \varsigma_{\bZ} \eps^{-\alpha}\int_{\simplex{K-1}} \|\load \bm \lambda\|^\alpha\, \Phifactor(\diff\bm \lambda).
    \end{align*}
    In view of the full column rank assumption on $\load$, we have $\load \bm \lambda \neq 0$ for any $\bm \lambda \in \simplex{K-1}$. Hence, $\bm \lambda \mapsto \| \load \bm \lambda\|^\alpha$, considered as a function on the compact sphere $\simplex{K-1}$, is bounded away from zero. Hence, since $\Phifactor$ is not the null measure, $\nu_{h(\bZ)}$ cannot be the null measure either.

    It remains to show that the spectral measure of $h(\bZ)$ satisfies $\Phi_{h(\bZ)}=\Phi_{\alpha, K, \load,\Phifactor}$. For that purpose note that, by Proposition 2.1.12 in \cite{kulik2020heavy}, for an arbitrary Borel set $C \subset \simplex{d-1}$,
    \begin{align*}
    m(C) := \nu_{h(\bZ)}\left( \left\{ \bx \in \R^d : \|\bx\| > 1, \frac{\bx}{\|\bx\|} \in C \right\} \right) 
    &= 
    \nu_{\bZ}\left( \left\{ \bz \in \R^K : \|\load\bz\| > 1, \frac{\load \bz}{\|\load \bz\|} \in C \right\} \right) 
    \\ 
    &= \varsigma_{\bZ} \int_{\simplex{K-1}} \int_0^\infty \indicator{\{\|\load r \bm \lambda \| > 1, \frac{\load \bm \lambda}{\|\load \bm \lambda\|} \in C\}} \alpha r^{-\alpha -1} \, \diff r\, \Phifactor(\diff \bm \lambda) \\
    &= 
    \varsigma_{\bZ} \int_{\simplex{K-1}} \|\load \bm \lambda\|^\alpha \indicator{C}\left( \frac{\load \bm \lambda}{\|\load \bm \lambda\|} \right) \Phifactor(\diff \bm \lambda)
    \\&=
    \varsigma_\bZ \cdot c_{\alpha, \load, \Phifactor} \int_{\simplex{K-1}} \indicator{C}\left( \frac{\load \bm \lambda}{\|\load \bm \lambda\|} \right) \Phifactor^{\alpha, \load}(\diff \bm \lambda)  
    \\& = \varsigma_\bZ \cdot  c_{\alpha, \load, \Phifactor} \cdot \Phi_{\alpha, K, \load, \Phifactor}(C),
    \end{align*}
    where we have used the definitions of $c_{\alpha, \load, \Phifactor}, \Phifactor^{\alpha, \load}$ and $\Phi_{\alpha, K, \load, \Phifactor}$ from \eqref{eq:phi-linear-factor} and \eqref{eq:cZB}.  As a consequence, by definition of $\Phi_{h(\bZ)}$ in \eqref{eq:definition-spectral-measure}, $\Phi_{h(\bZ)}(C) = 
    m(C)/{m(\simplex{d-1})} 
    =
    \Phi_{\alpha, K, \load, \Phifactor}(C)$
    as asserted.
\end{proof}

\begin{proof}[Proof of Proposition~\ref{prop:tpdm-latent-factor-linear-model-regular-variation}]
By the definition of $\Phi_\bX=\Phi_{\alpha, K,\load,\Phifactor}$ from \eqref{eq:phi-linear-factor} and the change of variable formula, we have
\begin{align*} 
    \shittyTPDM_\bX 
    = 
    \int_{\simplex{d-1}} \bx \bx^\top \Phi_\bX(\diff \bx)
    &= 
    \int_{\simplex{d-1}} \bx \bx^\top (\Phifactor^{\alpha, \load} \circ f_\load)^{-1}(\diff \bx)
    \\&=
    \int_{\simplex{K-1}} \frac{\load \bz \bz^\top \load^\top}{\| \load \bz\|^{2}} \, \Phifactor^{\alpha, \load}(\diff \bz)
    =
     \load c_{\alpha, \load, \Phifactor}^{-1} \int_{\simplex{K-1}}  \frac{\bz \bz^\top}{\| \load \bz\|^{2-\alpha}} \, \Phifactor(\diff \bz) \load^\top,
\end{align*}
as asserted.
\end{proof}

The proof of Theorem~\ref{thm:identifieability-regular-variation} is based on the the following auxiliary lemma, for which we need to slightly extend the domain of some of the definitions from Remark~\ref{rem:reconstructing-parameters-tail-dependence}. First, for $\load \in \setloadpurerv$, let
\[
I(\load) =\bigcup_{a \in [K]} I_a(\load), \qquad I_a(\load) = \{ j \in [d]: \load_{j \cdot}^\top \in \{ \pm \be_a\} \}.
\]
Next, extend the definitions of $m_j(\Sigma)$ and $S_j(\Sigma)$ from \eqref{eq:m_j-and_S_j} from $\Sigma \in [0,\infty)^{d \times d}$ to $\Sigma \in \R^{d \times d}$
by 
\begin{align} \label{eq:m_j-and_S_j-extended}
m_j(\Sigma) := m_j(|\Sigma|) = \max_{\ell \in [d]} |\Sigma(j,\ell)|, \quad S_j(\Sigma) := S_j(|\Sigma|) = \argmax_{\ell \in [d]} |\Sigma(j,\ell)|.
\end{align}
Recall $\diagdom{K}$ from  \eqref{eq:diagdom} and $\setloadpurerv$ from \eqref{eq:pure-loading-matrices-rv}.

\begin{lemma}
\label{lem:identifiability_pure}
   Let $\load \in \setloadpurerv$ and $C \in \diagdom{K}$, and write $\Sigma=\load C \load^\top \in \R^{d \times d}$. Then:
       \begin{compactenum}[(1)]
        \item \label{item:identifieabilty-1}
        For any $a \in [K]$ and $j\in I_a(\load)$, we have $|\Sigma(j,\ell)|\le C(a,a)$ for all $\ell\in[d]$, with strict inequality for $\ell \notin I_a$. As a consequence, $S_j(\Sigma) = I_a(\load)$ and $M_j(\Sigma) = C(a,a)$.
        \item \label{item:identifieabilty-2}
        For any $j \in [d]$, we have $S_j(\Sigma) \cap I(\load) \ne \emptyset$.
        \item  \label{item:identifieabilty-3}
        We have $j \in I(\load)$ if and only if $M_j(\Sigma) = M_\ell(\Sigma)$ for all $\ell \in S_j(\Sigma)$. As a consequence, $I(\load)$ is uniquely determined by $\Sigma$:
         \[
        I(\load) = \{ j \in [d] : M_j(\Sigma) = M_\ell(\Sigma) \textrm{ for every } \ell \in S_j(\Sigma)\}.
        \] 
        \item \label{item:identifieabilty-4}
        The relation on $I(\load)$ given by 
        \[
        j \sim_\TPDM \ell \quad :\Longleftrightarrow \quad  \ell \in S_j(\TPDM)
        \]
        is an equivalence relation, the number of equivalence classes is equal to $K$, and each equivalence class corresponds to exactly one set $I_a(\load)$ with $a \in [K]$. In other words, both $K$ and the partition $\bm{\Ic} = \{I_a(\load)\}_{a \in [K]}$ can be determined from $\Sigma$.
    \end{compactenum}
\end{lemma}

\begin{proof}[Proof of Lemma~\ref{lem:identifiability_pure}]~

\textit{Proof of \eqref{item:identifieabilty-1}:} 
For given $j \in [d]$, we define the set $s(j) := \{ b \in [K] :  \load_{jb} \ne 0 \}$. The identity $\Sigma=\load C \load^\top$ means that, for any $j, \ell \in [d]$,
\begin{align*}
    \Sigma(j,\ell) 
    = 
    \sum_{b \in [K]}  \sum_{b' \in [K]} \load_{jb} C(b,b')\load_{\ell b'} 
    = 
    \sum_{b \in s(j)}  \sum_{b' \in s(\ell)} \load_{jb} C(b,b')\load_{\ell b'}  .
\end{align*}
Specialized to $j,\ell\in I_a$ for some $a \in [K]$, we have $s(j)=s(\ell) = \{a\}$ and $|\load_{ja}| = |\load_{\ell a}|=1$, which yields 
\begin{align*}
    |\Sigma(j,\ell)| 
    = 
    \Big|\sum_{b \in s(j)}  \sum_{b' \in s(\ell)} \load_{jb} C(b,b')\load_{\ell b'}  \Big| 
    = 
    \left|\load_{ja}  C(a,a) \load_{\ell a}  \right| =  | C(a,a) | = C(a,a).
\end{align*}
This is the first claim. Regarding the second one, fix $a \in [K]$, and consider $j\in I_a$ and $\ell \notin I_a$. Since $s(j) = \{a\}$ and $|\load_{ja}|=1$, we obtain
\begin{align*}
    |\Sigma(j,\ell)| 
    = 
    \Big|\sum_{b \in s(j)}  \sum_{b' \in s(\ell)} \load_{jb} C(b,b')\load_{\ell b'}  \Big|  
    &= 
    \Big| \sum_{b' \in s(\ell)}  C(a,b')\load_{\ell b'}  \Big| 
    \\&\le 
    \sum_{b' \in s(\ell)}  |C(a,b')\load_{\ell b'}| 
    \\&=
    \bm 1(a \in s(\ell))  C(a,a) |\load_{\ell a}| +
    \sum_{b' \in s(\ell), b'\ne a} |C(a,b')\load_{\ell b'}|.
\end{align*}
Here, the index set of the last sum is non-zero since $\ell \notin I_a$. As a consequence, since $C \in \diagdom{K}$, see \eqref{eq:diagdom},
\begin{align*}
    |\Sigma(j,\ell)| 
    < 
    \bm 1(a \in s(\ell))  C(a,a)|\load_{\ell a}| +
    C(a,a)\sum_{b' \in s(\ell), b'\ne a} |\load_{\ell b'}|
    = C(a,a) \sum_{b' \in s(\ell)} |\load_{\ell b'}| \le C(a,a),
\end{align*}
where the last inequality follows from $\load \in \setloadpurerv$; see \eqref{eq:pure-loading-matrices-rv}. The proof of \eqref{item:identifieabilty-1} is finished.

\medskip
\textit{Proof of \eqref{item:identifieabilty-2}:}  The result from \eqref{item:identifieabilty-1} implies that $S_j(\Sigma) \cap I(\load) \ne \emptyset$ for any $j \in I(\load)$, so it remains to consider $j \in J:= [d] \setminus I(\load)$. For such an index, we have
\begin{align*}
        M_j(\Sigma) 
        = 
        \max_{\ell \in [d]} |\Sigma(j,\ell)| 
        &=  
        \max_{\ell \in [d] } \Big|\sum_{b' \in s(\ell)}\load_{\ell b'}\sum_{b \in s(j)}  \load_{jb} C(b,b')  \Big| 
        \\&\le 
        \max_{\ell \in [d]} \Big( \max_{b' \in s(\ell)} \Big| \sum_{b \in s(j)} \load_{jb} C(b,b') \Big| \Big) \times \Big(\sum_{b' \in s(\ell)}|\load_{\ell b'}|  \Big)
        \\&\le
        \max_{\ell \in [d]} \max_{b' \in s(\ell)} \Big| \sum_{b \in s(j)} \load_{jb} C(b,b') \Big|,
        %\\&= 
        %\max_{\ell \in [d]} \Big| \sum_{a \in s(j)} \load_{ja} C(a,b^*) \Big|,
\end{align*}
where we have used the row-sum condition from $\load \in \setloadpurerv$. The inner maximum over $b'\in s(\ell)$ in the last line is attained at some $b^* \in s(\ell) \subset [K]$. For any $\ell \in I_{b^*}(\load)$, which is non-empty since $\load \in \setloadpurerv$, a direct computation yields $|\Sigma(j,\ell)| = |\sum_{b \in s(j)} \load_{jb} C(b,b^*)|$, which implies that $\emptyset \ne I_{b^*}(\load) \subset S_j(\Sigma)$ by definition of $S_j$. In particular,  $S_j(\Sigma) \cap I(\load) \ne \emptyset$ as asserted.

\medskip
\textit{Proof of \eqref{item:identifieabilty-3}:}  For the necessity part, assume $j \in I(\load)$. Then $j \in I_a(\load)$ for some $a \in [K]$. The result in \eqref{item:identifieabilty-1} implies that $S_j(\Sigma) = I_a(\load)$ and $M_j(\Sigma) = C(a,a)$. Take any $\ell \in S_j(\Sigma)$. Since $S_j(\Sigma) = I_a(\load)$, we have $\ell \in I_a(\load)$, and the result in (1) gives $M_\ell(\Sigma)  = C(a,a) = M_j(\Sigma)$, as required. 
For the sufficiency part, assume $M_j(\Sigma) = M_\ell(\Sigma)$ for all $\ell \in S_j(\Sigma)$. We want to show that $j \in I(\load)$. The result in \eqref{item:identifieabilty-2} allows to choose some $\ell \in S_j(\Sigma) \cap I(\load)$. Then $\ell \in I_b(\load)$ for some $b \in [K]$, and \eqref{item:identifieabilty-1} yields $M_\ell(\Sigma) = C(b,b)$. Since $\ell \in S_j(\Sigma)$, we also have $|\Sigma(j,\ell)| = M_j(\Sigma)$. The assumption $M_j(\Sigma) = M_\ell(\Sigma)$ hence implies that $|\Sigma(j,\ell)| = C(b,b)$. But then $j \in I_b(\load) \subset I(\load)$ as required, since otherwise $|\Sigma(j,\ell)| < C(b,b)$ by~\eqref{item:identifieabilty-1}. 

\medskip
\textit{Proof of \eqref{item:identifieabilty-4}:}  Reflexivity of $\sim_\Sigma$ is immediate. For proving symmetry, suppose $j,\ell \in I(\load)$ satisfy $j \sim_\Sigma \ell$. Choose $a\in[K]$ such that $j \in I_a(\load)$. It follows from \eqref{item:identifieabilty-1} that $S_j(\Sigma)=I_a(\load)$; hence, by definition of $\sim_\Sigma$, we have $\ell \in I_a(\load)$. Applying \eqref{item:identifieabilty-1} again, we obtain that $S_\ell(\Sigma) = I_a(\load)$. Hence, $j \in S_\ell(\Sigma)$, which means that $\ell \sim_\Sigma j$. For proving transitivity, suppose $j,\ell,k \in I(\load)$ satisfy $j \sim_\Sigma \ell$ and $\ell \sim_\Sigma k$. Again, choose $a\in[K]$ such that $j \in I_a(\load)$. As shown in the symmetry proof, $j \sim_\Sigma \ell$ implies that $S_j(\Sigma) = S_\ell(\Sigma) =I_a(\load)$. In particular, $\ell \in I_a(\load)$, and then $\ell \sim_\Sigma k$ implies that $S_\ell(\Sigma) = S_k(\Sigma) =I_a(\load)$. Hence, $j \in S_k(\Sigma)$, which means that $k\sim_\Sigma j$ as required.
  
Finally, each equivalence class under $\sim_\Sigma$, i.e., $[j]_{\Sigma} := \{\ell \in I(\load) : \ell \sim_\Sigma j\}$, corresponds to one pure variable set $I_a$, with $a \in [K]$. More precisely, for each $j \in I(\load)$, there exists $a \in [K]$ such that $j \in I_a(\load)$, and we claim that $[j]_\Sigma=I_a(\load)$. Indeed, if $\ell \in I_a(\load)$, then $S_\ell(\Sigma) =I_a(\load)$ by \eqref{item:identifieabilty-1}, so $j \in S_\ell(\load)$ and hence $\ell \sim_\Sigma j$. Conversely, if $\ell\in [j]_{\Sigma}$ then $j \sim_\Sigma \ell$ by symmetry, which means that $\ell \in S_j(\Sigma)=I_a(\load)$, with the last equality following from \eqref{item:identifieabilty-1}.
\end{proof}

\begin{proof}[Proof of Theorem~\ref{thm:identifieability-regular-variation}]
Clearly, $\alpha^{(1)}=\alpha^{(2)}=:\alpha$. Further, write
\[
\Phi_\bX = \Phi_{\alpha^{(1)}, K^{(1)},\load^{(1)}, \Phifactor^{(1)}} = \Phi_{\alpha^{(2)}, K^{(2)},\load^{(2)}, \Phifactor^{(2)}} 
\quad \text{ and } \quad
\shittyTPDM = \int_{\simplex{d-1}} \bx\bx^\top \, \Phi_{\bX}(\diff \bx).
\]
As a consequence of Proposition~\ref{prop:tpdm-latent-factor-linear-model-regular-variation} and Lemma~\ref{lem:identifiability_pure}\,\eqref{item:identifieabilty-3}, we have $I:=I(\load^{(1)}) = I(\load^{(2)})$. Further, Lemma~\ref{lem:identifiability_pure}\,\eqref{item:identifieabilty-4} implies that $K:=K^{(1)}=K^{(2)}$, and that there exists a permutation $\pi:[K] \to [K]$ such that $I_a(\load^{(1)}) = I_{\pi(a)}(\load^{(2)})$. Let $\Pi=(\be_{\pi(1)}, \dots, \be_{\pi(K)}) = (\indicator{\{a=\pi(b)\}})_{a,b\in[K]} \in [0,1]^{K\times K}$. Next, for each $a\in[K]$ and each $j\in I_a(\load^{(1)})$, we have $\load_{j\cdot}^{(1)} \in \{\pm \be_a^\top\}$ and $\load_{j\cdot}^{(2)} \in \{\pm \be_{\pi(a)}^\top\}$. In particular, there exists a sign $s_a \in \{-1,1\}$ such that $\load_{ja}^{(1)}=s_a\load_{j\pi(a)}^{(2)}$. Writing $D=\diag(s_1, \dots, s_K) \in \R^{K \times K}$, and $P=\Pi \cdot D =(s_1\be_{\pi(1)}, \dots, s_K\be_{\pi(K)})=(s_b\indicator{\{a=\pi(b)\}})_{a,b\in[K]}$, the equation $\load_{ja}^{(1)}=s_a\load_{j\pi(a)}^{(2)}$ can be rewritten as
$
\load^{(1)}_{ja} = (\load^{(2)} \cdot P)_{ja},
$
and we have shown that it holds for all $ a \in [K]$ and $j \in I_a(\load^{(1)})$.
In other words, 
\begin{align} \label{eq:permutation-load}
\load^{(1)}_{j\cdot } =  (\load^{(2)} \cdot P)_{j\cdot} \qquad \forall j \in I,
\end{align}
and we now show that the same equation also holds for $j \in [d] \setminus I$. For that purpose, define $C^{(r)} := C_{\alpha^{(r)}, \load^{(r)}, \Phifactor^{(r)}}$, and note that
\begin{align} \label{eq:tpdm-rv-decomposition}
\sum_{b \in [K]}  \sum_{b' \in [K]} \load_{jb}^{(r)} C^{(r)}(b,b')\load_{\ell b'} ^{(r)}  = \shittyTPDM(j,\ell)
\end{align}
by Proposition~\ref{prop:tpdm-latent-factor-linear-model-regular-variation}, with the right-hand side not depending on $r$.  Fix $a,b\in[K]$, and let $j\in I_a(\load^{(r)})$ and $\ell \in I_b(\load^{(r)})$. The previous display then yields
$
C^{(r)}(a,b) =  \load_{ja}^{(r)} \load_{\ell b} ^{(r)}\shittyTPDM(j,\ell),
$
where have used that $\load_{ja}^{(r)}, \load_{\ell b} ^{(r)} \in \{-1,1\}$. Hence, for $j\in I_a(\load^{(1)})=I_{\pi(a)}(\load^{(2)})$ and $\ell \in I_b(\load^{(1)})=I_{\pi(b)}(\load^{(2)})$,
\[
C^{(1)}(a,b) = \load_{ja}^{(1)}\load_{\ell b} ^{(1)} \shittyTPDM(j,\ell) = s_a s_b \load_{j\pi(a)}^{(2)}\load_{\ell \pi(b)} ^{(2)} \shittyTPDM(j,\ell)
=
s_a s_b C^{(2)}(\pi(a),\pi(b)), 
\]
where we have used that $\load_{ja}^{(1)}=s_a\load_{j\pi(a)}^{(2)}$ and  $\load_{\ell b}^{(1)}=s_b\load_{\ell\pi(b)}^{(2)}$. The formula in the previous display holds for all $a,b \in [K]$, which can concisely be written as $C^{(1)} = P^\top C^{(2)} P$.

Fix $a \in [K]$, $j \in I_a(\load^{(r)})$ and $\ell \in [d] \setminus I$. Observing that $\load_{j\cdot}^{(r)} \in \{\pm \be_a^{\top}\}$, Equation~\eqref{eq:tpdm-rv-decomposition} implies that
\[
\load_{j,a}^{(r)}\shittyTPDM(j,\ell) 
=  
\sum_{b\in[K]}  C^{(r)}(a,b)\load_{\ell,b}^{(r)}
=C^{(r)}_{a\cdot} \bm \beta^{(r)}(\ell),
\]
where $\bm \beta^{(r)}(\ell) = (\load_{\ell,\cdot}^{(r)})^\top$ denotes the $\ell$th row of $\load^{(r)}$ (transposed) and where $C^{(r)}_{a\cdot}=(C_{a1}^{(r)}, \dots, C_{aK}^{(r)})$ is the $a$th row of $C^{(r)}$. Since the previous equation holds for all $j \in I_a$, we get that 
\[
C^{(r)}_{a\cdot} \bm \beta^{(r)}(\ell) = \frac{1}{|I_a(\load^{(r)})|} \sum_{j \in I_a(\load^{(r)})} \load_{j,a}^{(r)}\shittyTPDM(j,\ell) =: \theta_a^{(r)}(\ell) \qquad \forall a \in[K], \ell \in [d] \setminus I,
\]
or, in matrix-vector notation,
\[
C^{(r)} \bm \beta^{(r)}(\ell)  = \bm \theta^{(r)}(\ell) := (\theta_1^{(r)}(\ell), \dots, \theta_K^{(r)}(\ell))^\top \qquad \forall  \ell \in [d] \setminus I.
\]
Observe that, since $I_a(\load^{(1)})=I_{\pi(a)}(\load^{(2)})$ and $\load_{ja}^{(1)}=s_a\load_{j\pi(a)}^{(2)}$, we have 
$\theta_a^{(1)}(\ell)  = s_a \theta_{\pi(a)}^{(2)}(\ell) = (P^\top\bm \theta^{(2)}(\ell))_a$, that is,  $\bm \theta^{(1)}(\ell) = P^\top \bm \theta^{(2)}(\ell)$. Since $C^{(r)}$ is invertible by assumption, the previous display can be rewritten as $\bm \beta^{(r)}(\ell) = (C^{(r)})^{-1} \bm \theta^{(r)}(\ell)$. 
The relations $C^{(1)} = P^\top C^{(2)} P$ and $\bm \theta^{(1)}(\ell) = P^\top \bm \theta^{(2)}(\ell)$ shown above then yield
\[
        \bm \beta^{(1)}(\ell) 
        = (P^\top C^{(2)} P)^{-1} P^\top \bm \theta^{(2)}(\ell) 
        = P^\top (C^{(2)})^{-1} P P^\top \bm \theta^{(2)}(\ell) 
        = P^\top (C^{(2)})^{-1} \bm \theta^{(2)}(\ell) 
        = P^\top \bm \beta^{(2)}(\ell),
\]
which means that $\load^{(1)}_{\ell\cdot} =  (\load^{(2)} \cdot P)_{\ell\cdot}$ for all $\ell \in [d] \setminus I$ as required. Together with \eqref{eq:permutation-load}, we get that $\load^{(1)} = \load^{(2)}P$.

It remains to show that $\Phifactor^{(1)} = L_{P^\top}\, \# \, \Phifactor^{(2)}$, or, equivalently, that $L_P\, \# \, \Phifactor^{(1)} = \Phifactor^{(2)}$, where $L_P: \simplex{K-1} \to \simplex{K-1}$ is defined by $L_P(\bz) = P \bz$. Recall the function $f_\load: \simplex{K-1} \to \simplex{d-1}$ defined by $f_\load(\bz)=\load \bz / \| \load \bz\|$. In view of $\load^{(1)} = \load^{(2)}P$, we have $f_{\load^{(1)}}(\bz) = f_{\load^{(2)}P}(\bz) =f_{\load^{(2)}}(P\bz) = (f_{\load^{(2)}} \circ L_P)(\bz)$ for all $\bz \in \simplex{K-1}$.

Fix an arbitrary Borel set $D \subset \simplex{K-1}$. Choose a Borel set $C \subset \simplex{d-1}$ such that $f_{\load^{(2)}}^{-1}(C) = D$; since $f_{\load^{(2)}}$ is injective, we may for instance take $C=f_{\load^{(2)}}(D)$. On the one hand, by definition in \eqref{eq:phi-linear-factor}, we have 
\[
\Phi_\bX(C) 
= 
\Phi_{\alpha, K,\load^{(2)}, \Phifactor^{(2)}}(C)
=
(\Phifactor^{(2)})^{\alpha, \load^{(2)}}(D)
\]
On the other hand,  
\begin{align*}
\Phi_\bX(C)  
= 
\Phi_{\alpha, K,\load^{(1)}, \Phifactor^{(1)}}(C)
&=
(\Phifactor^{(1)})^{\alpha, \load^{(1)}}(f_{\load^{(1)}}^{-1}(C))
= (\Phifactor^{(1)})^{\alpha, \load^{(1)}}(L_P^{-1}(D)) = (L_P \, \# \, (\Phifactor^{(1)})^{\alpha, \load^{(1)}})(D).
\end{align*}
Since $D$ was arbitrary, we have shown that $L_P \, \# \, (\Phifactor^{(1)})^{\alpha, \load^{(1)}} = (\Phifactor^{(2)})^{\alpha, \load^{(2)}}$.
As a consequence, by the change of variable formula, for each Borel set $D \subset \simplex{K-1}$,
\begin{align*}
c_{\alpha, \load^{(2)}, \Phifactor^{(2)}} \Phifactor^{(2)}(D) 
&= 
\int_D \| \load^{(2)} \bm z\|^{-\alpha} \, (\Phifactor^{(2)})^{\alpha, \load^{(2)}}(\diff \bz)
\\&=
\int_{L_P^{-1}(D)} \| \load^{(2)} P \bm z\|^{-\alpha} \, (\Phifactor^{(1)})^{\alpha, \load^{(1)}}(\diff \bz)
\\&=
\int_{L_P^{-1}(D)} \| \load^{(1)}  \bm z\|^{-\alpha} \, (\Phifactor^{(1)})^{\alpha, \load^{(1)}}(\diff \bz) = c_{\alpha, \load^{(1)}, \Phifactor^{(1)}} \Phifactor^{(1)}(L_P^{-1}(D)),
\end{align*}
where we have used that $\load^{(1)} P^\top = \load^{(2)}$.
Take $D=\simplex{K-1}$ in the previous to equation to deduce that $c_{\alpha, \load^{(1)}, \Phifactor^{(1)}}=c_{\alpha, \load^{(2)}, \Phifactor^{(2)}}$, which then implies
$
\Phifactor^{(2)}(D)  = \Phifactor^{(1)}(L_P^{-1}(D)).
$
Since $D$ was arbitrary, we have shown that $L_P\, \# \, \Phifactor^{(1)} = \Phifactor^{(2)}$ as required.
\end{proof}

\section{Further auxiliary results}
\label{sec:auxiliary-results}

\begin{lemma}
    \label{lem:vanishing_noise}
    Let $\bX_1$ and $\bX_2$ be two random vectors with values in $\R^{d_1}$ and $\R^{d_2}$ respectively where $d_1, d_2 \in \mathbb{N}_{\ge 1}$. Assume $\bX_1$ is regularly varying with tail index $\alpha>0$, auxiliary function $g(\cdot)$ and associated exponent measure $\nu_{\bX_1}$. Moreover, assume that 
   $\Prob(\|\bX_2 \|>x) =o(x^{-\alpha})$ as $x \to \infty$. Then $(\bX_1, \bX_2)$ is regularly varying  with tail index $\alpha>0$, auxiliary function $g(\cdot)$ and associated exponent measure $\nu_{\bX_1}  \otimes \delta_{\bm 0}$. If $d_1 = d_2 = d$, then both $\bX_1 + \bX_2$ and $\bX_1 \vee \bX_2$ are regularly varying with tail index $\alpha$, auxiliary function $g(\cdot)$ and associated exponent measure $\nu_{\bX_1}$.
\end{lemma}

\begin{proof}[Proof of Lemma~\ref{lem:vanishing_noise}]
    The proof is a minor adaptation of the proof of Lemma 5.1 in \cite{boulin2026structured}; a lemma that treats the special case where $d_1=d_2$ and where the coordinates of $\bm X_j$ are non-negative. 
\end{proof}

\begin{lemma} \label{lem:constant-in-front-of-L}
    Consider the situation of Lemma 2.1 in \cite{boulin2026structured}. Let $\| \cdot \|$ be an arbitrary norm on $\R^d$. The constant $\varsigma=\varsigma_{\bm Y}^{\| \cdot\|} = \nu_{\bm Y}(\{\bm y : \| \bm y \| \ge 1\})$ defined in their item (iii) satisfies
    \[
    \varsigma_{\bm Y}^{\| \cdot\|} 
    =
    d\int_{\simplexpos{d-1}(\| \cdot\|_1)}  \| \bm \lambda\| \, \diff\Phi_{\bm Y}^{\|\cdot\|_1}(\bm \lambda)
    \]
    where $\Phi_{\bm Y}^{\|\cdot\|_1}$ is the spectral measure of $\bm Y$ with respect to the 1-norm.
\end{lemma}

\begin{proof}
    Consider the function $T: [0,\infty)^d \setminus \{0\} \to (0,\infty) \times \simplexpos{d-1}(\| \cdot\|_1)$ defined by 
    \[
    T(\bm y) = \Big( \| \bm y \|_1, \frac{\bm y}{ \| \bm y \|_1} \Big).
    \]
    Then
    \[
    \big\{\bm y : \| \bm y \|>1\big\} = T^{-1} \big(\big\{ (r, \bm \lambda)\in (0,\infty) \times \simplexpos{d-1}(\| \cdot\|_1): r > 1/\|\bm \lambda\|  \big\}\big).
    \]
    As a consequence, by definition of the spectral measure $\Phi_{\bm Y}^{\|\cdot\|_1}$   in (2.3) in \cite{boulin2026structured},
    \begin{align*}
    \varsigma_{\bm Y}^{\| \cdot\|} 
    &= 
    \nu_{\bm Y}(\{\bm y \in [0,\infty)^d \setminus \{0\}: \| \bm y \| \ge 1\})
    \\&=
    \varsigma^{\| \cdot\|_1}  \cdot (\nu_1 \otimes \Phi_{\bm Y}^{\|\cdot\|_1})\big(\big\{ (r, \bm \lambda)\in (0,\infty) \times \simplexpos{d-1}(\| \cdot\|_1): r > 1/\|\bm \lambda\|  \big\}\big)
    \\&=
    d \int_{\simplexpos{d-1}(\| \cdot\|_1)}  \| \bm \lambda\| \, \diff\Phi_{\bm Y}^{\|\cdot\|_1}(\bm \lambda),
    \end{align*}
    where we have used Fubini's theorem and the fact that $\varsigma^{\| \cdot\|_1}=d$.
\end{proof}

\begin{lemma} \label{lem:set-of-spectral-measures}
    Let $p \in [1,\infty]$. Then, the set of spectral dependence measures with respect to the $\|\cdot \|_p$-norm, $\setspectraldepposd{\| \cdot \|_p}$, satisfies
    \[
     \setspectraldepposd{\| \cdot \|_p}
    =
    \Big\{ \Psi \text{ prob.\, measure on } \simplexpos{d-1}(\| \cdot\|_p) \, \Big| \, \Exp_{\bm \Lambda \sim \Psi}[\Lambda_1] =  \dots = \Exp_{\bm \Lambda \sim \Psi}[\Lambda_d] \in [d^{-1},d^{-1/p}]\Big\},
    \]
    where we interpret $d^{-1/\infty}=d^0=1$.
\end{lemma}

\begin{proof}
    `$\supset$' follows from the construction in Lemma 2.2.2 in \cite{kulik2020heavy}. 
    For `$\subset$', we only need to show the constraints on the marginal expectations. By Lemma 2.1(iii) in \cite{boulin2026structured}, we have
    \[
    1 = L(\bm e_j) = \varsigma^{\| \cdot\|_p} \Exp_{\bm \Lambda_j \sim \Psi}[\Lambda_1]
    \]
    for any $j \in [d]$, with $\varsigma^{\| \cdot\|_p}=\varsigma_{\bm Y}^{\| \cdot\|_p}$ as in Lemma~\ref{lem:constant-in-front-of-L}.
    In particular, all marginal expectations are the same and equal to $1/\varsigma^{\| \cdot\|_p}$, and it remains to show that $\varsigma^{\| \cdot\|_p} \in [d^{1/p},d]$, with any value from the interval being attained. For that purpose, note that
    $
    \max\big\{ \| \bm x \|_p : \bm x \in[0,\infty)^d, \| \bm x \|_1=1 \big\} = 1,
    $
    which is attained at every unit vector. As a consequence, by Lemma~\ref{lem:constant-in-front-of-L},
    \[
    \varsigma^{\| \cdot\|_p} 
    = d\int_{\simplexpos{d-1}(\| \cdot\|_1)}  \| \bm \lambda\|_p \, \diff\Phi^{\|\cdot\|_1}(\bm \lambda)
    \le d,
    \]
    with equality if $\Phi^{\|\cdot\|_1}=d^{-1}\sum_{j \in [d]} \delta_{\bm e_j}$. Likewise, 
    $
    \min\big\{ \| \bm x \|_p :\bm x \in[0,\infty)^d, \| \bm x \|_1=1 \big\} = d^{1/p-1},
    $
    with equality at $\bm x = (1/d, \dots, 1/d)$. As a consequence,  by Lemma~\ref{lem:constant-in-front-of-L},
    \[
    \varsigma^{\| \cdot\|_p} 
    = d\int_{\simplexpos{d-1}(\| \cdot\|_1)}  \| \bm \lambda\|_p \, \diff\Phi^{\|\cdot\|_1}(\bm \lambda)
    \ge d^{1/p},
    \]
    with equality if $\Phi^{\|\cdot\|_1}=\delta_{(1/d, \dots, 1/d)}$. Finally, $ \varsigma^{\| \cdot\|_p} $ can attain any value in the interval $[d^{1/p}, d]$; simply consider a convex combination of the maximizing and minimizing measures.
\end{proof}

\begin{lemma}
\label{lem:geometry}
Let $\|\cdot\|_{1}$ and $\|\cdot\|_{2}$ be two norms on $\R^d$, and let $\simplex{d-1}_1$ and $\simplex{d-1}_2$ denote their associated unit spheres. Likewise, let $|\cdot|_{1}$ and $|\cdot|_{2}$ be two norms on $\R^K$, with associated spheres $\simplex{K-1}_{1}$ and $\simplex{K-1}_{2}$. Suppose that $\bX \in \R^d$ is of latent linear $K$-factor form with respect to the observation space norm $\|\cdot\|_2$ and the latent norm $|\cdot|_{2}$; that is, the spectral measure of $\bX$ with respect to $\|\cdot\|_2$ is given by
\[
\Phi_{2}
    = 
f_{\load}^{2 \to 2} \, \# \,
\phi_{2}^{\alpha, \load, 2}
\]
for some $K \in \N_{\ge 1}$, a full column rank matrix $\load \in \R^{d \times K}$, and a probability measure $\phi_{2}$ on $\simplex{K-1}_{2}$. Here, for $i,j\in\{1,2\}$, $f_{\load}^{i \to j}: \simplex{K-1}_i \to \simplex{d-1}_j$ is given by $\bz \mapsto \load \bz/\| \load \bz\|_j$ and $\phi_{i}^{\alpha, \load, j}$ is defined as in \eqref{eq:cZB}, but with $\Phifactor$ replaced by $\phi_{i}$ and with $\| \cdot \|$ replaced by $\| \cdot \|_j$; i.e.,
\[
\phi_{i}^{\alpha, \load, j}(\diff \bz) = c_{\alpha, \load, \phi_i,j}^{-1} \|\load \bz \|_j^\alpha \, \phi_{i}(\diff z)
\quad  \text{ with } \quad 
c_{\alpha, \load, \phi_i,j}
=
\int_{\simplex{K-1}_i}\| \load \bz \|_j^\alpha \, \phi_{i}(\diff z).
\]
Then the following holds:
\begin{compactenum}
\item \label{item:lemma_geometry_1} \textbf{Invariance under changes of the ambient-space norm.}  
Changing the norm on $\R^d$ from $\|\cdot\|_{2}$ to $\|\cdot\|_{1}$ leaves the latent representation unchanged: the spectral measure of $\bX$ with respect to $\|\cdot\|_1$ is given by
\[
    \Phi_{1}
    = 
    f_{\load}^{2 \to 1} \, \# \, {\phi}_{2}^{\alpha, \load, 1}.
\]
\item \label{item:lemma_geometry_2} \textbf{Reparameterization under changes of the latent-space norm.} 
$\bX$ is of latent linear $K$-factor form with respect to $|\cdot|_1$, with the same tail index $\alpha$, the same loading matrix $\load$, but a different factor spectral measure $\phi_1$. More specifically, 
if we define $\phi_{1}$ on $\simplex{K-1}_1$ by
\[
\phi_{1}(\diff \bz) = c^{-1} |\bz|_2^{-\alpha} (\varphi \,\#\, \phi_2)(\diff \bz) 
\quad \text{ with } \quad 
c = c_{\alpha,\varphi,\phi_{2}} = \int_{\simplex{K-1}_1}|\bz|_2^{-\alpha} (\varphi  \,\#\, \phi_{2})(\diff \bz), 
\]
where $\varphi : \simplex{K-1}_{2} \to \simplex{K-1}_{1}$ is defined by $\varphi(\bz) =\bz / |\bz|_{1}$, then
\[
\Phi_{2} = f_{\load}^{1 \to 2} \, \# \,\phi_1^{\alpha, \load, 2}.
\]
\end{compactenum}
\end{lemma}

\begin{proof}
(1) 
By definition of the spectral measure in \eqref{eq:definition-spectral-measure}  and Proposition 2.2.3 in \cite{kulik2020heavy} we have, for any Borel set $C \subset \simplex{d-1}_1$,
\begin{align*}
\Phi_1(C) 
&= 
\varsigma_1^{-1} \cdot \nu_{\bX}\Big( \Big\{ \bx \in \Eb_0^d : \| \bx \|_1 > 1, \frac{\bx}{\| \bx \|_1} \in C\Big\} \Big)
\\&=
\varsigma_1^{-1} \cdot \nu_{\bX}\Big( \Big\{ \bx \in \Eb_0^d : \| \bx \|_1 > 1, \frac{\bx}{\| \bx \|_2} \in h^{-1}(C)\Big\} \Big)
\\ &=
\varsigma_1^{-1}\varsigma_2
\int_{\simplex{d-1}_2} \int_0^\infty \indicator{\{ \|r\bm \lambda\|_1>1, \bm \lambda \in h^{-1}(C) \}} \, \alpha r^{-\alpha-1} \, \diff r \, \Phi_2(\diff \bm \lambda)
\\&=
\varsigma_1^{-1}\varsigma_2
\int_{\simplex{d-1}_2}  \indicator{ \{\bm \lambda \in h^{-1}(C) \}}\| \bm \lambda\|_1^\alpha \, \, \Phi_2(\diff \bm \lambda)
\\&=
\varsigma_1^{-1}\varsigma_2
\int_{\simplex{d-1}_2}  \indicator{ C}\Big(\frac{\bm \lambda}{\|\bm \lambda \|_1}\Big) \| \bm \lambda\|_1^\alpha \, \, \Phi_2(\diff \bm \lambda)
\end{align*}
where $h:\simplex{d-1}_2 \to \simplex{d-1}_1, \bm \lambda \mapsto \bm \lambda/\|\bm \lambda\|_1$ and where $\varsigma_j=\nu_\bX(\{ \bx \in \Eb_0^d: \| \bx \|_j >1 \})$.
By assumption on $\Phi_2$ and the change of variable formula, the right-hand side of the previous display can be written as
\begin{align*}
    \int_{\simplex{d-1}_2}  \indicator{C}\Big(\frac{\bm \lambda}{\|\bm \lambda \|_1}\Big) \| \bm \lambda\|_1^\alpha \, \, \Phi_2(\diff \bm \lambda)
    &=
    \int_{\simplex{d-1}_2}  \indicator{C}\Big(\frac{\bm \lambda}{\|\bm \lambda \|_1}\Big) \| \bm \lambda\|_1^\alpha \, \, \Phi_2(\diff \bm \lambda)
    \\&= 
    \int_{\simplex{K-1}_2}  \indicator{C}\Big(\frac{\load \bz}{\|\load \bz \|_1}\Big) \Big\| \frac{\load\bz}{\| \load \bz\|_2}\Big\|_1^\alpha \, \, 
    \phi_{2}^{\alpha, \load, 2}(\diff \bz)
    \\&= 
    c_{\alpha, \load, \phi_2,2}^{-1}
    \int_{\simplex{K-1}_2}  \indicator{C}\Big(\frac{\load \bz}{\|\load \bz \|_1}\Big) \| \load\bz \|_1^\alpha \, \, \phi_{2}(\diff \bz)
\end{align*}
Altogether, 
\begin{align*}
\Phi_1(C) 
=
\frac{\varsigma_2}{\varsigma_1c_{\alpha, \load, \phi_2,2}} \int_{\simplex{K-1}_2}  \indicator{C}\Big(\frac{\load \bz}{\|\load \bz \|_1}\Big) \| \load\bz \|_1^\alpha \, \, \phi_{2}(\diff \bz).
\end{align*}
Plugging in $C=\simplex{K-1}_1$, we obtain that $\varsigma_1=\varsigma_2$, which then implies
\begin{align*}
\Phi_1(C) 
=
\int_{\simplex{K-1}_2}  \indicator{C}\Big(\frac{\load \bz}{\|\load \bz \|_1}\Big)  \, \, \phi_{2}^{\alpha, \load, 1}(\diff \bz) = \phi_{2}^{\alpha, \load, 1}((f_\load^{2 \to 1})^{-1}(B)) = (f_\load^{2 \to 1} \, \#\, \phi_{2}^{\alpha, \load, 1})(B)
\end{align*}
as asserted.

(2) 
Note that $f_{\load}^{2 \to 2} = f_{\load}^{1 \to 2} \circ \varphi$, and that $(f \circ g) \# \mu = f \# (g \# \mu)$. Hence, 
\[
    \Phi_{2}
    =
    f_{\load}^{2 \to 2} \, \# \, \phi_2^{\alpha, \load, 2}
    = 
    f_{\load}^{1 \to 2} \, \# \, ( \varphi \,\#\, \phi_2^{\alpha, \load, 2}),
\]
so it remains to show that $\varphi \,\#\, \phi_2^{\alpha, \load, 2} = \phi_1^{\alpha, \load, 2}$. For an arbitrary Borel set $C \subset \simplex{K-1}_1$, we have, by the change of variable formula and the fact that $\varphi^{-1}(\bz) = \bz / |\bz |_2$,
\begin{align*}
    (\varphi \,\#\, \phi_2^{\alpha, \load, 2})(C)
    &=
    \frac1{c_{\alpha, \load, \phi_2,  2}} \int_{\varphi^{-1}(C)} \| B \bm z \|_2^\alpha \, \phi_2(\diff \bz)
    \\&=
    \frac1{c_{\alpha, \load, \phi_2,  2}} \int_{C} \| B \varphi^{-1}(\bm z) \|_2^\alpha \, (\varphi \,\#\,\phi_2)(\diff \bz) 
    \\&=
    \frac1{c_{\alpha, \load, \phi_2,  2}} \int_{C} \| B \bm z \|_2^\alpha |\bz|_2^{-\alpha} \, (\varphi \,\#\,\phi_2)(\diff \bz) 
    =
    \frac{c}{c_{\alpha, \load, \phi_2, 2}} \int_{C} \| B \bm z \|_2^\alpha \, \phi_1(\diff \bz) 
\end{align*}
If we plug in $C=\simplex{K-1}_1$, we obtain that 
\[
\frac{c}{c_{\alpha, \load, \phi_2, 2}} 
= 
\Big( \int_C \| B \bm z \|_2^\alpha \, \phi_1(\diff \bz) \Big)^{-1} = c_{\alpha, \load, \phi_1,2}^{-1}.
\]
Using this formula, the right-hand side of the penultimate display is $\phi_1^{\alpha, \load, 2}(C)$ as required.
\end{proof}

\section{More details on the case study}
\label{sec:details-on-case-study}

\subsection{Vertical Extrapolation of Wind Speeds to Hub Height}

As discussed in Section~\ref{sec:case-study}, the high-resolution HOSTRADA data set (spatial resolution 1\,km, temporal resolution 1\,hour) provides wind-speed measurements at a height of 10\,m only. Since modern wind turbines typically operate at a hub height of approximately 100\,m, an extrapolation procedure is required. We adopt a power-law model, which is standard in wind energy applications; see, for example, \cite{cai2021wind, jourdier2020evaluation, schallenberg2013methodological}. Specifically, the wind speed $v_h$ (in m\,s$^{-1}$) at height $h$ (in m above ground) is assumed to satisfy
\begin{align} \label{eq:power-law}
v_h = v_{10} \left( \frac{h}{10\,\mathrm{m}} \right)^{\shear},
\end{align}
where $v_{10}$ denotes the wind speed at the reference height of 10\,m and $\shear$ is the Hellmann (or friction) exponent. Depending on time of day, atmospheric stability, and surface characteristics, this exponent typically ranges between 0.1 and 0.6, with a commonly used benchmark value of 0.14 for neutral conditions over open terrain.

Solving \eqref{eq:power-law} for $\shear$ yields
\[
\shear = \frac{\ln(v_h) - \ln(v_{10})}{\ln(h) - \ln(10)}.
\]
We use this formula to estimate $\shear$ from the ERA5 data set, which provides hourly wind speeds at both 10\,m and 100\,m, albeit on a coarser spatial grid (approximately 27.5\,km resolution). We assume that $\shear$ varies with location and time of day. Specifically, we compute \textit{hour-of-day monthly means}: for each hour and each of the 15 ERA5 grid cells within the study region, $\shear$ is averaged over all days of the corresponding month, resulting in 24 representative values per month and location. 
The extrapolation formula \eqref{eq:power-law} is then applied to each observation in the HOSTRADA data set after spatial matching, where a HOSTRADA grid cell is assigned to the ERA5 grid cell containing its center.

\subsection{Bootstrap Confidence Intervals}

To quantify estimation uncertainty, we employ a bootstrap procedure to construct confidence intervals for both the empirical estimator $\hat p_\alpha^{\mathrm{emp}}(\bm w)$ from \eqref{eq:estimator-tail-parameter-empirical-benchmark} and the model-based estimator $\hat p_\alpha(\bm w) := \hat p(\bm x, \mathcal J_\alpha)$ from \eqref{eq:estimator-tail-parameter}. 
We generate $B=999$ bootstrap samples by resampling the observed data with replacement. For the empirical estimator, the procedure is straightforward: $\hat p_\alpha^{\mathrm{emp}}(\bm w)$ is recomputed for each bootstrap sample. The model-based estimator is treated analogously. For each bootstrap sample, we re-estimate the marginal quantities with $k(j)=k=135$, apply the tuning-parameter selection scheme described in Section~\ref{subsec:hyperparameter-and-model-validation} with $k'=k=135$, and use the resulting factor estimates $\hat A$ and $\hat K$ at the selected hyperparameter to compute $\hat p_\alpha(\bm w)$.

In the following, let $\hat p \in \{ \hat p_\alpha(\bm w), \hat p_\alpha^{\mathrm{emp}}(\bm w) \}$ denote a generic estimator of $p_\alpha(\bm w)$. Repeating the above procedure for all $B$ bootstrap samples yields bootstrap replicates $\hat p^{(1)}, \dots, \hat p^{(B)}$, from which we construct basic bootstrap confidence intervals \citep{DavisonHinkley1997}.  
Specifically, let $q_\beta^*(\hat p)$ denote the empirical $\beta$-quantile of $\hat p^{(1)}, \dots, \hat p^{(B)}$. The basic bootstrap confidence interval for $\hat p$ at level $1-\beta \in (0,1)$ is then given by $\mathrm{CI}_{1-\beta} = [\hat L_\beta, \hat U_\beta]$, where
\begin{align*}
    \hat L_\beta &= \max \big\{ 0, \, 2\hat p - q_{1-\beta/2}^*(\hat p) \big\}, 
    \qquad 
    \hat U_\beta = \min \big\{ 1, \, 2\hat p - q_{\beta/2}^*(\hat p) \big\}.
\end{align*}

\bibliographystyle{apalike}
\bibliography{pick}
\end{document}